\documentclass[12pt]{article}

\usepackage{placeins}
\usepackage{amsmath}
\usepackage{graphicx,psfrag,epsf}
\usepackage{enumerate}
\usepackage{natbib}
\usepackage[normalem]{ulem}
\usepackage{url} 
\usepackage[total={170mm,230mm},
 left=25mm,
 top=20mm]{geometry}
\usepackage{amssymb}
\usepackage[colorlinks=true, allcolors=blue]{hyperref}
\usepackage[table,xcdraw]{xcolor}
\usepackage{booktabs}
\usepackage{array}
\usepackage{bbm}
\newcolumntype{P}[1]{>{\centering\arraybackslash}p{#1}}
\newcolumntype{M}[1]{>{\centering\arraybackslash}m{#1}}
\newcolumntype{L}[1]{>{\raggedright\let\newline\\\arraybackslash\hspace{0pt}}m{#1}}
\newcolumntype{C}[1]{>{\centering\let\newline\\\arraybackslash\hspace{0pt}}m{#1}}
\newcolumntype{R}[1]{>{\raggedleft\let\newline\\\arraybackslash\hspace{0pt}}m{#1}}
\usepackage{makecell} 
\usepackage{siunitx}
\newcommand{\EST}[2]{\makecell{#1 (#2)}}         

\usepackage{amsthm}
\newtheorem{thm}{Theorem}
\newtheorem{lm}{Theorem}
\newtheorem{lemma}[lm]{Lemma}
\newtheorem{prop}[thm]{Proposition}

\usepackage{xpatch}
\xpatchcmd{\proof}{\itshape}{\normalfont\proofnamefont}{}{}
\newcommand{\proofnamefont}{\bfseries}

\usepackage{mathtools}

\usepackage{mathrsfs}

\usepackage{tabularx,array}
\newcolumntype{Y}{>{\centering\arraybackslash}X}
\newcommand{\blind}{1}

\definecolor{ch}{RGB}{30,136,229}
\definecolor{chbs}{RGB}{34,139,34}

\newcommand\extrafootertext[1]{%
    \bgroup
    \renewcommand\thefootnote{\fnsymbol{footnote}}%
    \renewcommand\thempfootnote{\fnsymbol{mpfootnote}}%
    \footnotetext[0]{#1}%
    \egroup
}

\newcommand\blfootnote[1]{%
  \begingroup
  \renewcommand\thefootnote{}\footnote{#1}%
  \addtocounter{footnote}{-1}%
  \endgroup
}

\newsavebox{\smlmat}
\savebox{\smlmat}{$\boldsymbol{\Sigma}=\left(\begin{smallmatrix}1&1.05\\1.05&2.25\end{smallmatrix}\right)$}

\begin{document}

\def\spacingset#1{\renewcommand{\baselinestretch}%
{#1}\small\normalsize} \spacingset{1}


\if1\blind
{
  \title{\bf A Bayesian multivariate extreme value mixture model}}
  \author{Chenglei Hu$^1$, Ben Swallow$^2$ and Daniela Castro-Camilo$^1$*\blfootnote{Corresponding author. Email address: Daniela.CastroCamilo@glasgow.ac.uk}}
  
  \footnotetext[1]{
\baselineskip=10pt School of Mathematics and Statistics, University of Glasgow, UK.}
\footnotetext[2]{ 
\baselineskip=10pt School of Mathematics and Statistics, University of St Andrews, UK.}
  \maketitle
 \fi


\if0\blind
{
  \bigskip
  \bigskip
  \bigskip
  \begin{center}
    {\LARGE\bf Title}
\end{center}
  \medskip
} \fi

\bigskip
\begin{abstract}
Accurate characterisation of joint risk in multivariate systems requires models that adequately capture both typical behaviour and rare extremes, yet existing approaches either sacrifice theoretical guarantees or entangle bulk and tail inference.
We introduce a Bayesian multivariate extreme value mixture framework that resolves this limitation by decoupling bulk and tail modelling on disjoint supports, thereby preserving the exact asymptotic properties of extreme value theory while enabling full-distribution inference.
The proposed model treats the threshold as a learnable parameter and embeds it within a unified probabilistic framework that delivers coherent uncertainty quantification for all components, including extremal dependence. 
This separation ensures that tail behaviour is governed solely by the multivariate generalised Pareto distribution, eliminating the instability induced by bulk-tail interactions in existing copula-based and mixture approaches.
To enable practical inference under strong parameter dependence, we develop a computational strategy combining reparametrisation with automated factor slice sampling, yielding efficient and robust posterior exploration.
Extensive simulation studies demonstrate accurate recovery of marginal and dependence structures across diverse regimes, including robustness under model misspecification. An application to UK temperature shows substantial gains in estimating moderate and high quantiles, with clear implications for environmental risk assessment.
\end{abstract}

\noindent%
\vfill

\newpage
\spacingset{1.45} 
\section{Introduction}
\label{sec: chp2_intro}
In many applications, risk depends on both extreme and non-extreme events. For example, environmental hazards such as flooding or heatwaves might arise either from isolated extreme events or from sustained moderate conditions.
Classical extreme value theory (EVT) provides asymptotically justified models for the tail of a distribution, but does not describe the bulk, and is therefore insufficient when the full range of the data is of interest.

In the univariate setting, this limitation has motivated numerous approaches that combine a generalised Pareto distribution (GPD) for threshold exceedances with a lighter-tailed distribution for the remainder of the data. 
Examples include the dynamically weighted mixtures of~\cite{frigessiDynamicMixtureModel2002} and the Bayesian threshold models of~\cite{behrensBayesianAnalysisExtreme2004}, as well as semi- and non-parametric extensions by~\cite{tancredi2006accounting,donascimentoSemiparametricBayesianApproach2012, macdonaldFlexibleExtremeValue2011} and \cite{huang2019estimating}. 
Other unified approaches, such as \citet{naveau2016modeling} and \citet{stein2021parametric}, reproduce GPD behaviour in both tails without the need for explicit thresholding. \citet{krock2022nonstationary} extended Stein’s model to accommodate nonstationary settings, while \citet{opitzINLAGoesExtreme2018} and \citet{castro2019spliced} employed discontinuous spliced models to represent the bulk and tail of environmental variables over space and time, leading to improved prediction and forecasting of extreme events.

Extending these ideas to the multivariate setting remains challenging, as both marginal distributions and dependence structures must be modelled. 
A common approach transforms margins to a standard scale and models dependence via copulas, often in a mixture form to represent both bulk and tail behaviour~\citep{aulbach2012multivariate, Andr__2024}. 
Alternative models combining marginal mixture approaches with copula-based dependence have also been proposed~\citep{leonelliSemiparametricBivariateModelling2020}. 
However, this type of approaches suffers from entanglement between the bulk and tail: the bulk copula can influence the tail region, and the tail copula remains sensitive to bulk data, even with optimised weighting. 
This dependency could become problematic in small samples, where limited tail observations lead to fragile tail inference.
Piecewise copula constructions, such as that of \citet{Aulbach_2012}, offer partial relief but cannot simultaneously model the margins.

To address these issues, we develop a Bayesian multivariate extreme mixture model that jointly captures bulk and tail behaviour while preserving the theoretical properties of EVT.
The key feature of our approach is a strict separation between bulk and tail components: observations below a multivariate threshold are modelled using a flexible parametric distribution, while exceedances are described by a multivariate generalised Pareto distribution (mGPD).
By defining the two components on disjoint regions of the sample space, the proposed model ensures that tail behaviour is governed entirely by the mGPD, independently of the bulk specification.

A second key contribution is treating the multivariate threshold as an unknown parameter.
Rather than fixing the threshold a priori as is customary, we estimate it jointly with all other model parameters, allowing the data to inform its value and enabling coherent uncertainty quantification.
This is particularly important in multivariate settings, where thresholds are more likely to exhibit multimodality and threshold selection is more complex.

From a computational perspective, inference in this model is challenging due to the strong dependence between the threshold and tail parameters.
To address this, we combine a reparametrisation strategy with an automated factor slice sampler (AFSS), which substantially improves mixing efficiency relative to standard Metropolis-Hastings approaches.
This allows reliable posterior inference even in the presence of strong parameter dependencies.
Through simulation studies, we demonstrate that the proposed model achieves accurate parameter recovery under a range of tail behaviours, maintains appropriate uncertainty quantification, and provides robust estimates of extremal dependence measures. We also show that, despite its asymptotic dependence structure, the model performs competitively when applied to asymptotically independent data in moderate sample sizes. An application to UK temperature data illustrates how the model improves estimation of high quantiles and captures joint tail risk more effectively than standard Gaussian approaches.


The remainder of this paper is organised as follows.
Section~\ref{sec:meth} introduces our model, including a detailed overview of the mGPD and our chosen representation.
Section~\ref{sec:inf} describes the prior specifications and posterior inference.
Simulation studies are presented in Section~\ref{sec:simu}, followed by an application to UK temperature data in Section~\ref{sec:app}.
Section~\ref{sec:discussion} discusses practical considerations for implementation, and Section~\ref{sec:conclusion} concludes.

\section{Multivariate extreme mixture model}
\label{sec:meth}
\subsection{Multivariate generalised Pareto distribution}\label{sec:mgpd}

Let $\boldsymbol{Y}_1,\dots\,\boldsymbol{Y}_n$ be $n$ independent and identical copies of the $d$-dimensional random vector $\boldsymbol{Y}$, which is in the max domain of attraction of a multivariate generalised extreme value distribution (mGEVD) $G$.
This means that there exists sequences of normalising vectors $\boldsymbol{\alpha}_n \in (0,\infty)^d$ and $\boldsymbol{\beta}_n \in \mathbb{R}^d$ such that 
\begin{align}\label{eq:mGEV}
\mathbb{P} \left( \boldsymbol{\alpha}_n \{\max _{1 \leq i \leq n} \boldsymbol{Y}_i\} + \boldsymbol{\beta}_n   \leq \boldsymbol{y} \right) = \mathbb{P}^n \left( \boldsymbol{\alpha}_n  \boldsymbol{Y} +\boldsymbol{\beta}_n  \leq \boldsymbol{y} \right)\rightarrow G(\boldsymbol{y}),\quad  n \rightarrow \infty,
\end{align}
where operations on vectors are componentwise. 
The mGEVD $G(\boldsymbol{y})$ has non-degenerating margins $G_j(x)$, $j=1,\ldots,d$ belonging to the univariate generalised extreme value family of distributions. Without loss of generality, we can assume that marginals are unit Fr\'{e}chet with cumulative distribution function $\exp(-y^{-1})$, for $y>0$. In that case, we can write the joint distribution function $G$ as
\begin{equation}\label{eq:mGEVdist}
	{G}(\boldsymbol{y})  = \exp\left\{-\int_{\mathcal{S}_d} \max\left(\frac{\omega_1}{y_1},\ldots,\frac{\omega_d}{y_d}\right)\text{d}Q(\boldsymbol{\omega})\right\},\quad \boldsymbol{y} > \boldsymbol{0},
\end{equation}
where $Q$ is called the spectral measure, an arbitrary positive finite measure over the unit simplex $\mathcal{S}_d = \{\boldsymbol{\omega}\in[0,1]^d:\sum_{j=1}^{d}\omega_j= 1\}$ that satisfies the constraint $$\int_{\mathcal{S}_d}w_j\text{d}Q(\boldsymbol{w}) = 1,\quad j=1,\ldots,d.$$
Let $\boldsymbol{Y} \nleqslant \boldsymbol{\beta}_n$ denote the event of at least one of the $\boldsymbol{Y}$ components exceeding the corresponding $\boldsymbol{\beta}_n$ component.
If the convergence in~\eqref{eq:mGEV} holds, the conditional random vector defined as
\begin{equation}\label{eq:excvec}
	\frac{\boldsymbol{Y}-\boldsymbol{\beta}_n}{\boldsymbol{\alpha}_n} \Big| \boldsymbol{Y} \nleqslant \boldsymbol{\beta}_n 
\end{equation}
converges in distribution to a $d$-dimensional random vector $\boldsymbol{X}$ with multivariate generalised Pareto distribution (mGPD) $H$ \citep{rootzenMultivariateGeneralizedPareto2006}. 
Based on the relationship in \eqref{eq:excvec}, $H$ can be expressed in terms of $G$ as
\begin{align}\label{eq:mgpd}
    H(\boldsymbol{x})=\frac{1}{\log G(\boldsymbol{0})} \log\frac{G(\boldsymbol{x}\wedge\boldsymbol{0} )}{G(\boldsymbol{x})},
\end{align}
where $\wedge$ denotes the componentwise minimum and $0<G(\boldsymbol{0})<1$ is assumed. 
The marginal distributions of $H(\boldsymbol{x})$ do not conform to univariate GPDs since $\boldsymbol{Y} \nleqslant \boldsymbol{\beta}_n$ does not necessarily imply that $Y_j>\beta_j$ for every $j$. 
However, conditioning on {$X_j>0$}, the conditional margin $H_j(x_j|x_j>0)$ are univariate GPD with scale $\sigma_j$, shape $\gamma_j$, and support depending on these parameters unless $\gamma_j=0$.
As we can see from~\eqref{eq:mGEVdist} and~\eqref{eq:mgpd}, both the mGEVD and the mGPD depend on the spectral measure $Q$ and therefore have no unique representation. 
Based on the work by~\cite{rootzen2018multivariate}, ~\cite{kiriliouk2019peaks} propose three equivalent representations of the mGPD density, denoted R, U, and T, for use in parametric modelling. These forms can be transformed into one another via a suitable change of variable. The choice of representation depends on the intended application (see \cite{kiriliouk2019peaks} for details), and does not affect our framework.
Here, we use the U representation with density
\begin{align}
    h_{\boldsymbol{U}}(\boldsymbol{x}) = \mathbbm{1} \{\max(\boldsymbol{x})>0\}\frac{\prod_{j=1}^d(\gamma_jx_j + \sigma_j)^{-1}}{\mathbb{E}[\exp(\max(\boldsymbol{U}))]}\int_0^{\infty}f_{\boldsymbol{U}}\left(\frac{1}{\boldsymbol{\gamma}}\log \left(\frac{\boldsymbol{\gamma}}{\boldsymbol{\sigma}}\boldsymbol{x}+\boldsymbol{1}\right) + \log t \right) \text{d}t,
    \label{eq: mGPD_U_obs}
\end{align}
where $f_{\boldsymbol{U}}$ is called a generator and is the density of a random vector $\boldsymbol{U}$ satisfying $\mathbb{E}[\exp(U_j)]<+\infty, j=1,\cdots,d$, and $[\log (\gamma_jx_j/{\sigma_j}+1)]/\gamma_j$ takes its limiting form ${x_j}/{\sigma_j}$ when $\gamma_j=0$. A standardised form of~\eqref{eq: mGPD_U_obs} with $\boldsymbol{\gamma}=0$ and $\boldsymbol{\sigma}=1$ can be obtained by applying the transformation 
\begin{align}
    \boldsymbol{Z} :=\frac{1}{\boldsymbol{\gamma}}\log\left(\frac{\boldsymbol{\gamma}}{\boldsymbol{\sigma}} \boldsymbol{X }+ 1\right),
    \label{eq:transformation}
\end{align}
where {the} random vector $\boldsymbol{X}$ has the density in \eqref{eq: mGPD_U_obs}, and operations are componentwise. As a result, the standardised density can be expressed as
\begin{align}
    h_{\boldsymbol{U}}(\boldsymbol{z}) =  \mathbbm{1} \{\max(\boldsymbol{z})>0\}\frac{1}{\mathbb{E}[\exp(\max(\boldsymbol{U}))]}\int_0^{\infty}f_{\boldsymbol{U}}(\boldsymbol{z}+ \log t) \text{d}t.
    \label{eq: mGPD_U_std}
\end{align}
In practice, the mGPD is usually first defined on the standardised scale given in \eqref{eq: mGPD_U_std} and then transformed back to the observation scale using \eqref{eq:transformation}.
Note that the U density in~\eqref{eq: mGPD_U_obs} or~\eqref{eq: mGPD_U_std} still lacks finite parametrisation since it depends on the density function $f_{\boldsymbol{U}}$, which determines the extremal dependence of the mGPD. 
Here, we assume 
$f_{\boldsymbol{U}}$ to have independent reverse 
exponential components, that is $f_{\boldsymbol{U}}(\boldsymbol{x})=\prod_{j=1}^d{a^{-1}_j}\exp({a_j}^{-1}{x_j})$, for $ x_j \in (-\infty,0), a_j>0$.
Other forms of the $f_{\boldsymbol{U}}$ and generators are briefly discussed in Section \ref{sec:discussion}.
In the reverse exponential case, \eqref{eq: mGPD_U_std} can be explicitly expressed as

\begin{align}
    h_{\boldsymbol{U}}(\boldsymbol{z})=\frac{\exp[-\max(\boldsymbol{z})(1+\sum_{j=1}^da_j^{-1})]}{\mathbb{E}[\exp(\max(\boldsymbol{U}))]}\frac{\prod_{i=1}^da_j^{-1}\exp(a_j^{-1}z_j)}{1+\sum_{i=1}^da_j^{-1}},
    \label{eq:explict_mgpd}
\end{align}
where $\mathbb{E}[\exp(\max(\boldsymbol{U}))] = \int_0^{\infty}1-\mathbb{P}\left(\exp(\max(\boldsymbol{U}))\leq t\right)\text{d}t = (\sum_{j=1}^da_j^{-1})/(1+\sum_{i=1}^da_j^{-1})$. 

\subsection{Multivariate bulk distribution}
\label{bulk_dist}
The mGPD in Section~\ref{sec:mgpd} is supported on the complement of the negative orthant. Introducing the threshold vector $\boldsymbol{u}$ shifts this support to the region {$A=\{\boldsymbol{x} \in \mathbb{R}^d : \boldsymbol{x}-\boldsymbol{u} \nleqslant \boldsymbol{0} \}$ }, that is, the set of points where at least one component exceeds its corresponding threshold.
On the complement $\mathbb{R}^d \setminus A$, we model observations with a multivariate bulk distribution $F_{\text{bulk}}$ on $\mathbb{R}^d$, having density $f_{\text{bulk}}$.
The choice of $F_{\text{bulk}}$ is flexible and should be guided by the characteristics of the system under study.
The only requirement is that both $f_{\text{bulk}}$ and $F_{\text{bulk}}$ can be evaluated exactly over their support.
This condition is met, for example, by a multivariate normal distribution or by copula-based constructions with specified marginals (see Section~\ref{sec:discussion}).
For illustration, in the remainder of this paper, we take $F_{\text{bulk}}$ to be multivariate normal, with density
\begin{align}
 f_{\text{bulk}}(\boldsymbol{x}| \boldsymbol{\mu}, \boldsymbol{\Sigma})
 =(2 \pi)^{-d / 2}(\operatorname{det} \boldsymbol{\Sigma})^{-\frac{1}{2}} \exp \left\{-\frac{1}{2}(\boldsymbol{x}-\boldsymbol{\mu})^{\prime} \boldsymbol{\Sigma}^{-1}(\boldsymbol{x}-\boldsymbol{\mu})\right\},\quad \boldsymbol{x} \leq \boldsymbol{u},
 \label{eq:bulk_density}
\end{align}
where $\boldsymbol{\mu}$ is the mean vector and $\boldsymbol{\Sigma}$ is the covariance matrix.

\subsection{Multivariate extreme mixture distribution}
Since the tail and bulk {define} a partition of the mixture model's support, combining \eqref{eq:bulk_density} and \eqref{eq:explict_mgpd}, the density function of our {multivariate extreme mixture} model can be written as 
\begin{equation}
f(\boldsymbol{x}|\boldsymbol{\mu},\boldsymbol{\Sigma},\boldsymbol{a},\boldsymbol{\sigma},\boldsymbol{\gamma},\boldsymbol{u}) = 
\begin{cases}
    f_{\text{bulk}}(\boldsymbol{x}|\boldsymbol{\mu},\boldsymbol{\Sigma}),& \text{if } \boldsymbol{x}\leqslant \boldsymbol{u}\\
    &\\
    [1-F_{\text{bulk}}(\boldsymbol{u}|\boldsymbol{\mu},\boldsymbol{\Sigma})]h_{\boldsymbol{U}}(\boldsymbol{x}-\boldsymbol{u}|\boldsymbol{a},\boldsymbol{\sigma},\boldsymbol{\gamma}),              & \text{otherwise}.
\end{cases}
\label{eq:pdf}
\end{equation}
The density in \eqref{eq:pdf} indicates that data are characterised by the multivariate normal distribution when all components are less than the threshold $\boldsymbol{u}$ and described by the U-representation of the mGPD with a reverse exponential generator when at least one component exceeds the threshold. 
Equation~\eqref{eq:pdf} can also be expressed in the standard mixture form by introducing the truncated bulk density
$f^*_{\text{bulk}}(\boldsymbol{x}|\boldsymbol{\mu},\boldsymbol{\Sigma})=\mathbbm{1}\{\boldsymbol{x} \leq \boldsymbol{u}\}f_{\text{bulk}}(\boldsymbol{x}|\boldsymbol{\mu},\boldsymbol{\Sigma})/F_{\text{bulk}}(\boldsymbol{u}|\boldsymbol{\mu},\boldsymbol{\Sigma})$.
In this notation, our extreme mixture model becomes
\begin{align*}
f(\boldsymbol{x}|\boldsymbol{\mu},\boldsymbol{\Sigma},\boldsymbol{a},\boldsymbol{\sigma},\boldsymbol{\gamma},\boldsymbol{u}) &= \pi f^*_{\text{bulk}}(\boldsymbol{x}|\boldsymbol{\mu},\boldsymbol{\Sigma})+ (1-\pi )h_{\boldsymbol{U}}(\boldsymbol{x}-\boldsymbol{u}|\boldsymbol{a},\boldsymbol{\sigma},\boldsymbol{\gamma}),
\end{align*}
where the mixture probability is $\pi =F_{\text{bulk}}(\boldsymbol{u}|\boldsymbol{\mu},\boldsymbol{\Sigma})$.
This form is particularly convenient for sampling from \eqref{eq:pdf} when $f^*_{\text{bulk}}$ can be sampled directly.

\savebox{\smlmat}{$\boldsymbol{\Sigma}=\left(\begin{smallmatrix}1&1.05\\1.05&2.25\end{smallmatrix}\right)$}

\begin{figure}[!htbp]
\begin{center}
\includegraphics[width=3in]{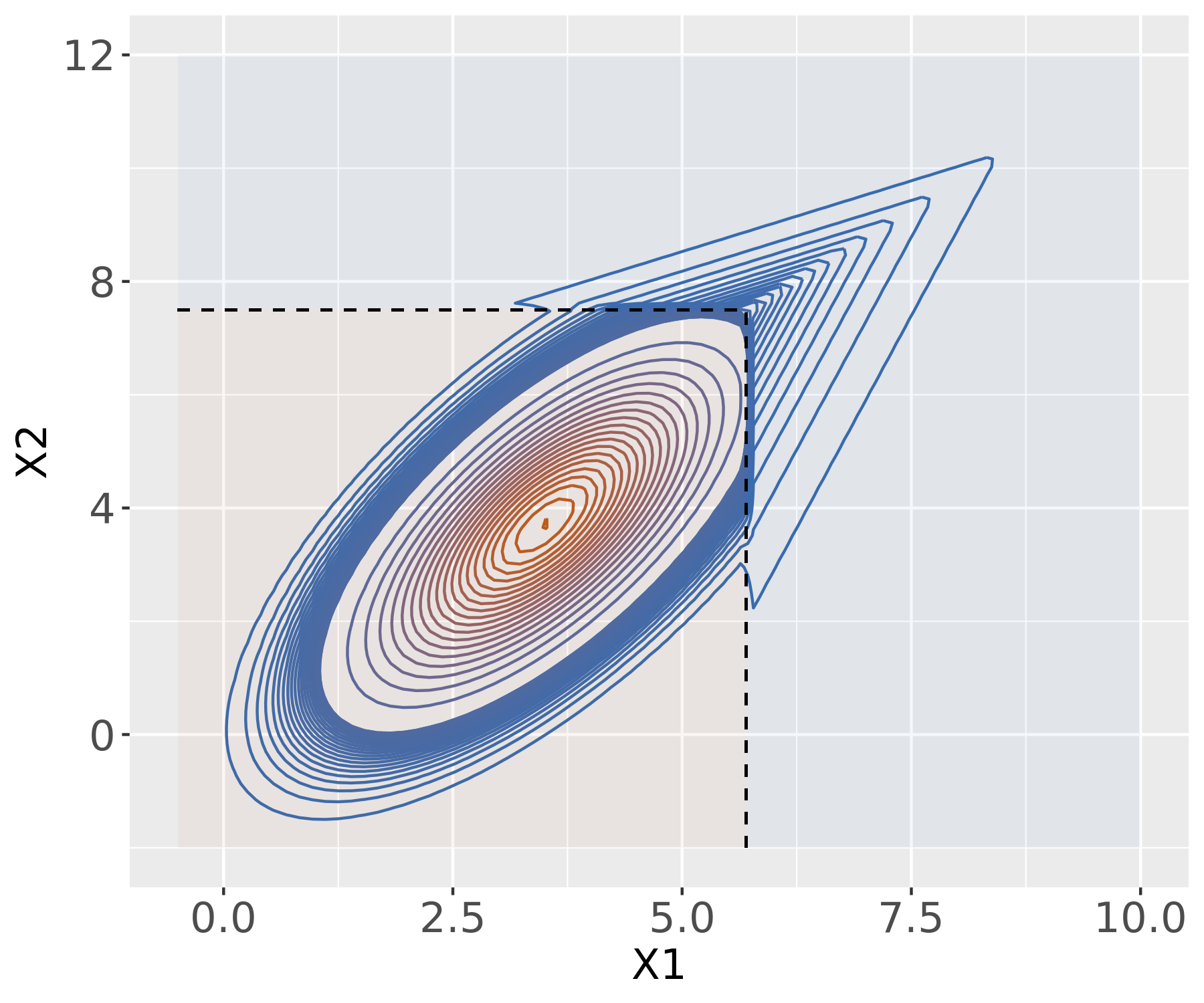}
\includegraphics[width=3in]{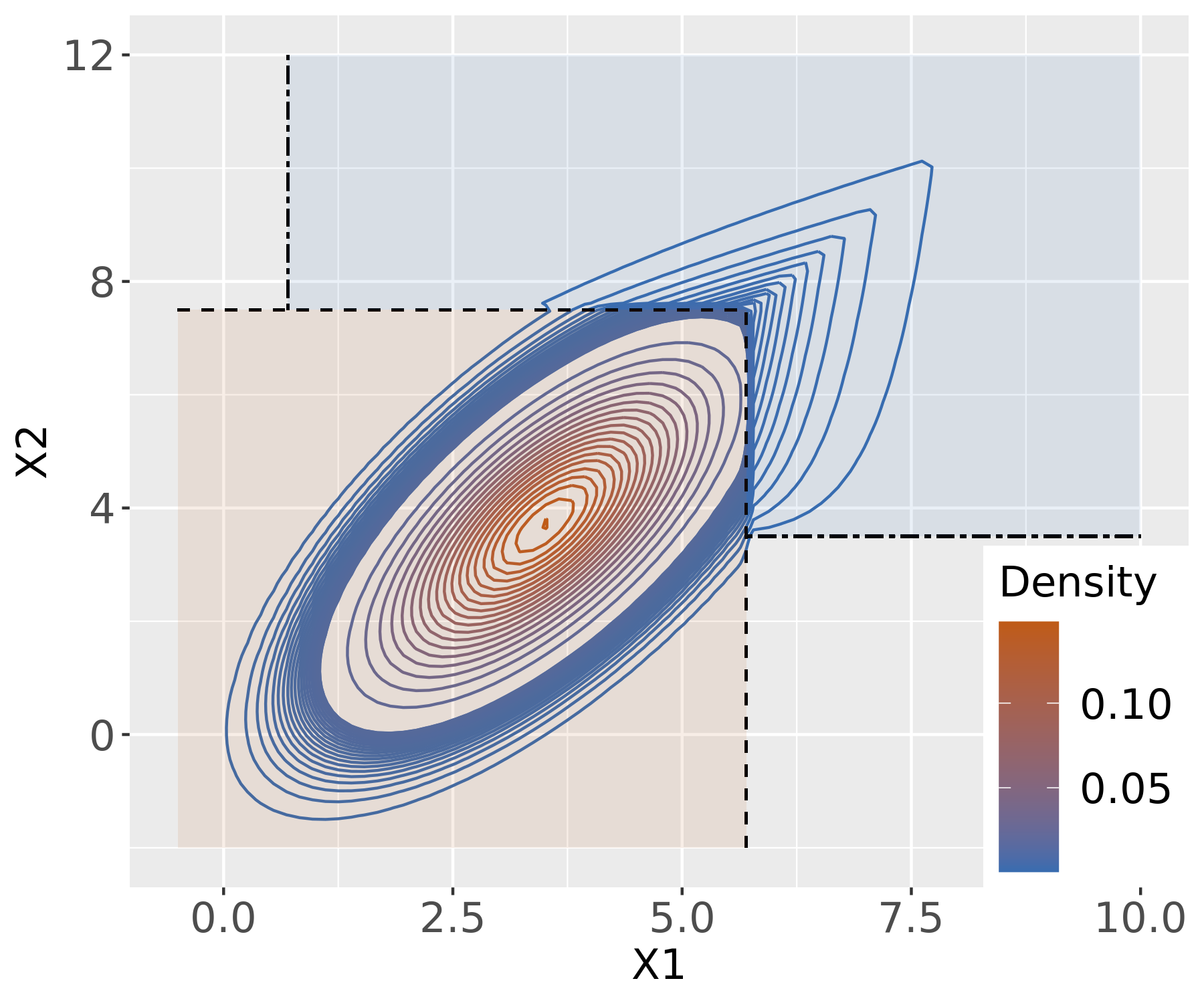}
\end{center}
\caption[\footnotesize{Two-dimensional density contour plots derived from \eqref{eq:pdf}.}]{\footnotesize{Two-dimensional density contour plots derived from \eqref{eq:pdf}. Dashed lines indicate threshold values, whereas the double-dashed lines in the right plot highlight the lower endpoints of the bivariate GPD. Shaded areas represent the support of each mixture component. For the left plot, the bivariate GPD parameters are specified as $\boldsymbol{\sigma}=\boldsymbol{1}$, $\boldsymbol{\gamma}=\boldsymbol{0}$, and $\boldsymbol{a}=(1,2)$. Conversely, the right plot employs parameters $\boldsymbol{\sigma}=(1,1.2)$, $\boldsymbol{\gamma}=(0.2,0.3)$, and $\boldsymbol{a}=(1,2)$. Both plots share common parameters for the bulk and threshold, represented by $\boldsymbol{\mu}=(3.5,3.7)$, ~\usebox{\smlmat}, and $\boldsymbol{u}=(5.7,7.5)$. }
}
\label{fig:plot_BEMM}
\end{figure}

Figure \ref{fig:plot_BEMM} shows a two-dimensional representation of the density in~\eqref{eq:pdf}. The dashed lines at $X_1 = u_1$ and $X_2 = u_2$ mark the boundary between the bulk and tail regions, as well as the discontinuities in the density.
These discontinuities are visible as abrupt changes in the contour lines when crossing the threshold.
Although the sharp change in density at the threshold may appear artificial, its practical effect is minimal for small sample sizes. Figure~\ref{fig:plot_BEMM} also helps to interpret the marginal densities. For example, in the left-hand panel, the marginal density $f_{1}(x_1)$ behaves as follows. For $X_1 > u_1 = 5.7$, the tail region is modelled entirely by the bivariate GPD, so the conditional marginal $f_1(x_1 \mid x_1 > u_1)$ reduces to a univariate GPD.
For $X_1 < u_1$, the marginal density $f_1(x_1 \mid x_1 < u_1)$ is obtained by integrating $x_2$ out of a mixture of the bivariate GPD (triangular contours) and the bivariate normal distribution (elliptical contours).
In this sense, our framework extends the univariate approach of \citet{donascimentoSemiparametricBayesianApproach2012} to the multivariate case, with both methods combining a mixture of parametric bulk distributions with a univariate GPD for the tail in the marginal analysis.

As shown in~\eqref{eq:pdf}, the tail behaviour of our model is entirely determined by the mGPD.
Consequently, our multivariate extreme mixture model inherits all theoretical properties of the mGPD, including threshold stability, GPD conditional margins, and sum-stability under shape constraints \citep{kiriliouk2019peaks}. 
In particular, our model lies in the max-domain of attraction of an mGEVD, a result that follows directly from Theorem 2.2 of~\cite{rootzenMultivariateGeneralizedPareto2006}. 
A formal statement and proof are provided in Section \ref{sec:appendix-proofs} of the Supplementary Materials. 
As a consequence, our model is always asymptotically dependent; that is, the coefficient
\begin{equation}
    \chi :=\lim_{r\rightarrow1^-} \chi(r), \qquad \chi(r)=\frac{\mathbb{P}\{\bigcap_{j=1}^d \{X_j>F_j^{-1}(r)\} \}}{1-r} 
    \label{eq:chi},
\end{equation}
is strictly positive.
This implies that the model may introduce bias when the underlying data exhibit asymptotic independence (i.e., $\chi = 0$ as $r \rightarrow 1^-$).
In practice, however, the distinction between weak asymptotic dependence and complete asymptotic independence is often ambiguous, particularly in low-dimensional and small-sample settings.
Our simulation results demonstrate that the model can still provide reasonable estimates even when the data are asymptotically independent.

\section{Bayesian inference}
\label{sec:inf}
Although maximum likelihood inference could, in principle, be performed using methods such as the EM algorithm for mixture models, the Bayesian framework provides a natural means of quantifying all sources of uncertainty in our model, particularly for the threshold vector.
Here, we detail the prior specification for all model parameters: $\boldsymbol{\mu}$ and $\boldsymbol{\Sigma}$ for the bulk; the threshold vector $\boldsymbol{u}$; and $\boldsymbol{a}$, $\boldsymbol{\gamma}$, and $\boldsymbol{\sigma}$ for the tail.

\subsection{Priors for the parameters in the bulk}
A standard approach for assigning priors in the multivariate normal setting is to use a multivariate normal prior for the mean vector and an inverse-Wishart prior for the covariance matrix.
While this choice offers conjugacy, it can be overly restrictive for the covariance matrix, as noted by \citet{sunObjectiveBayesianAnalysis2007}.
To address this limitation, \citet{barnardModelingCovarianceMatrices2000} proposed decomposing the covariance matrix as $\boldsymbol{\Sigma} = \text{diag}(\boldsymbol{S})\boldsymbol{C} \text{diag}(\boldsymbol{S})$, where $\boldsymbol{C}$ is the $d$-dimensional correlation matrix and $\boldsymbol{S}$ is the $d \times 1$ vector of standard deviations.
Priors are then assigned separately to these components.
In our simulations and case study (Sections~\ref{sec:simu} and~\ref{sec:app}), we assign independent $\text{Half-Normal}( h_{s_j})$ priors to each $s_j$ in $\boldsymbol{S}$, choosing $ h_{s_j}$ sufficiently large to be noninformative; specifically, $ h_{s_j}$ is set to 50 times the scale of the data.
For the correlation matrix, we adopt the Lewandowski–Kurowicka–Joe (LKJ) prior \citep{lewandowskiGeneratingRandomCorrelation2009}, given by 
$$\pi (\boldsymbol{C} )\propto (\det \boldsymbol{C})^{\delta -1},$$
where the parameter $\delta$ controls the strength of correlation.
When $\delta = 1$, the prior is uniform over all valid $d$-dimensional correlation matrices.
$\delta > 1$ favors stronger correlations (larger diagonal elements), while $\delta < 1$ favors weaker correlations.
For computational considerations, the LKJ prior can also be defined on the upper-triangular Cholesky factor $\boldsymbol{L}$ of $\boldsymbol{C}$  using a change of variable.
For the mean vector $\boldsymbol{\mu}$, we adopt independent normal priors $\mu_j \sim N(m_j, t_j^2)$ for $j = 1, \dots, d$.
In summary, the joint prior for the bulk parameters is
\begin{equation}
\pi(\boldsymbol{\mu}, \boldsymbol{\Sigma}) = \pi(\boldsymbol{\mu}, \boldsymbol{S}, \boldsymbol{L})
\propto \left| \boldsymbol{L}^\mathsf{T} \boldsymbol{L} \right|^{\delta}
\prod_{j=1}^d \varphi\left(\frac{\mu_j - m_j}{t_j}\right)  \exp\left\{-\frac{s_j^2}{2 h_{s_j}^2}\right\}\mathbbm{1}(s_j > 0)
\label{eq:bulk_prior}
\end{equation}
where $\varphi$ denotes the standard normal density.

\subsection{Prior for the threshold}
The choice of threshold plays a critical role in the threshold exceedance framework. A high threshold results in too few exceedances for reliable inference, leading to increased variance in parameter estimates. Conversely, a low threshold may violate the asymptotic assumptions of the mGPD, thereby introducing bias.
We seek to retain thresholds in the upper quantiles while allowing sufficient flexibility for the data to determine their optimal values and to quantify the associated uncertainty. To this end, we parametrise the threshold $u_j$ using a marginal reference location at the $\tau$-quantile, denoted by $q_{\tau,j}$, together with a positive offset. Specifically,
\begin{equation}
u_j = q_{\tau,j} + o_j, \quad o_j \stackrel{\mathrm{i.i.d.}}{\sim} \text{Half-Normal}(h_{o}),
\label{eq: prior_threshold}
\end{equation}
where $h_{o}$ is the scale parameter of the half-normal distribution, controlling how far the threshold may deviate from the reference level. Although $q_{\tau,j}$ is data-dependent, our substantive prior is placed on the uncertainty of exceedances and is itself data-independent.

The prior in \eqref{eq: prior_threshold} can be viewed as a location-aware extension of the truncated normal prior for thresholds \citep{behrensBayesianAnalysisExtreme2004, macdonaldFlexibleExtremeValue2011, donascimentoSemiparametricBayesianApproach2012}, enabling its application to datasets that may include negative values. An alternative approach is to impose a discrete prior on upper order statistics \citep{de2001predictive, behrensBayesianAnalysisExtreme2004}. However, in our framework, this choice is more prone to poor mixing during Markov chain Monte Carlo (MCMC) inference than in univariate settings \citep{behrensBayesianAnalysisExtreme2004, macdonaldFlexibleExtremeValue2011, donascimentoSemiparametricBayesianApproach2012}. 
We therefore adopt a continuous prior on the threshold to preserve flexibility and facilitate efficient MCMC sampling.

Assuming independence across marginal priors, the joint prior distribution of the threshold vector $\boldsymbol{u}$ is given by
\begin{equation*}
    \pi(\boldsymbol{u}) \propto\prod_{j=1}^d\exp\left\{-\frac{(u_j-q_{\tau,j})^2}{2h_{o}^2}\right\}\mathbbm{1}(u_j > q_{\tau,j})
\end{equation*}



\subsection{Priors for the parameters in the tail}
The marginal parameters of the mGPD are $\boldsymbol{\gamma} = (\gamma_1, \dots, \gamma_d)$ and $\boldsymbol{\sigma} = (\sigma_1, \dots, \sigma_d)$, whose interpretations mirror those in the univariate GPD, since each marginal conditional distribution $H_j(x_j \mid x_j > 0)$ is univariate GPD. 
Various priors for GPD parameters have been proposed, including quantile-based priors \citep{colesBayesianMethodsExtreme1996}, Jeffreys priors \citep{castellanosDefaultBayesianProcedure2007}, and penalised complexity priors \citep{opitzINLAGoesExtreme2018}.
In our framework, we impose weakly informative priors with hard constraints to preserve essential tail properties, for example, ensuring finite marginal excess expectations by restricting $\gamma_j < 1$, while avoiding further parameter preference.
Specifically, we assume $\gamma_j \stackrel{\mathrm{i.i.d.}}{\sim} \text{Uniform}(l_{\gamma}, r_{\gamma})$ and $\sigma_j \stackrel{\mathrm{i.i.d.}}{\sim} \text{Half-Norm}(h_{\sigma})$, where $l_{(\cdot)}$ and $r_{(\cdot)}$ denote lower and upper bounds, respectively.
The components of the tail dependence parameter $\boldsymbol{a} = (a_1, \dots, a_d)$, defined in~\eqref{eq:explict_mgpd}, are independently assigned prior $\text{Half-Norm}(h_{a})$.
The joint prior for all tail parameters is therefore
\begin{align}
\pi(\boldsymbol{a}, \boldsymbol{\gamma}, \boldsymbol{\sigma}) \propto
\prod_{j=1}^{d} \exp\left\{-\frac{a_j^2}{2h_a^2}-\frac{\sigma_j^2}{2h^2_{
\sigma
}}\right\}\mathbbm{1}( a_j > 0)  \cdot \mathbbm{1}( \sigma_j > 0)\cdot \mathbbm{1}(l_{\gamma} < \gamma_j < r_{\gamma}).
\label{eq:tail_prior}
\end{align}

\subsection{Posterior inference}
\label{sec:posterior_inference}
Let \( \boldsymbol{\theta}_b = (\boldsymbol{\mu},\boldsymbol{S},{\boldsymbol{L})} \), and \(  \boldsymbol{\theta}_t = (\boldsymbol{a},\boldsymbol{\gamma}, \boldsymbol{\sigma}) \). 
Let \( \boldsymbol{X} = (\boldsymbol{x}_1, \cdots, \boldsymbol{x}_n)^T\) be an \( n \times d \) sample matrix. 
Define the mutually exclusive sets $B$ and $T$ to correspond to the rows of \( \boldsymbol{X} \) classified as bulk and tail data, respectively, with $ B,T \subset  \{1,\ldots,n\}$.
The log posterior density of the multivariate extreme mixture model in~\eqref{eq:pdf} is given by: 
\begin{alignat}{2}
\begin{split}
\log \pi( \boldsymbol{\theta}_b,  \boldsymbol{u},  \boldsymbol{\theta}_t|{\boldsymbol{X}}) &\propto \sum_{i=1}^n  \log  \left\{  f_{\text{bulk}}(\boldsymbol{x_i}| \boldsymbol{\theta}_b)\cdot \mathbbm{1}(\boldsymbol{x_i} \leq  \boldsymbol{u}) \right. \\ 
&\quad + \left. [1-F_{\text{bulk}}( \boldsymbol{u}| \boldsymbol{\theta}_b)]f_{\text{tail}}(\boldsymbol{x_i}- \boldsymbol{u}| \boldsymbol{\theta}_t){\cdot \mathbbm{1}(\boldsymbol{x_i} \nleqslant    \boldsymbol{u})}   \right\} \\ 
&\quad + \log \pi( \boldsymbol{\theta}_b) + \log \pi( \boldsymbol{u}) + \log \pi( \boldsymbol{\theta}_t) \\
& {\propto} \sum_{i \in B}  \log f_{\text{bulk}}(\boldsymbol{x_i}| \boldsymbol{\theta}_b) +  \sum_{i \in T} \log f_{\text{tail}}(\boldsymbol{x_i}- \boldsymbol{u}| \boldsymbol{\theta}_t)  + \left | T \right |\log[1-F_{\text{bulk}}( \boldsymbol{u}| \boldsymbol{\theta}_b)] \\ &\quad + \log \pi( \boldsymbol{\theta}_b) + \log \pi( \boldsymbol{u}) + \log \pi( \boldsymbol{\theta}_t),
\end{split}
\label{eq: posterior}
\end{alignat}
where {$ \left | T \right | := \left | T( \boldsymbol{u}) \right |$} denotes the number of tail observations, which depends on the threshold. 

We perform inference using Markov chain Monte Carlo (MCMC) sampling from~\eqref{eq: posterior}.
As an initial step, we employ a multivariate random-walk Metropolis–Hastings (MH) algorithm with independent normal proposals for each parameter.
Proposal scaling is adaptively tuned to achieve the theoretical optimal acceptance rate \citep{robertsOptimalScalingVarious2001}.
However, this approach can result in suboptimal mixing for the threshold parameters $\boldsymbol{u}$ and certain components of $\boldsymbol{\theta}_t$, due to strong posterior dependence among parameters.
Specifically, the threshold stability property of the mGPD \citep{kiriliouk2019peaks} implies that increasing the threshold of an mGPD with shape parameter $\boldsymbol{\gamma}$ and scale parameter $\boldsymbol{\sigma}$ by $\boldsymbol{w}$ yields a new mGPD with identical marginal shape and tail dependence structure, but with an updated scale parameter $\boldsymbol{\sigma} + \boldsymbol{\gamma w}$. This induces strong dependence between $\boldsymbol{u}$ and $\boldsymbol{\sigma}$ in the posterior distribution. 

To address this issue, we define $\tilde{\boldsymbol{\sigma}} = \boldsymbol{\sigma} - \boldsymbol{\gamma u}$, which is invariant for sufficiently high $\boldsymbol{u}$, and reparametrise $(\boldsymbol{\sigma}, \boldsymbol{\gamma}, \boldsymbol{u})$ as $(\tilde{\boldsymbol{\sigma}}, \boldsymbol{\gamma}, \boldsymbol{u})$. 
Retaining the original priors for $\boldsymbol{\sigma}$, $\boldsymbol{\gamma}$, and $\boldsymbol{u}$, the induced prior for $\tilde{\boldsymbol{\sigma}}$ is
\begin{equation*}
    \pi(\tilde{\boldsymbol{\sigma}}\mid \boldsymbol{u},\boldsymbol{\gamma})=\pi_{\boldsymbol{\sigma}}(\tilde{\boldsymbol{\sigma}}+\boldsymbol{\gamma u})\cdot |1| \propto \prod_{j=1}^d\exp\left\{-\frac{(\tilde{\sigma}_j+\gamma_j u_j)^2}{2 h_{\sigma}^2}\right\}\mathbbm{1}(\tilde{\sigma}_j>-\gamma_j u_j)
\end{equation*}
To further improve sampling efficiency, we block update the threshold $\boldsymbol{u}$ and tail parameters $(\boldsymbol{a},\tilde{\boldsymbol{\sigma}},\boldsymbol{\gamma})$ using the automated factor slice sampler (AFSS; \citealp{tibbitsAutomatedFactorSlice2014}). The AFSS mitigates posterior dependence by performing univariate slice sampling along the eigenvectors of the covariance matrix obtained from a quadratic approximation of the target distribution.
Empirically, this approach substantially improves computational efficiency, often by more than an order of magnitude when measured by the ratio of effective sample size to runtime. 
The remaining bulk parameters are updated using MH with a Gaussian random-walk proposal. 
All AFSS-based MCMC routines are implemented using the \texttt{NIMBLE} package in R \citep{perry-nimble2017}. For consistency with the parametrisation in \eqref{eq:pdf}, we report results using the original parametrisation rather than the reparametrised form.

\section{Simulations}
\label{sec:simu}
We conduct simulation studies to evaluate (i) the ability of our model to recover bulk, tail, and threshold parameters under diverse tail behaviours, (ii) robustness to model misspecification, particularly when the true process is asymptotically independent, and (iii) its capacity to capture key dependence measures $\chi$, $\bar\chi$ and Kendall's $\tau$. 
{These aspects are essential both for accurate extrapolation to unseen observations and for a proper characterisation of the bulk data.}
We focus on the bivariate case to enable extensive Monte Carlo evaluation.

\subsection{Well-specified scenarios (Scenario 1)}
In the first experiment, we assess parameter recovery under correctly specified models with three different marginal tail types, detailed in Table~\ref{tab:scenario_gamma}.
The bulk parameters and thresholds are fixed across scenarios, as shown in the ``True Value'' (TV) columns of Table~\ref{tab:tail_thres_Scen1.1-1.3}. 
Each dataset contains $n=2000$ observations, with approximately 5\% classified as tail data.
\begin{table}[ht]
\centering
\caption{\footnotesize{Tail shape parameters for the three well-specified simulation scenarios. 
Each scenario differs in the heaviness of the marginal tails, as determined by the shape parameters $(\gamma_1, \gamma_2)$.}}
\vspace{0.5cm}
\label{tab:scenario_gamma}
\begin{tabular}{ccccl}
\hline
Scenario & $\gamma_1$ & $\gamma_2$ && Description \\
\hline
1.1 & 0.3   &  0.1   && Both margins heavy-tailed \\
1.2 & 0.2   & -0.2   && One heavy-, one light-tailed \\
1.3 & -0.1  & -0.3   && Both margins light-tailed \\
\hline
\end{tabular}
\end{table}


We adopt weakly informative priors. For the bulk prior in \eqref{eq:bulk_prior}, we set $\delta = 1.3$, $m_1 = m_2 = 0$, and $h_{s_1} = h_{s_2}= 50$. The threshold vector is constrained to lie above the 0.8-quantile of each margin, with scale parameter  $h_o = 10$. 
For the tail parameters in \eqref{eq:tail_prior}, we specify $h_a = h_{\sigma} = 50$ to reflect limited prior information on the scale in both the margins and $f_{\boldsymbol{U}}$. The shape parameter bounds are set to $l_{\gamma} = -1$ and $r_{\gamma} = 1$ to ensure a finite mean for the GPD of the marginal threshold excesses.
\begin{table}
\centering
\small
\setlength{\tabcolsep}{6pt}
\renewcommand{\arraystretch}{1.1}
\caption[\footnotesize{Average posterior means with 95\% credible interval (CI) lengths in brackets, and coverage rates (CR) for all parameters across three well-specified scenarios.}]{\footnotesize{Average posterior means with 95\% credible interval (CI) lengths in brackets, and coverage rates (CR) for all parameters across three well-specified scenarios. Each scenario uses $n = 2000$ observations with approximately 5\% classified as tail data. The ``True Value'' (TV) column gives the parameter values used for data generation. Results are averaged over 1000 simulation replicates. CI length is the average length of the marginal 95\% credible interval across replicates. CR is the proportion of replicates in which the true value falls within the CI. Bulk, tail, and threshold parameters are grouped separately for clarity. }}
\vspace{0.5cm}
\begin{tabular}{c c C{6em} c c C{6em} c c C{6em} c}
\toprule
& \multicolumn{3}{c}{\textbf{Scenario 1.1}}
& \multicolumn{3}{c}{\textbf{Scenario 1.2}}
& \multicolumn{3}{c}{\textbf{Scenario 1.3}} \\
\cmidrule(lr){2-4}\cmidrule(lr){5-7}\cmidrule(lr){8-10}
& TV & EST & CR & TV & EST & CR & TV & EST & CR \\
\midrule
\multicolumn{5}{l}{\textbf{Tail parameters}} \\
$a_1$     & 0.5 & \EST{0.55}{0.52} & 0.95 & 0.5 & \EST{0.54}{0.47} & 0.95 & 0.5 & \EST{0.53}{0.45} & 0.94 \\
$a_2$     & 1.2 & \EST{1.23}{0.86} & 0.95 & 1.2 & \EST{1.25}{0.83} & 0.96 & 1.2 & \EST{1.25}{0.82} & 0.96 \\
$\sigma_1$& 0.5 & \EST{0.52}{0.30} & 0.95 & 0.5 & \EST{0.51}{0.29} & 0.94 & 0.5 & \EST{0.51}{0.28} & 0.95 \\
$\sigma_2$& 1.2 & \EST{1.25}{0.66} & 0.95 & 1.2 & \EST{1.23}{0.64} & 0.96 & 1.2 & \EST{1.23}{0.64} & 0.95 \\
$\gamma_1$& 0.3 & \EST{0.31}{0.40} & 0.96 & 0.2 & \EST{0.22}{0.39} & 0.95 & -0.1& \EST{-0.09}{0.35} & 0.95 \\
$\gamma_2$& 0.1 & \EST{0.09}{0.25} & 0.95 & -0.2& \EST{-0.20}{0.30} & 0.96 & -0.3& \EST{-0.30}{0.32} & 0.95 \\
\midrule
\multicolumn{5}{l}{\textbf{Threshold parameters}} \\
$u_1$     & 5.5 & \EST{5.51}{0.21} & 0.95 & 5.5 & \EST{5.50}{0.14} & 0.94 & 5.5 & \EST{5.50}{0.12} & 0.95 \\
$u_2$     & 6.7 & \EST{6.70}{0.08} & 0.94 & 6.7 & \EST{6.70}{0.08} & 0.95 & 6.7 & \EST{6.70}{0.09} & 0.95 \\
\midrule
\multicolumn{5}{l}{\textbf{Bulk parameters}} \\
$L[1,2]$   & 0.70   & \EST{0.70}{0.05}  & 0.96 & 0.70   & \EST{0.70}{0.05}  & 0.96 & 0.70   & \EST{0.70}{0.05}  & 0.96 \\
$L[2,2]$   & 0.71 & \EST{0.71}{0.05} & 0.96 & 0.71 & \EST{0.71}{0.05} & 0.96 & 0.71 & \EST{0.71}{0.05} & 0.96 \\
$\mu_1$    & 3.5   & \EST{3.50}{0.09}  & 0.97 & 3.5   & \EST{3.50}{0.09}  & 0.98 & 3.5   & \EST{3.50}{0.09}  & 0.98 \\
$\mu_2$    & 4.0     & \EST{4.00}{0.13}    & 0.98 & 4.0     & \EST{4.00}{0.13}    & 0.98 & 4.0     & \EST{4.00}{0.13}    & 0.98 \\
$s_1$      & 1.0     & \EST{1.00}{0.07}    & 0.97 & 1.0     & \EST{1.00}{0.07}    & 0.96 & 1.0     & \EST{1.00}{0.07}    & 0.96 \\
$s_2$      & 1.5   & \EST{1.50}{0.10}   & 0.97 & 1.5   & \EST{1.50}{0.10}   & 0.98 & 1.5   & \EST{1.50}{0.10}   & 0.98 \\
\bottomrule
\bottomrule
\end{tabular}
\label{tab:tail_thres_Scen1.1-1.3}
\end{table}

Posterior samples are obtained using the AFSS sampler with three parallel chains, each of length 30,000, a burn-in of 20,000, and thinning by a factor of 10. Each scenario is replicated 1,000 times to assess estimation variability.
Table \ref{tab:tail_thres_Scen1.1-1.3} reports average posterior means, the average lengths of 95\% credible intervals (in brackets), and coverage rates (CR). Across all parameters and scenarios, CR values are close to 0.95, indicating accurate recovery of both marginal and dependence parameters.

\subsection{Misspecification scenario (Scenario 2)}

We now assess robustness under departures from the model’s asymptotic dependence assumption.
This is relevant because our model is always asymptotically dependent ($\chi>0$), and applying it to asymptotically independent data can, in principle, induce bias.

We generate data from a bivariate normal distribution, which is asymptotically independent ($\chi=0$; \citealt{sibuyaBivariateExtremeStatistics1960}). Bulk parameters match those in Table \ref{tab:tail_thres_Scen1.1-1.3}.
We compare our bivariate extreme mixture model (BEMM) to the bivariate mixture copula model (BMCM) of \citeauthor{Andr__2024}, using their best-performing configuration for Gaussian data: a Student-t copula for the bulk and an inverted Gumbel copula for the tail.
For each method, we assess overall dependence through the Kendall’s rank coefficient, defined as $$\tau = 2\mathbb{P}((X_1-X_1^\prime)(X_2-X_2^\prime)>0)$$ for a random pair $(X_1,X_2)$ and its independent replicate $(X_1^\prime,X_2^\prime)$, with marginal distribution $F_1$ and $F_2$.
For tail dependence, we use two measures. The first is $\chi$, defined in \eqref{eq:chi}, which distinguishes between asymptotic dependence ($\chi > 0$) and asymptotic independence ($\chi = 0$). For our model, based on the mGPD in \eqref{eq:explict_mgpd}, $\chi$ has the closed-form expression (see Section \ref{sec:appendix-proofs} of the Supplementary Materials for derivation):
\begin{align} 
\chi = 1-\left(\frac{1+a_{(1)}}{1+a_{(2)}} \right )^{1+a_{(2)}^{-1}}\frac{a_{(2)}}{a_{(1)}}\frac{a_1a_2}{a_1a_2+a_1+a_2},
\label{eq:theoretical_chi}
\end{align} where $a_{(1)}=\min \{a_1,a_2\}$, $a_{(2)}=\max \{a_1, a_2\}$. 
This expression further shows that our model is asymptotically dependent, as $\chi>0$ for all $a_1,a_2>0$.

While $\chi(r)$ is effective for distinguishing between asymptotic dependence and independence, it provides limited information on the strength of asymptotic independence. To address this, we also use $\bar{\chi}$ \citep{colesDependenceMeasuresExtreme1999}, defined as
\begin{align*}
    \bar{\chi} = \lim_{r\rightarrow 1^-} \bar{\chi}(r), \qquad 
    \bar{\chi}(r)=\frac{2\log \mathbb{P}(F_1(X_1)>r)}{\log \mathbb{P}(F_1(X_1)>r), F_2(X_2)>r)}-1, 
\end{align*}
which differentiates among asymptotically independent cases. 
$\bar{\chi}$ takes values in $[-1, 1)$, increases with the strength of extremal dependence, and approaches 1 for asymptotically dependent distributions.

For comparability with BMCM’s frequentist inference, we obtain point estimates for the BEMM by averaging the empirical dependence metrics ($\chi$, $\bar{\chi}$, and Kendall’s $\tau$) computed from 3,000 replicated datasets, each containing 2,000 observations. These replicated datasets, denoted $\boldsymbol{y}_{\text{rep}}$, are drawn from the posterior predictive distribution:
\begin{align}
p(\boldsymbol{y}_{\text{rep}}|\boldsymbol{X})=\int p(\boldsymbol{y}|\boldsymbol{\theta})p(\boldsymbol{\theta}|\boldsymbol{X})\mathrm{d} \boldsymbol{\theta},
    \label{eq:dist_rep}
\end{align}
where $p(\boldsymbol{y}|\boldsymbol{\theta})$ represents the density function in \eqref{eq:pdf}, and $p(\boldsymbol{\theta}|\boldsymbol{X})$ is the posterior distribution.

\begin{figure}
\begin{center}
\includegraphics[width=3in]{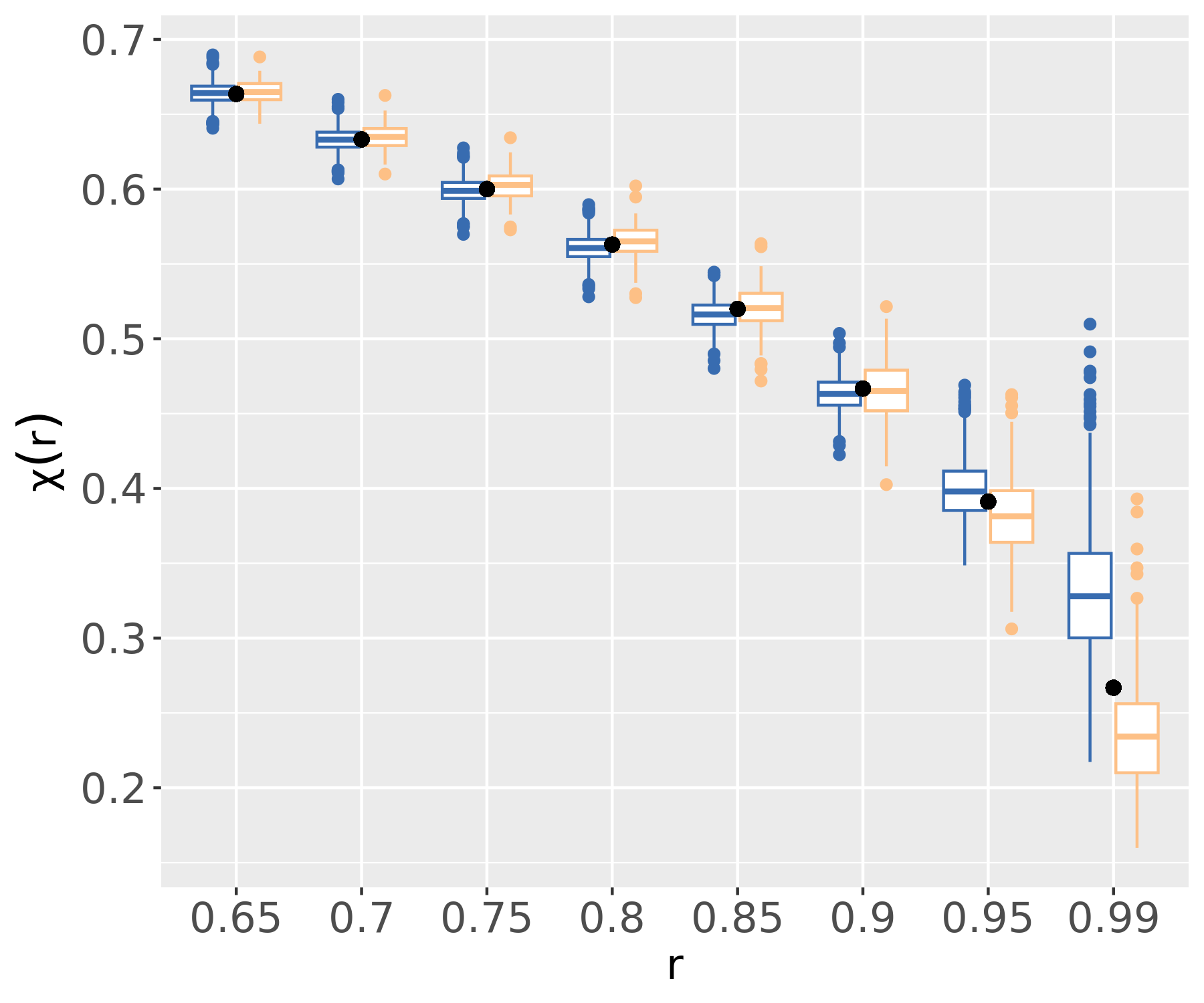}
\includegraphics[width=3in]{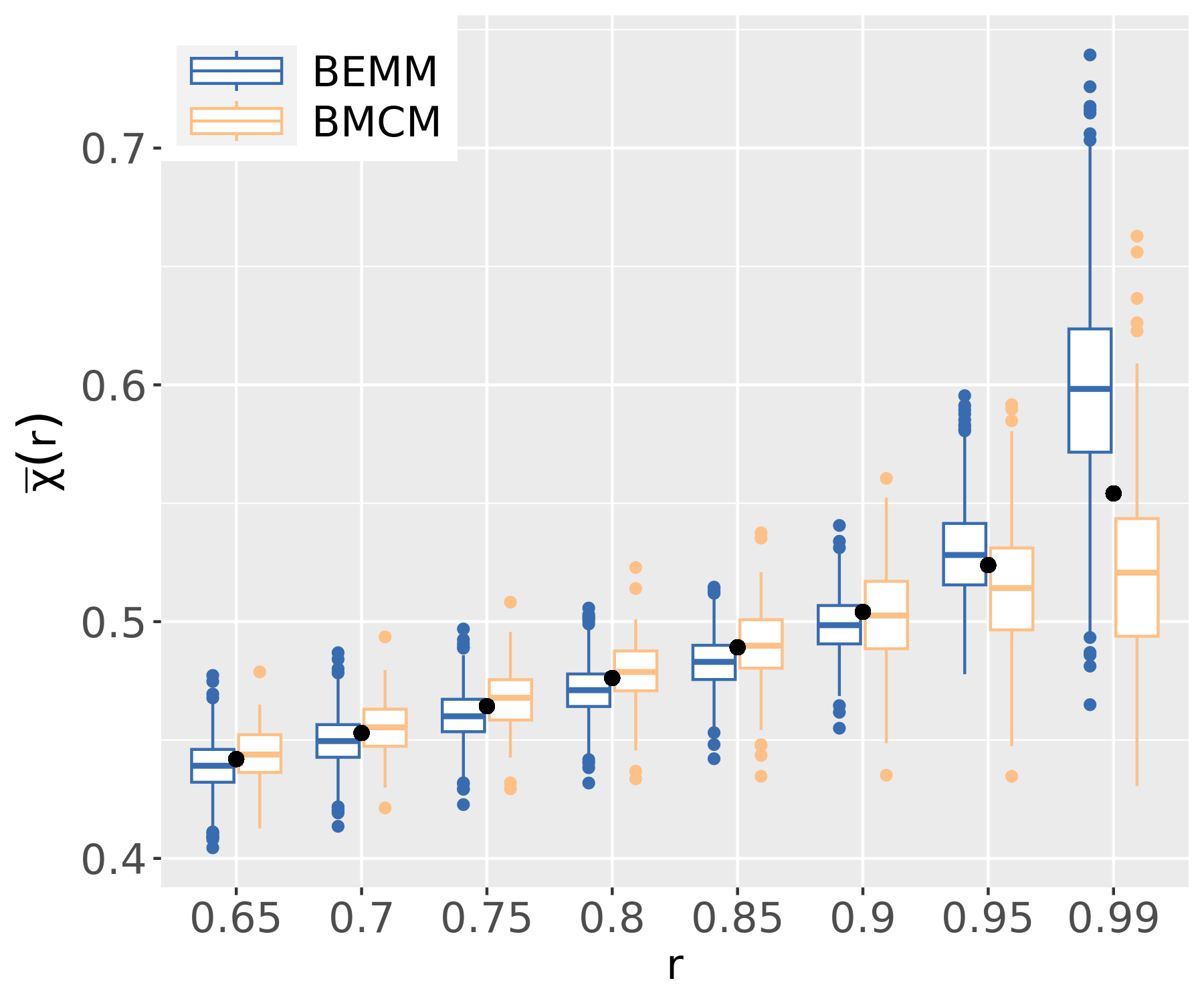}
\includegraphics[width=3in]{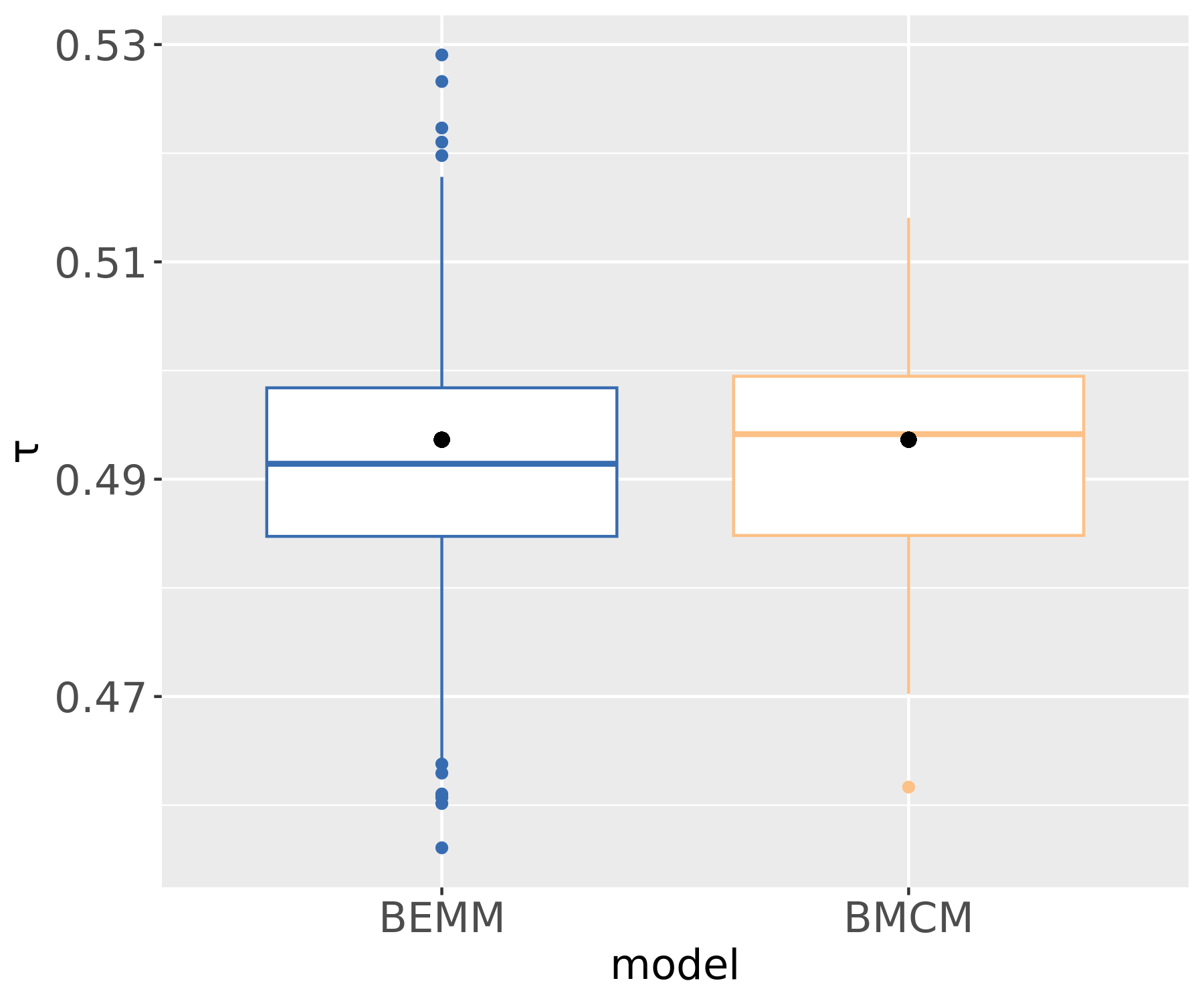}
\end{center}
\caption[\footnotesize{Boxplots comparing extremal dependence metrics ($\chi(r)$, $\bar{\chi}(r)$), and Kendall’s $\tau$ for the bivariate extreme mixture model (BEMM, blue) and the bivariate mixture copula model (BMCM, orange) under Scenario 2 (Gaussian data, asymptotic independence).}]{\footnotesize{Boxplots comparing extremal dependence metrics ($\chi(r)$, $\bar{\chi}(r)$), and Kendall’s $\tau$ for the bivariate extreme mixture model (BEMM, blue) and the bivariate mixture copula model (BMCM, orange) under Scenario 2 (Gaussian data, asymptotic independence).
Yellow boxes are based on estimates from 100 BMCM experiments. Blue boxes correspond to empirical metrics from 1,000 BEMM experiments, where each experiment reports the point estimates of the three metrics as averages over 3,000 replicated datasets. Black points denote the true values of the dependence metrics.
\label{fig:simulation2}}}
\end{figure}

Figure \ref{fig:simulation2} presents boxplots of $\chi(r)$, $\bar{\chi}(r)$, and Kendall’s $\tau$ from 100 BMCM experiments and 1,000 BEMM experiments. The number of BMCM experiments matches those in \citeauthor{Andr__2024}, while the number of BEMM experiments is increased to account for the additional parameter uncertainty arising from prior specification. For tail dependence, both models perform similarly for $r \leq 0.95$. At $r = 0.99$, the BEMM slightly overestimates, and the BMCM slightly underestimates these quantities; in both cases, true values (black dots) lie within the interquartile range. Kendall’s $\tau$ estimates are comparable across models, reflecting its lower sensitivity to tail misspecification.

Overall, both simulation studies show that the BEMM delivers accurate parameter estimation under correct specification, maintains nominal credible-interval coverage, and yields reasonable dependence estimates even under asymptotic independence. These results suggest that the model’s asymptotic dependence property is not unduly problematic in moderate-sample, low-dimensional settings, where the empirical distinction between weak dependence and independence is inherently blurred.

\section{Application}
\label{sec:app}
The frequency of heatwaves in the United Kingdom has increased markedly in recent years, with the 2022 event setting a record temperature of $40.3^{\circ}$C and causing substantial societal impacts, including an estimated 3,000 excess deaths. 
To study temperature behaviour during such extremes and compare patterns across regions, we analyse daily maximum temperatures from two stations: Bishop’s Lane (Ringmer, East Sussex) and Model Farm (Shirburn, Oxfordshire). 
We aim to quantify joint tail risk while accounting for uncertainty in both thresholds and dependence parameters. The data, obtained from the Centre for Environmental Data Analysis\footnote{\href{https://catalogue.ceda.ac.uk/uuid/dbd451271eb04662beade68da43546e1}{https://catalogue.ceda.ac.uk/uuid/dbd451271eb04662beade68da43546e1}}, cover 2016-2021, yielding 1,945 daily maximum temperature observations after removing records with missing values.


\subsection{Pre-processing and model specification}
Daily temperatures exhibit strong seasonality and short-term dependence, which we model using sinusoidal terms for the annual cycle and a first-order autoregressive term for temporal correlation.
Specifically, for the air temperature $Y_{t,j}$ on day $t$ at site $j \in \{1, 2\}$, we assume
\begin{align}
Y_{t,j} = \beta_{0,j} + \beta_{1,j}\sin\left(\frac{2\pi}{365}t\right) + \beta_{2,j}\cos\left(\frac{2\pi}{365}t\right) + \beta_{3,j}Y_{t-1,j} + \varepsilon_{j},
\label{eq: temp_lr}
\end{align}
where $\beta_{i,j}$ $(i = 0, 1, 2, 3)$ are regression coefficients and $\varepsilon_{j}$ is a noise term. Stationarity of the residuals is verified using autocorrelation plots (provided in Section~\ref{sec:appendix_statchecks} of the Supporting Materials).
The joint negative residual $\boldsymbol{E} = -(\varepsilon_{1},\varepsilon_{2})$, as shown in Figure \ref{fig:UK_Temp_Resi}, presents heavy right tails and strong extremal dependence.
To characterise these joint residuals, we fit our BEMM model to $\boldsymbol{E}$ using the priors from the simulation study, with the additional constraint $\gamma_j + 1/a_j \geq 0$ $(i = 1, 2)$ to ensure finite marginal expectations. 
Finite marginal expectations enable the use of proper scoring rules for model evaluation, as discussed at the end of this section.
The derivation of these constraints is provided in Section~\ref{sec:appendix-proofs} of the Supplementary Materials.
We run three parallel chains of 20,000 iterations using the hybrid MH and AFSS sampler as in the simulations, and discard the first 10,000 as burn-in and thinning by a factor of 10.
Convergence is assessed using trace plots and the Gelman–Rubin diagnostic (reported in Section~\ref{sec:appendix_mcmcchecks} of the Supplementary Materials).
\begin{figure}[!htbp]
\begin{center}
\includegraphics[width=3in]{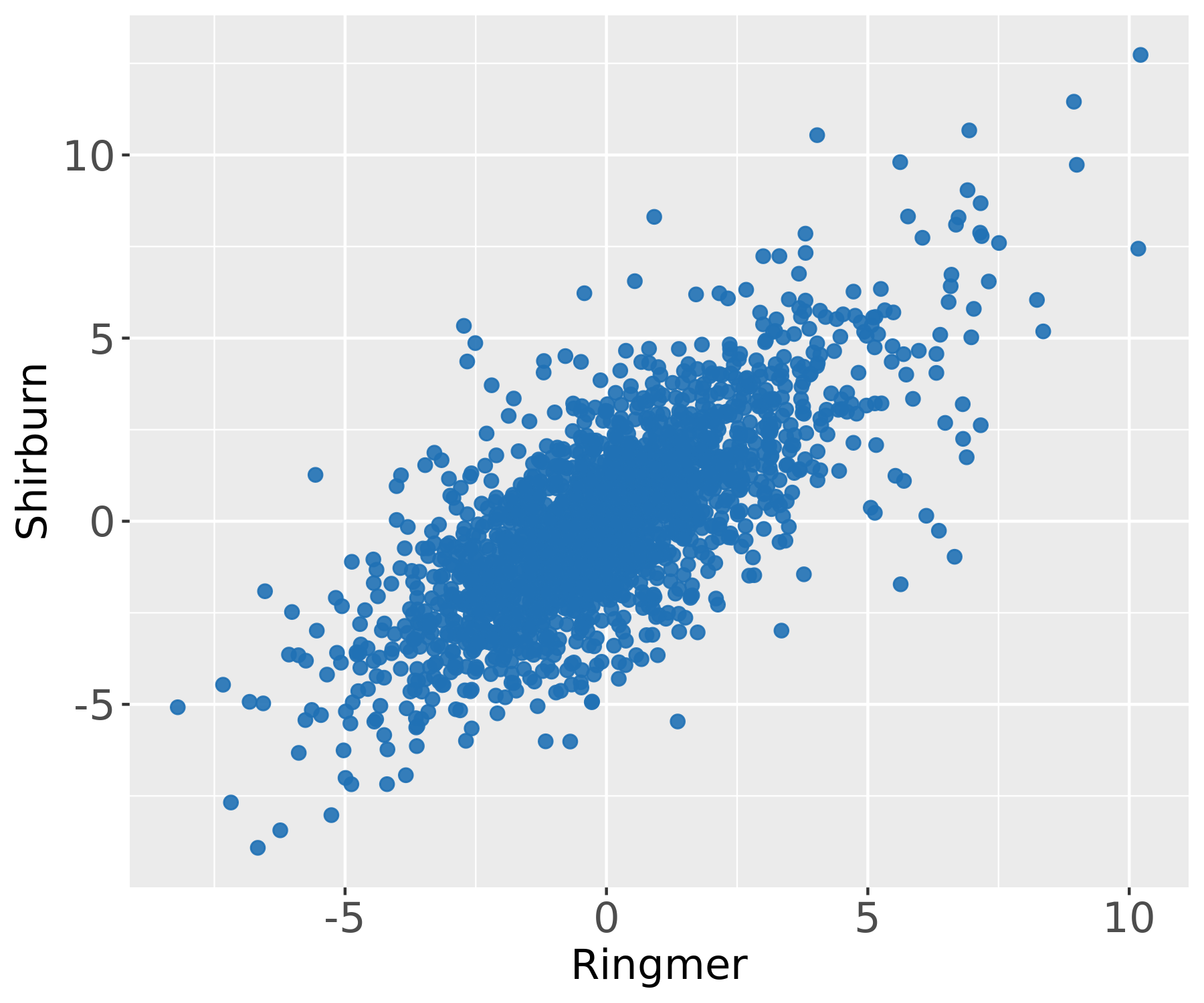}
\end{center}
\caption{\footnotesize{Scatterplot of the negative residuals from model \eqref{eq: temp_lr}.}
\label{fig:UK_Temp_Resi}}
\end{figure}

\subsection{Results}
\paragraph{Thresholds.} Figure \ref{fig: hist of posterior} presents histograms of the posterior samples for all parameters.
Notably, posterior distributions for $u_1$ and $u_2$ are multimodal.
For $u_1$, the dominant mode occurs near the 92nd percentile, with smaller peaks at the 89th and 94th percentiles. 
For $u_2$, the histogram is left-skewed, with clusters around the 93rd, 94th, and 95th percentiles.
To investigate this behaviour, we examine the joint posterior density $\pi(\boldsymbol{u} \mid \boldsymbol{X})$ together with a binned expected log-likelihood surface
$A(\boldsymbol{u})=\mathbb{E}[\log p( \boldsymbol{X}\mid\boldsymbol{\theta}_t,\boldsymbol{\theta}_b,\boldsymbol{u} )]$, computed by averaging $\log p( \boldsymbol{X}\mid\boldsymbol{\theta}_t,\boldsymbol{\theta}_b,\boldsymbol{u} )$ over posterior draws whose $\boldsymbol{u}$ values fall within each bin, as shown in Figure \ref{fig: posterior_density}.
The joint posterior exhibits three well-separated modes, and $A(\boldsymbol{u})$ reveals near-equal values across the modal regions. 
This indicates that the multimodality reflects genuine ambiguity in threshold selection supported by the likelihood under the UK temperature data, rather than a sampling artefact.

\begin{figure}
\begin{center}
\includegraphics[width=1.5in]{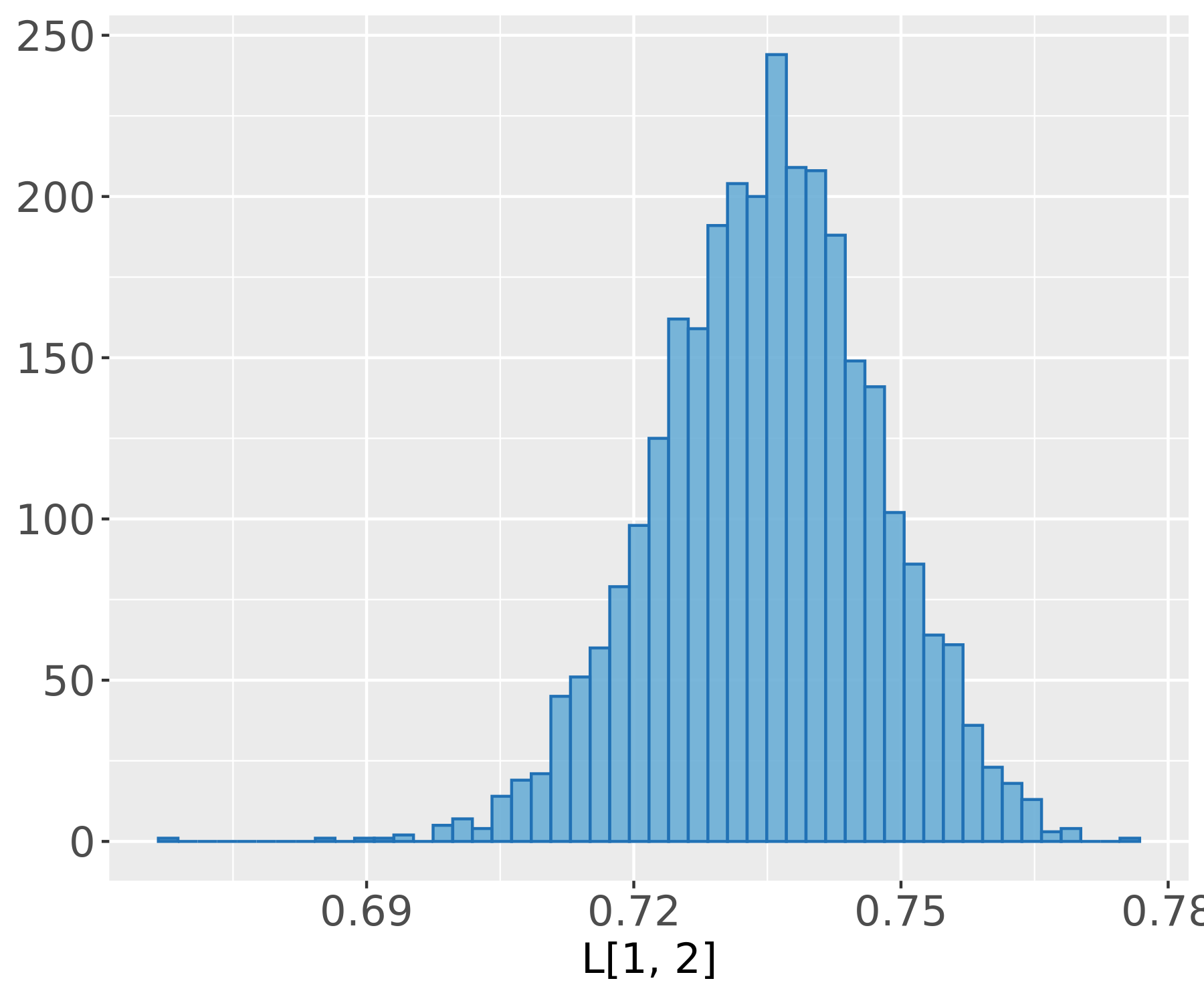}
\includegraphics[width=1.5in]{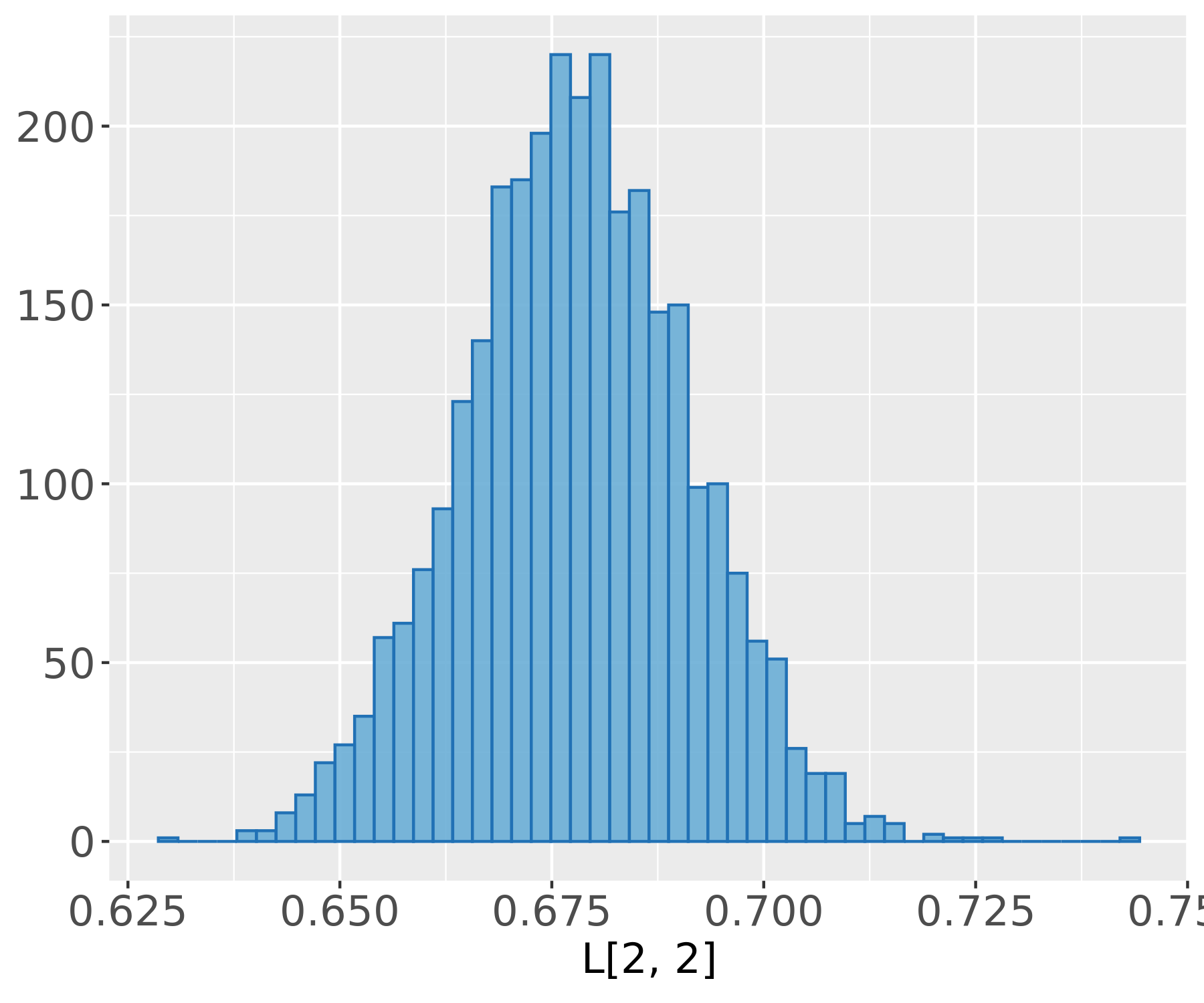}
\includegraphics[width=1.5in]{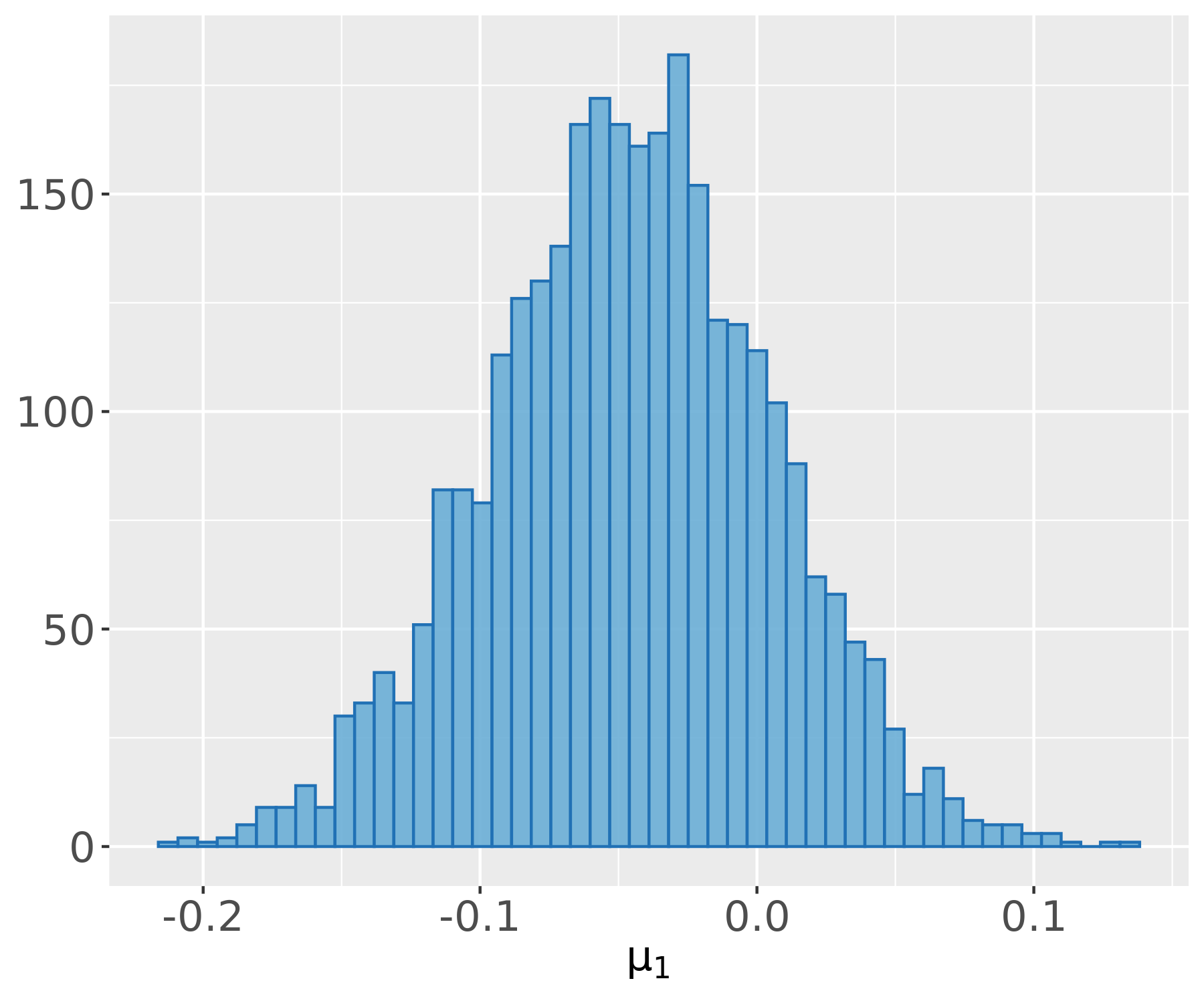}
\includegraphics[width=1.5in]{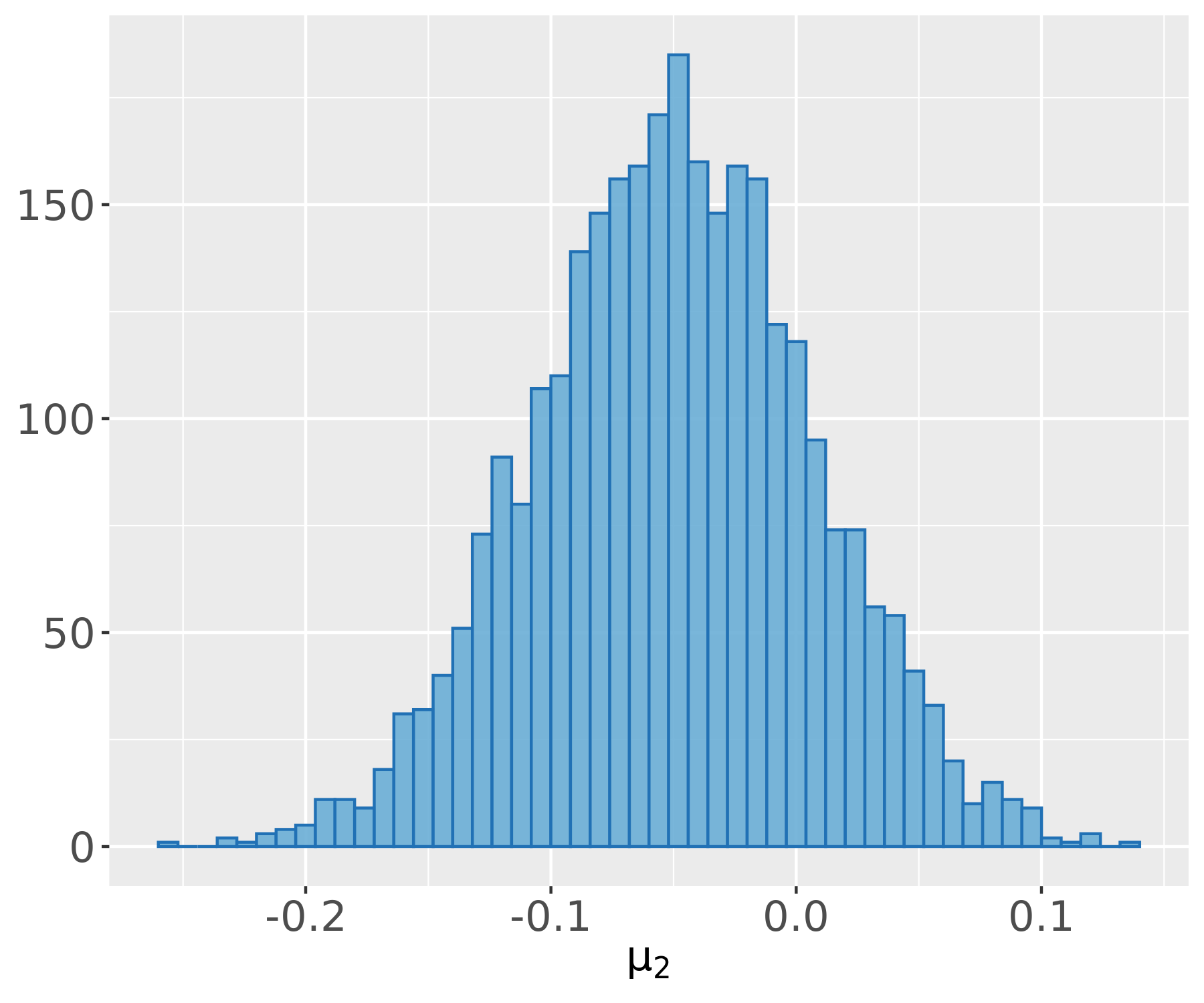}
\includegraphics[width=1.5in]{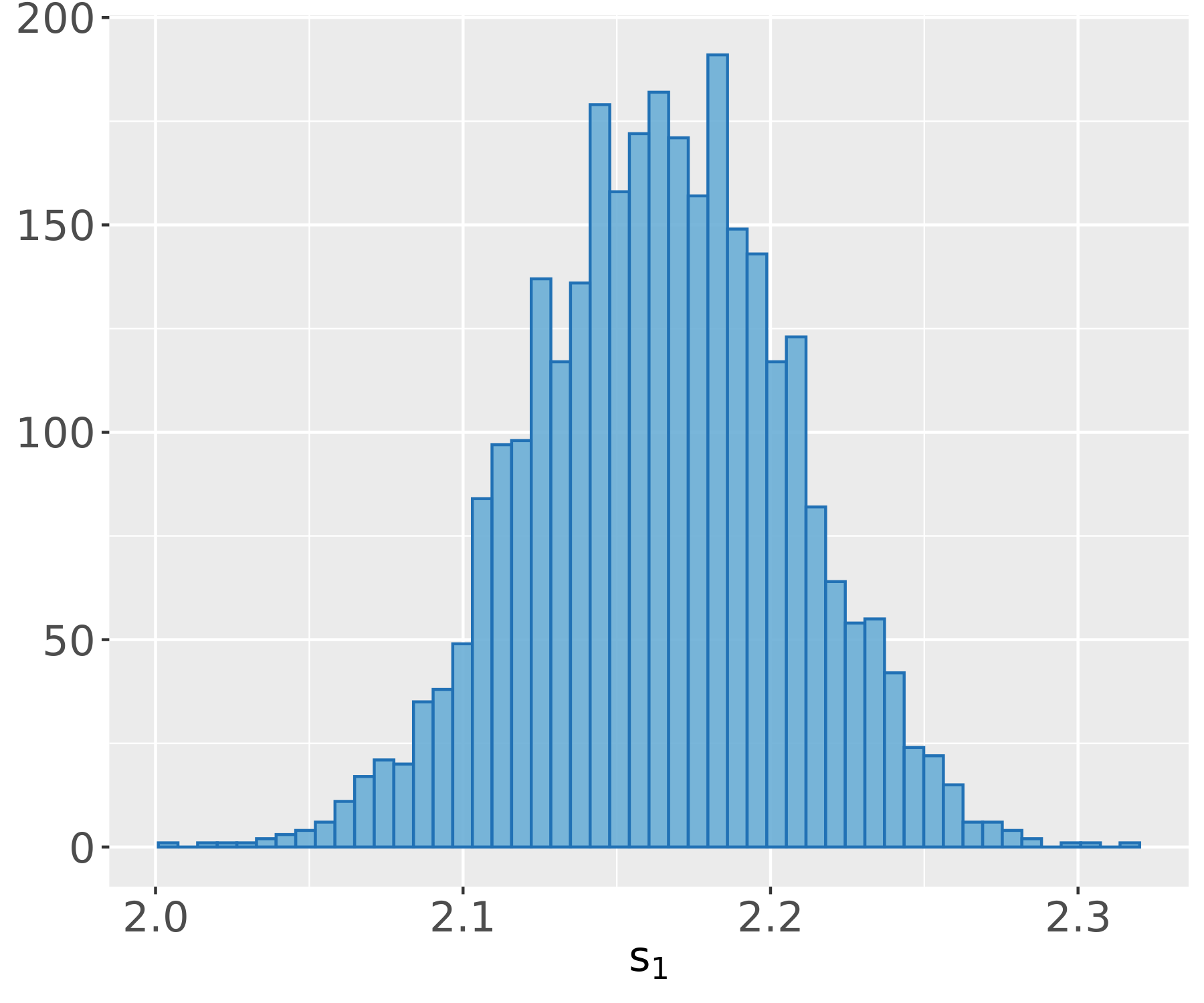}
\includegraphics[width=1.5in]{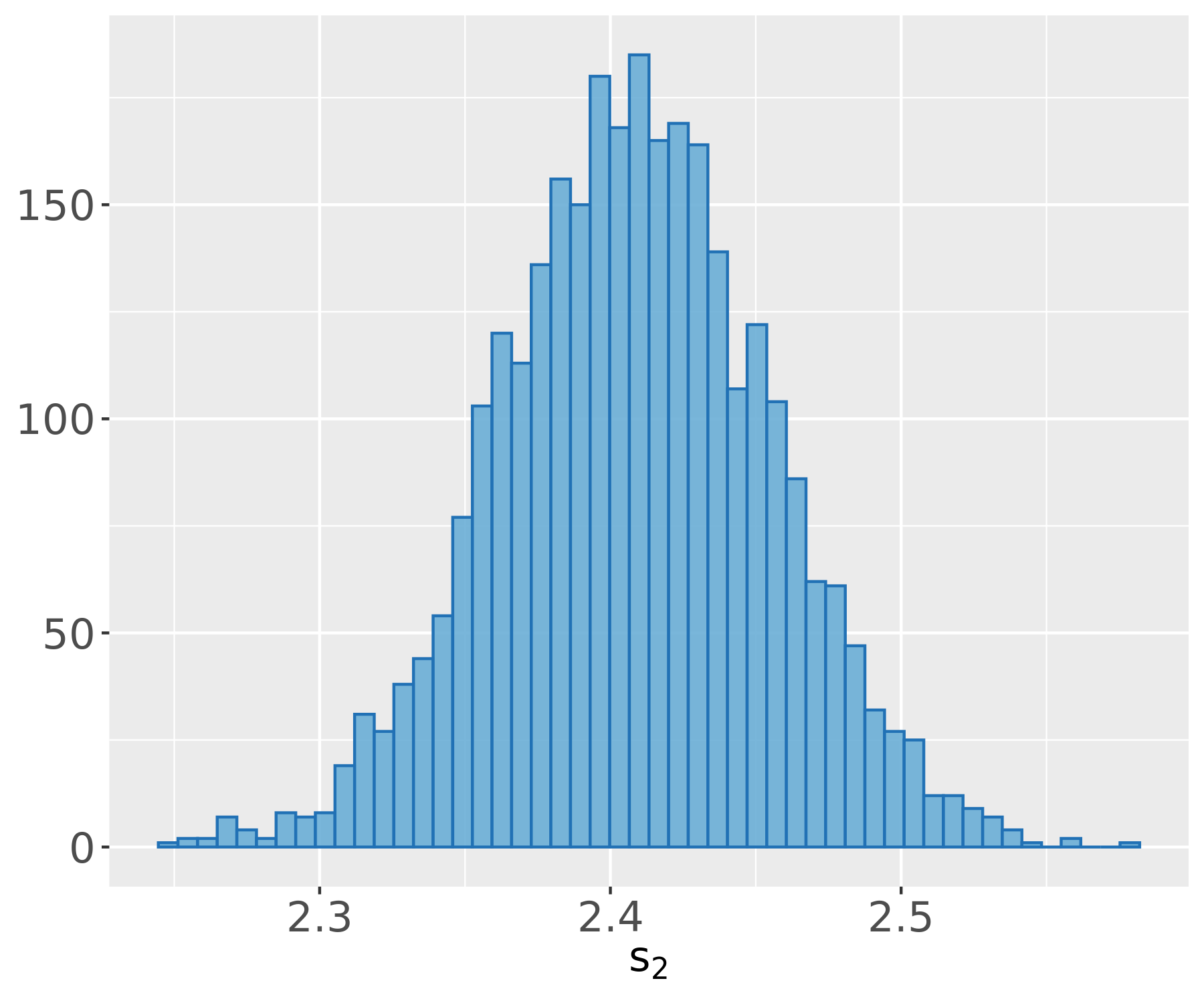}
\includegraphics[width=1.5in]{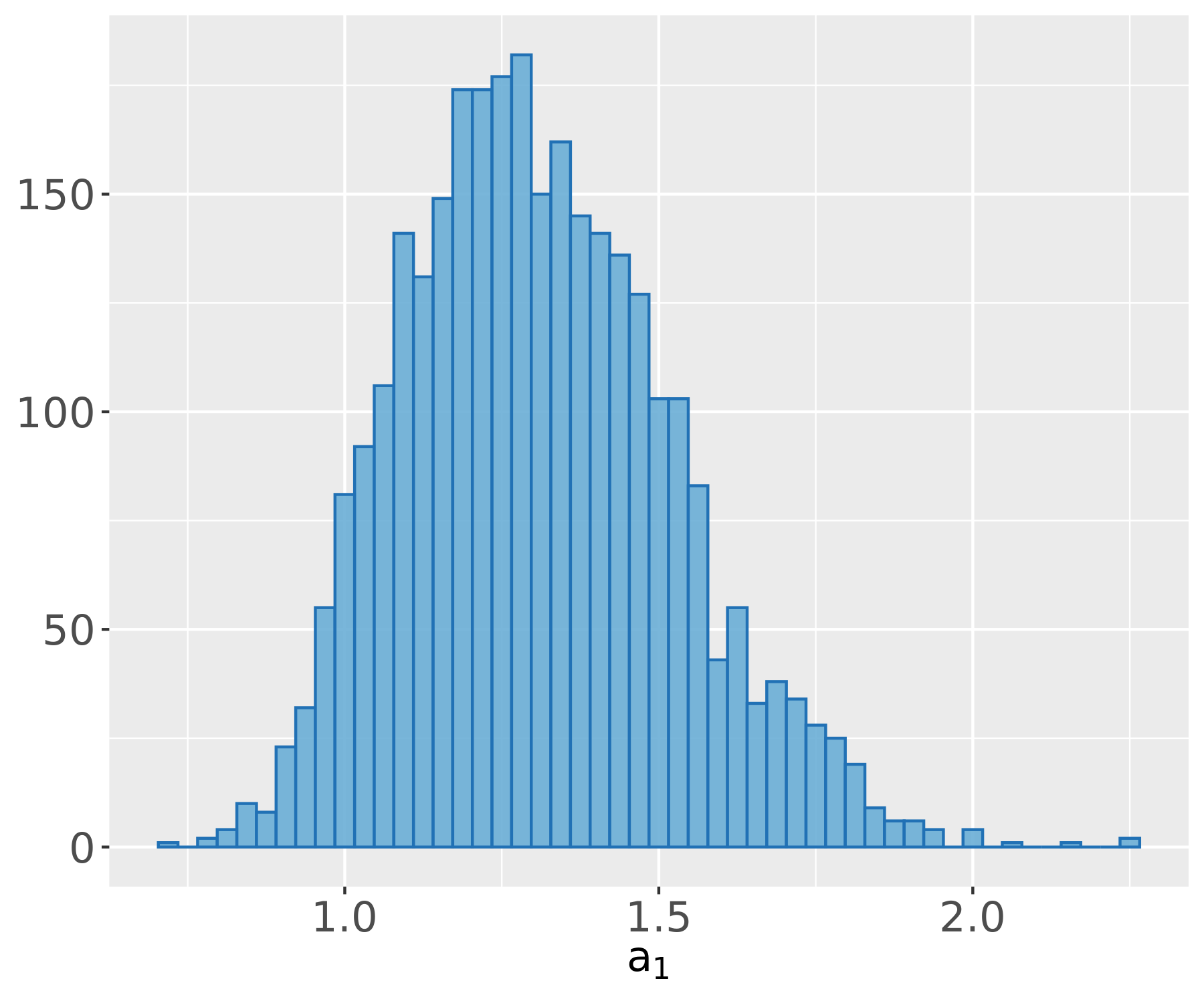}
\includegraphics[width=1.5in]{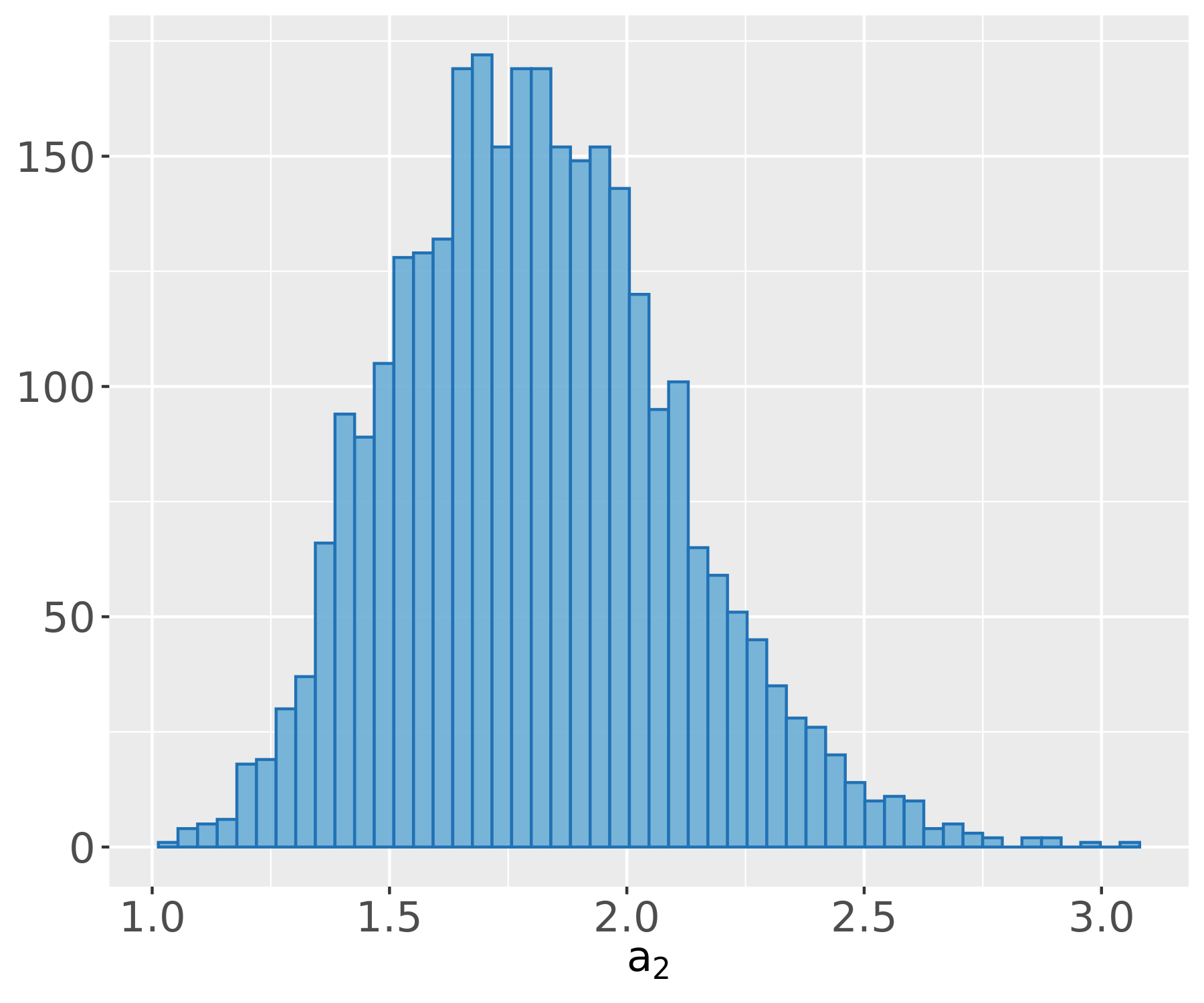}
\includegraphics[width=1.5in]{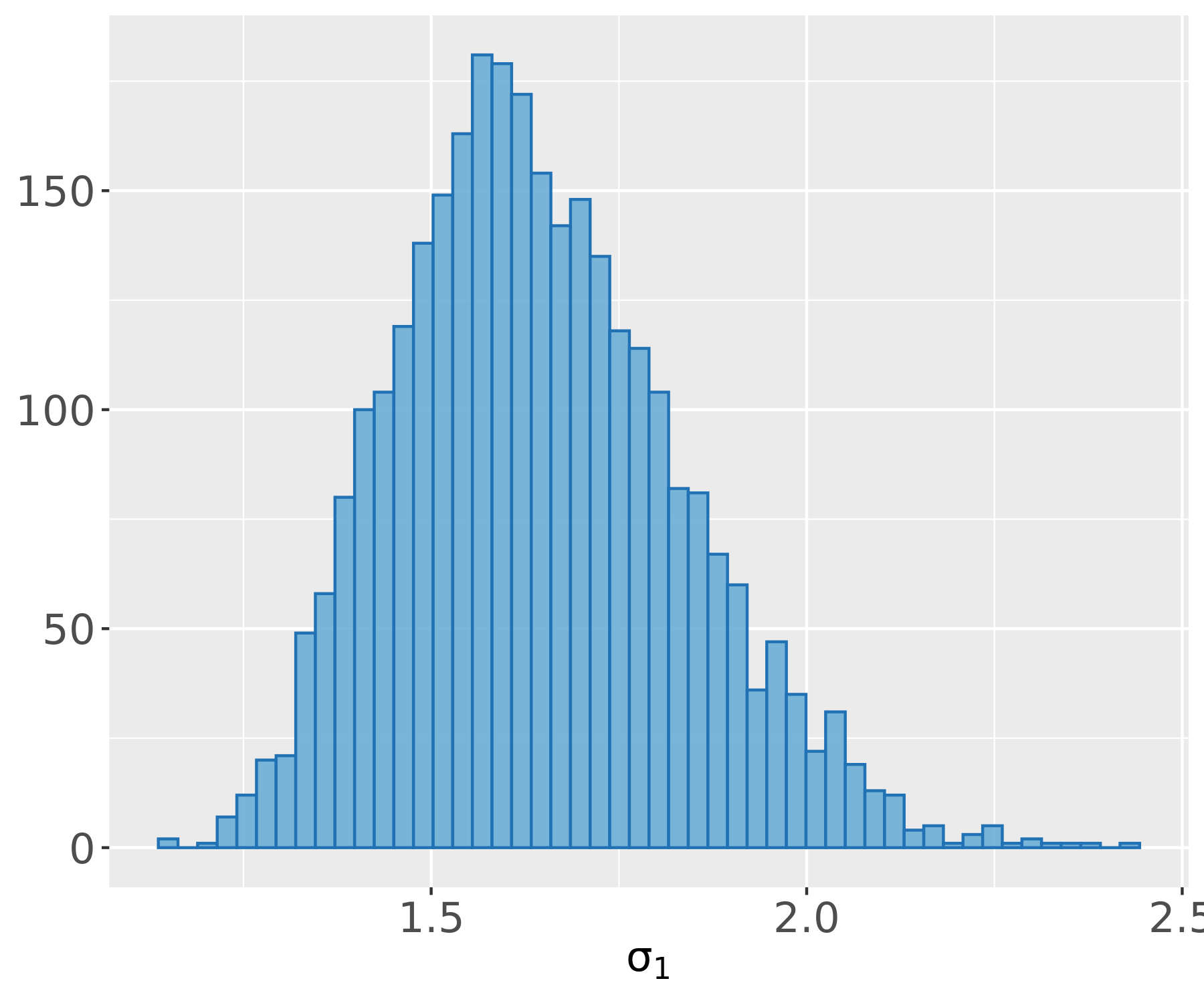}
\includegraphics[width=1.5in]{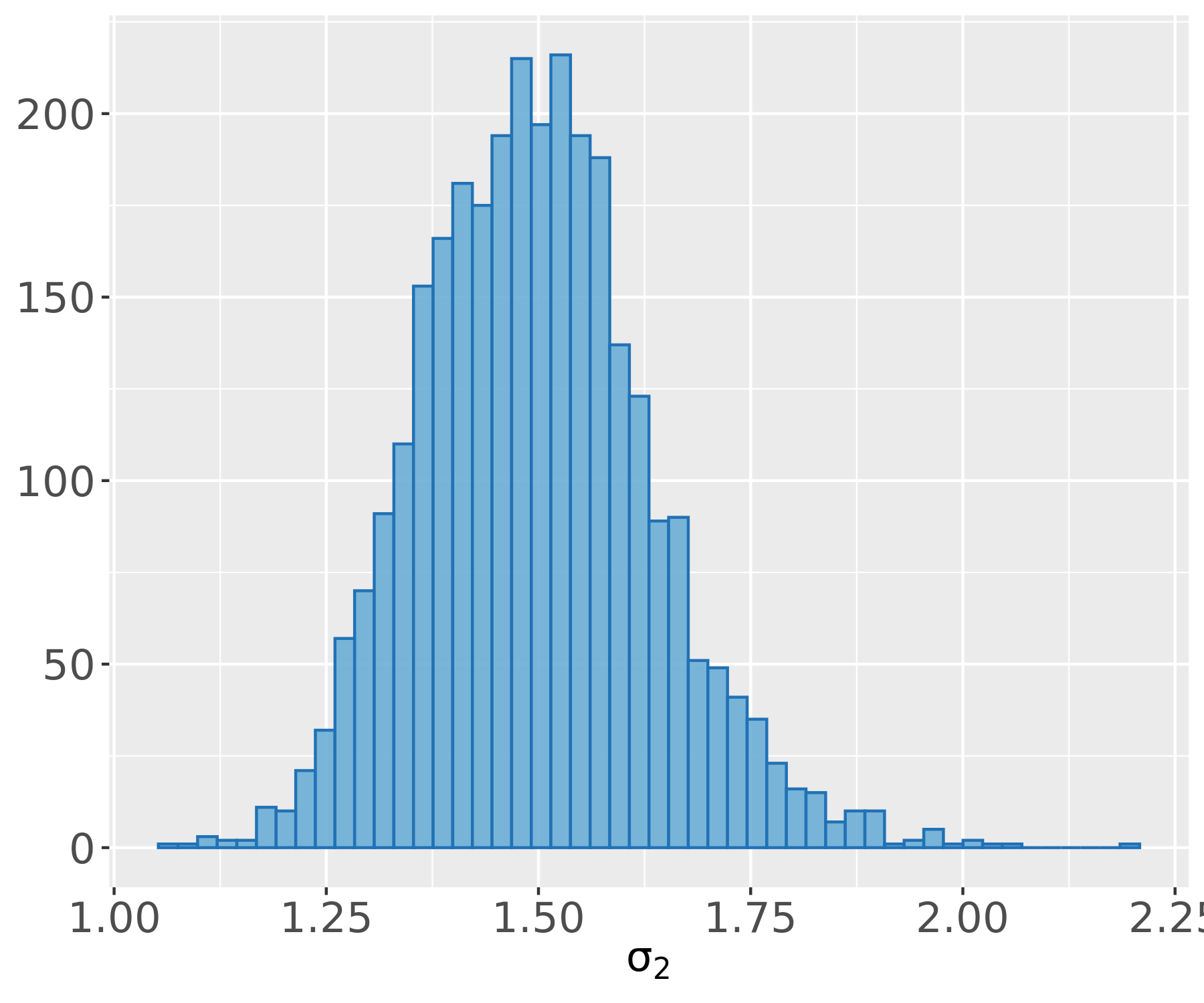}
\includegraphics[width=1.5in]{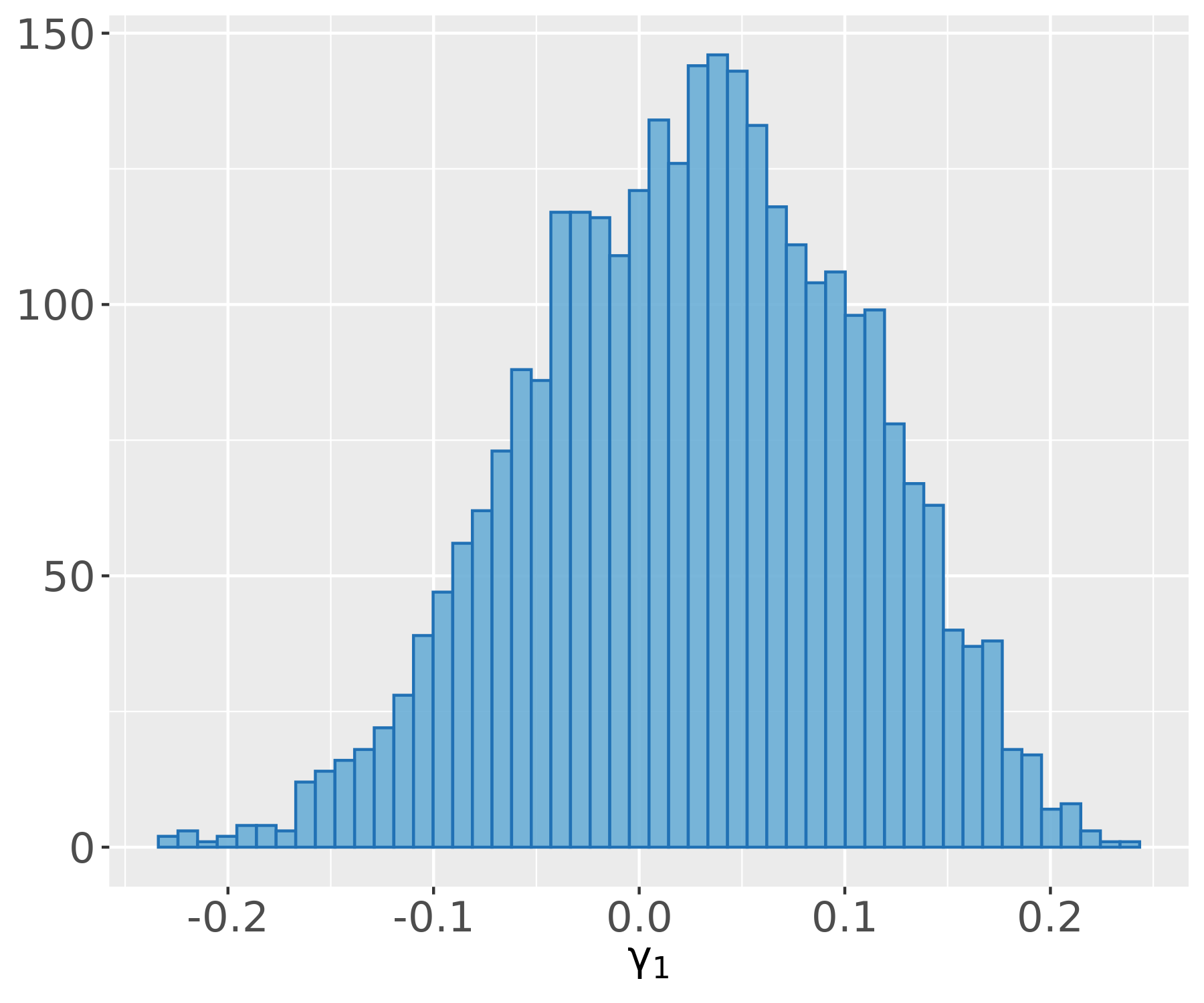}
\includegraphics[width=1.5in]{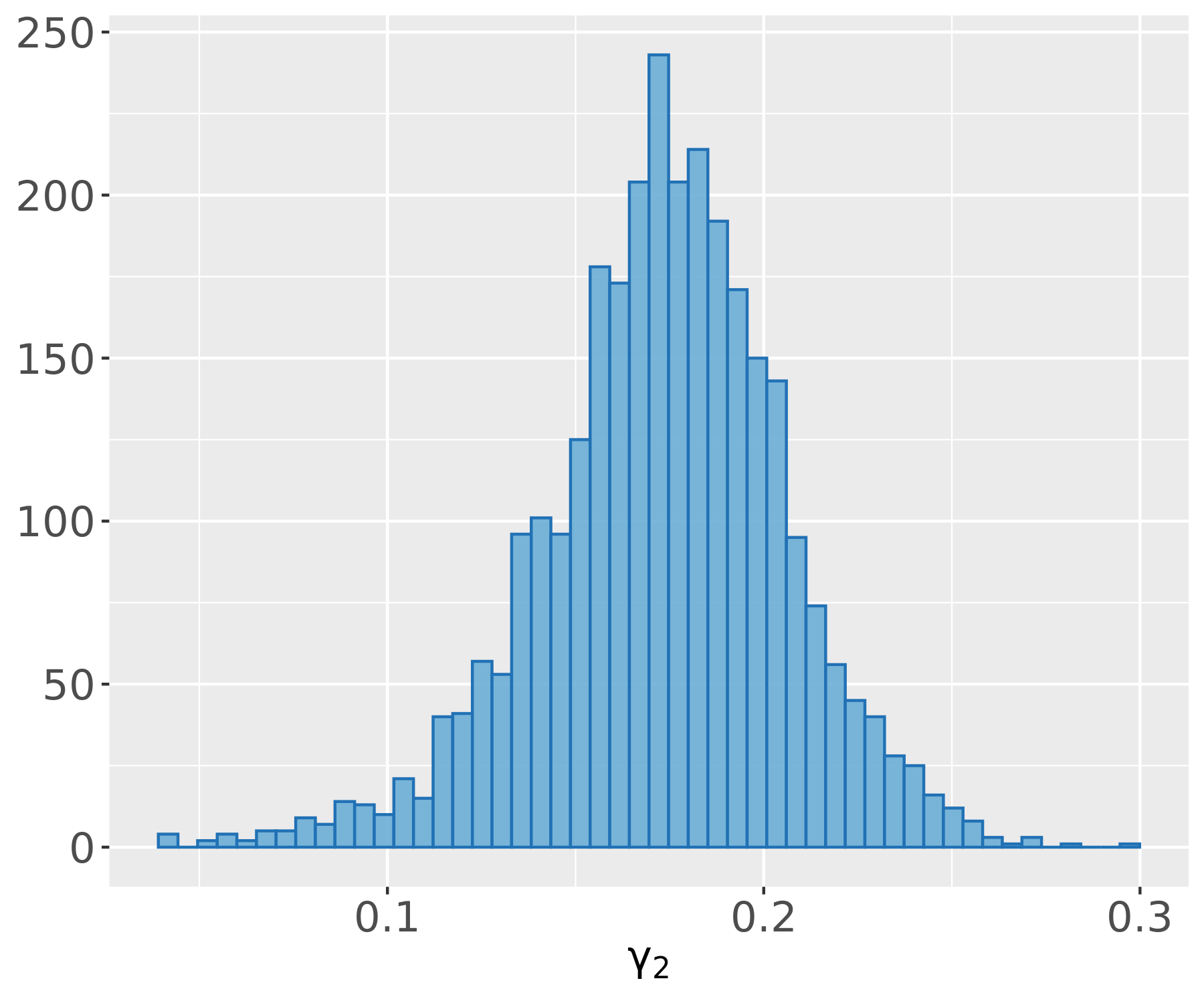}
\includegraphics[width=1.5in]{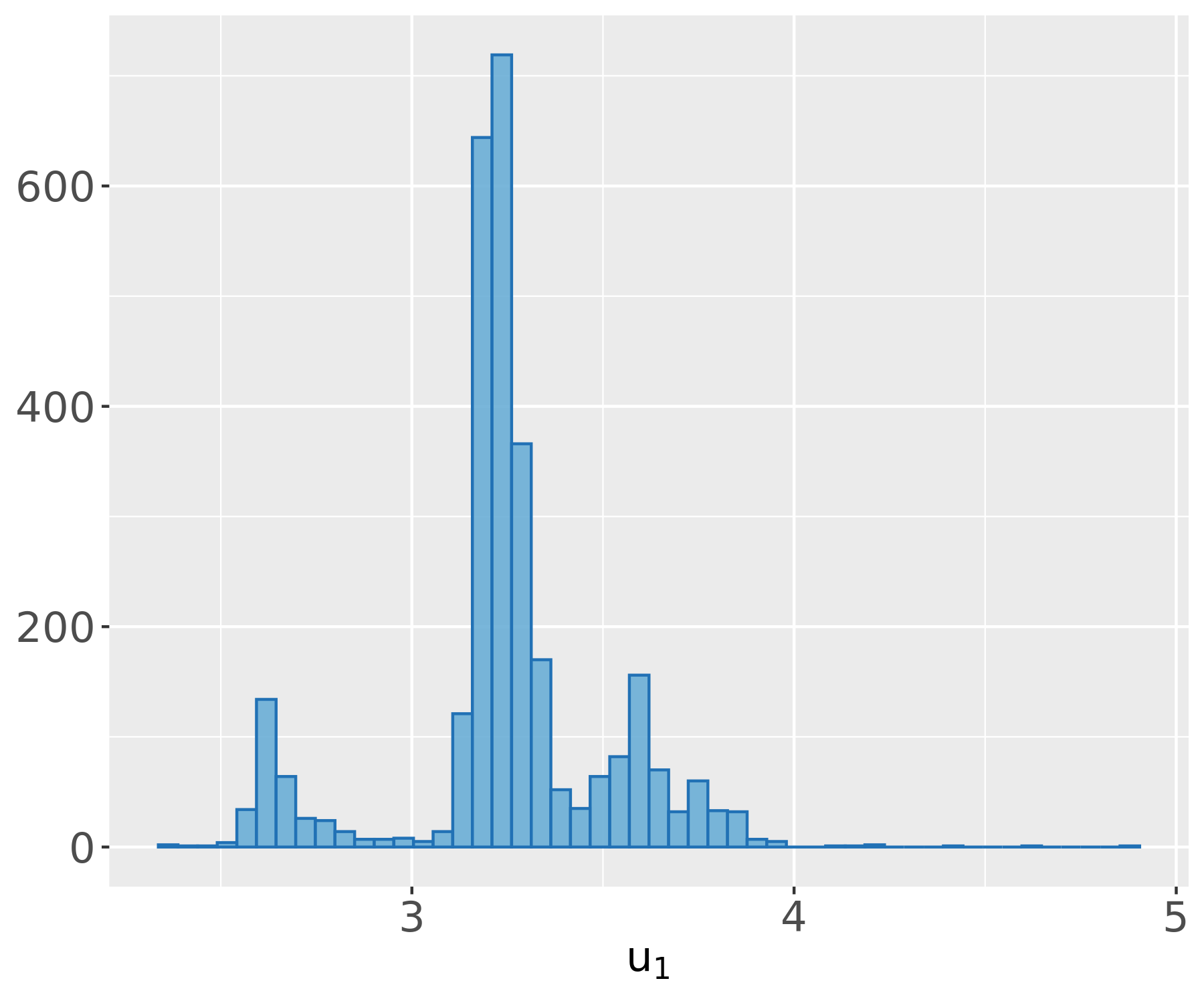}
\includegraphics[width=1.5in]{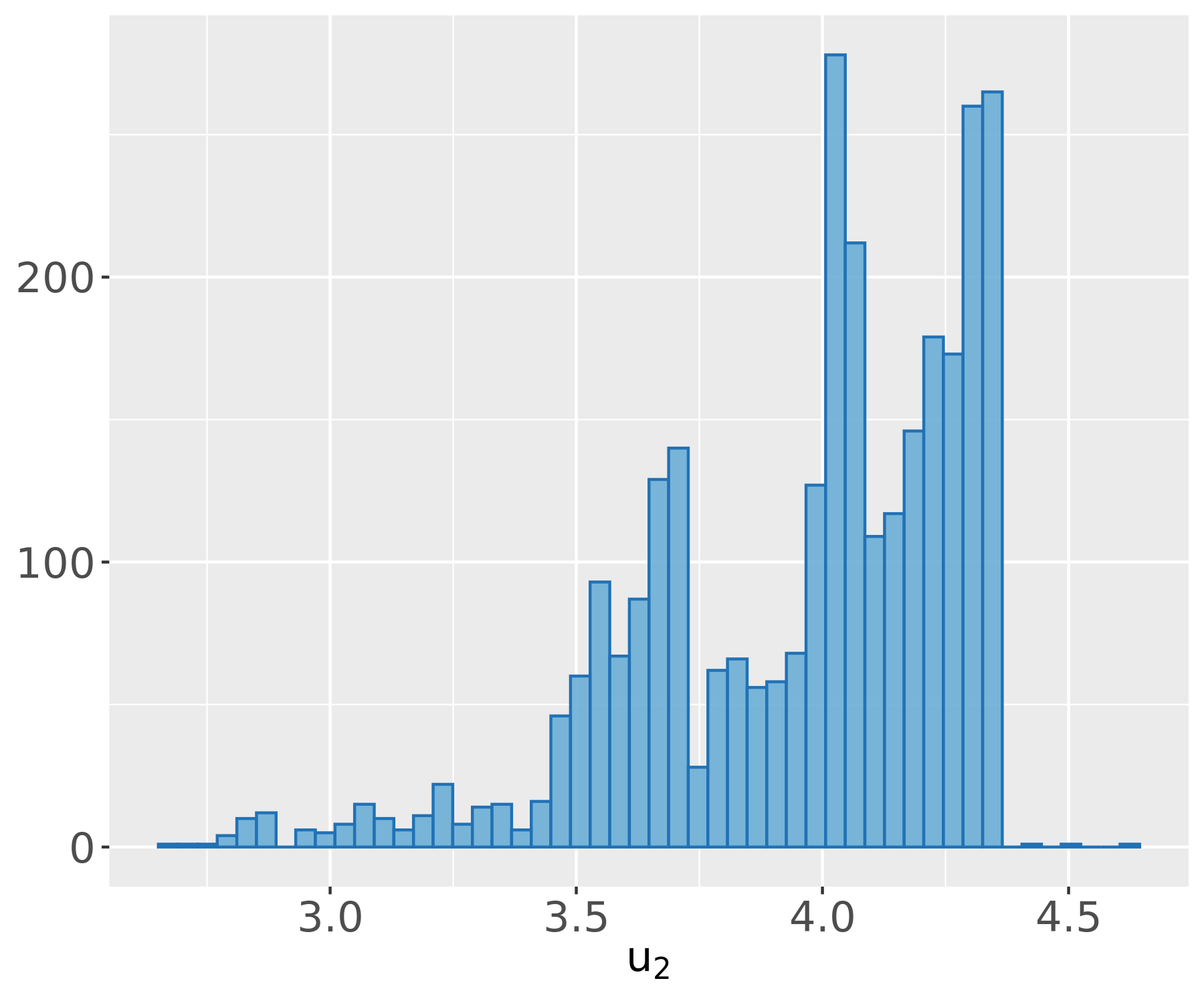}
\end{center}
\caption[\footnotesize{Histograms of the posterior distributions for all model parameters.}]{\footnotesize{Histograms of the posterior distributions for all model parameters. The first six plots correspond to bulk parameters; the remaining plots show tail parameters and threshold components (last two plots). }}
\label{fig: hist of posterior}
\end{figure}

\begin{figure}
\begin{center}
\includegraphics[width=3.5in]{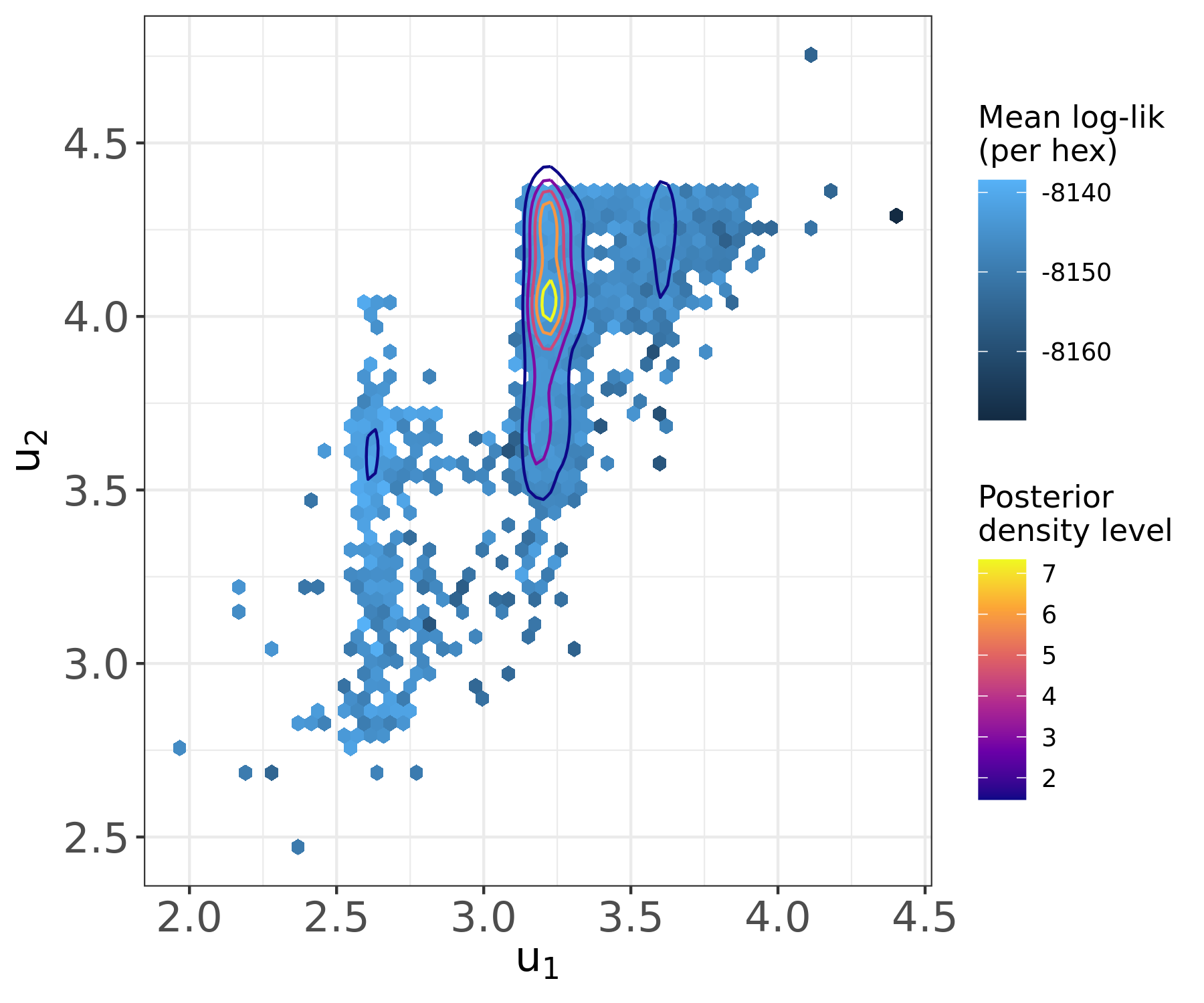}
\end{center}
\caption[\footnotesize{
Joint posterior density of the thresholds and the binned mean log-likelihood surface.}]{\footnotesize{
Joint posterior density of the thresholds and the binned mean log-likelihood surface.
Contours represent the joint posterior density $\pi(\boldsymbol{u} \mid \boldsymbol{X})$.
The background shading shows the binned mean log-likelihood,
$A(\boldsymbol{u}) = \mathbb{E}[\log p(\boldsymbol{X} \mid \boldsymbol{\theta}_t, \boldsymbol{\theta}_b, \boldsymbol{u})]$,
computed by averaging $\log p(\boldsymbol{X} \mid \boldsymbol{\theta}_t, \boldsymbol{\theta}_b, \boldsymbol{u})$
over posterior draws whose $\boldsymbol{u}$ values fall within each hexagonal bin.
}}
\label{fig: posterior_density}
\end{figure}

\paragraph{Marginal fit.} Marginal model performance is assessed using posterior predictive checks.
Specifically, we generate 3,000 replications of $\boldsymbol{E}$ from \eqref{eq:dist_rep} to construct posterior predictive distributions for marginal quantiles.
Figure \ref{fig:temp_qqplot} presents quantile–quantile (Q–Q) plots comparing posterior predictive quantiles from the replicated data to the empirical quantiles of $\boldsymbol{E}$, evaluated at uniformly spaced percentiles from 0 to 0.99. For both margins, the posterior predictive means lie close to the 45-degree line, which itself remains within the 95\% credible band. 
{For comparison, we include quantiles from a bivariate normal distribution, commonly used for residual modelling, fitted to $\boldsymbol{E}$ via maximum likelihood, though it lacks the tail correction offered by our mixture model.}
The results clearly show that our BEMM provides superior estimation, particularly for higher quantiles in both margins.

\paragraph{Dependence.} To assess dependence, we construct posterior predictive distributions for $\chi(r)$, $\bar{\chi}(r)$, and Kendall’s $\tau$ using the 3,000 replications of $\boldsymbol{E}$ generated above.
\begin{figure}[!htbp]
\begin{center}
\includegraphics[width=3in]{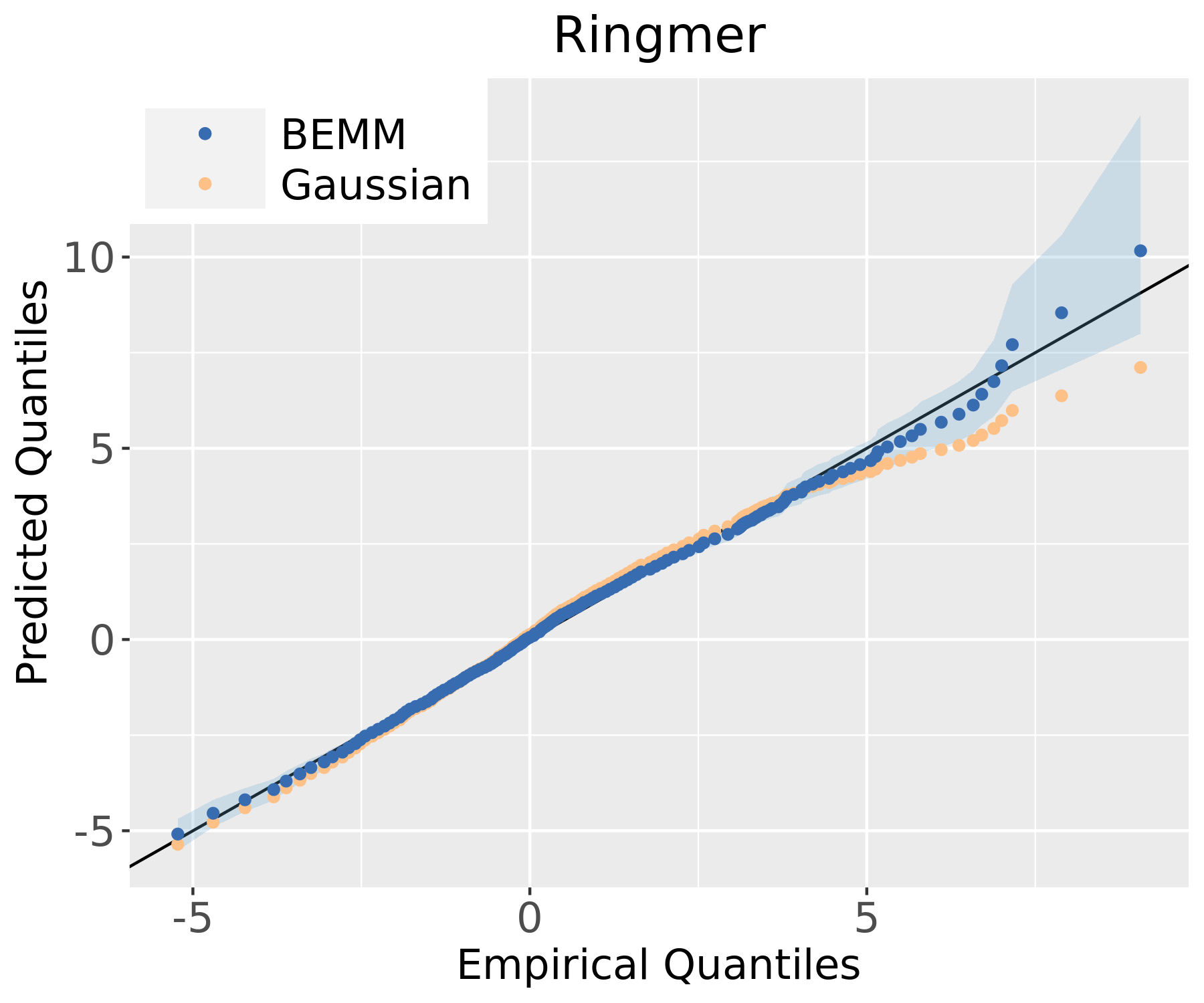}
\includegraphics[width=3in]{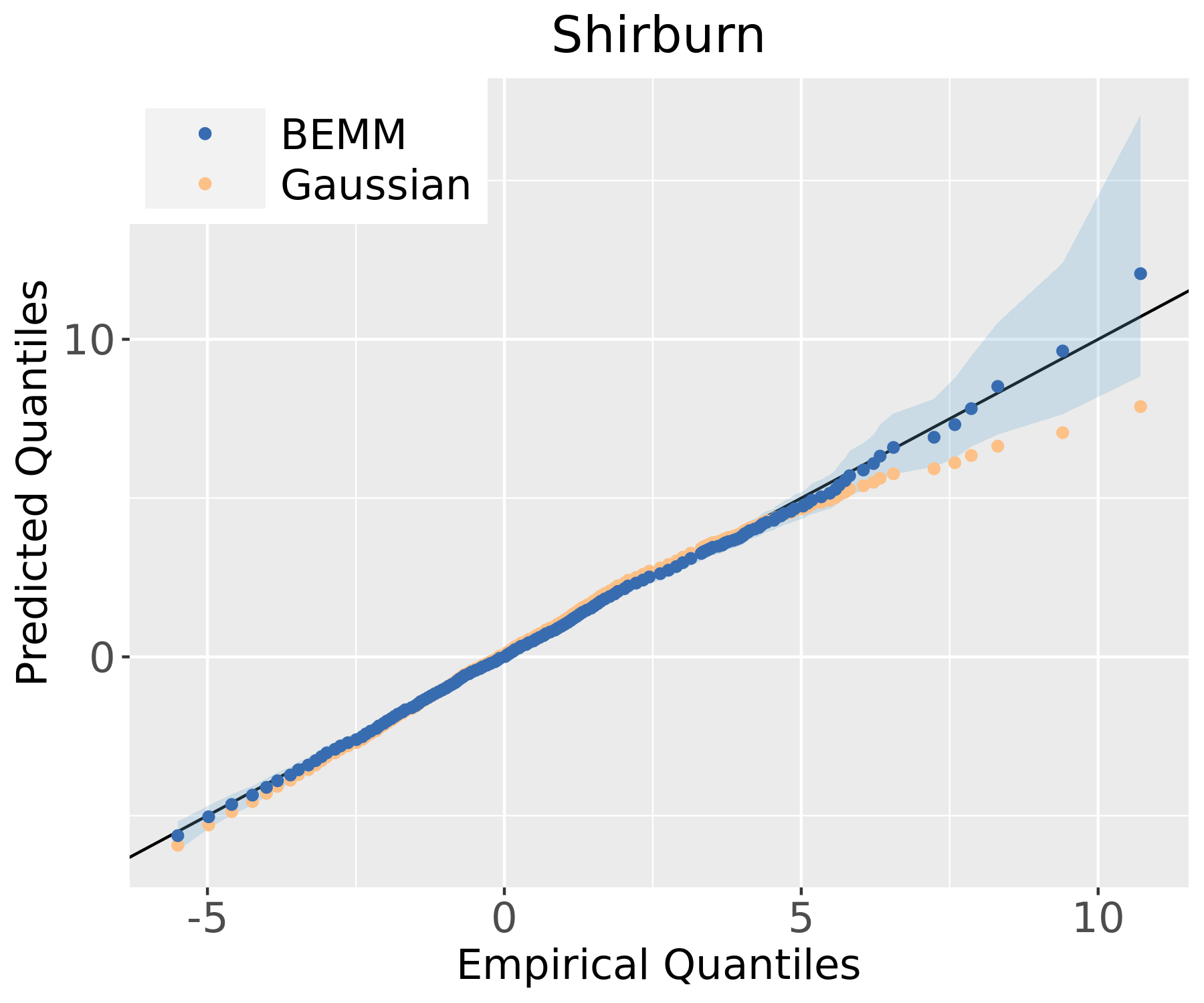}
\end{center}
\caption[\footnotesize{Q–Q plots of posterior predictive quantiles (blue) from 3,000 replicated datasets against the empirical quantiles.}]{\footnotesize{Q–Q plots of posterior predictive quantiles (blue) from 3,000 replicated datasets against the empirical quantiles. The shaded region indicates the 95\% credible band. Quantiles from a fitted bivariate normal distribution (orange) using the maximum likelihood method are included for comparison.}
\label{fig:temp_qqplot}}
\end{figure}
Figure \ref{fig:temp_dependence} shows posterior predictive distributions of $\chi(r)$, $\bar{\chi}(r)$, and Kendall’s $\tau$, alongside their empirical counterparts. The plots also include the 95\% credible interval for the theoretical $\chi$ computed from \eqref{eq:theoretical_chi}. For both $\chi(r)$ and $\bar{\chi}(r)$, empirical estimates fall within the 95\% credible bands for all $r$, except at $r = 0.68$. 
The theoretical $\chi$ lies near the centre of the posterior predictive $\chi(r)$ distribution, confirming the validity of \eqref{eq:theoretical_chi}.
As $r$ approaches 1, the empirical $\chi(r)$ moves toward the tail of the theoretical $\chi$ distribution, while $\bar{\chi}(r)$ approaches the posterior mean. Kendall’s $\tau$ also falls within the high-density region of its posterior predictive distribution.
All these suggest that the BEMM effectively captures the dependence structure. 
\begin{figure}[!htbp]
\begin{center}
\includegraphics[width=3in]{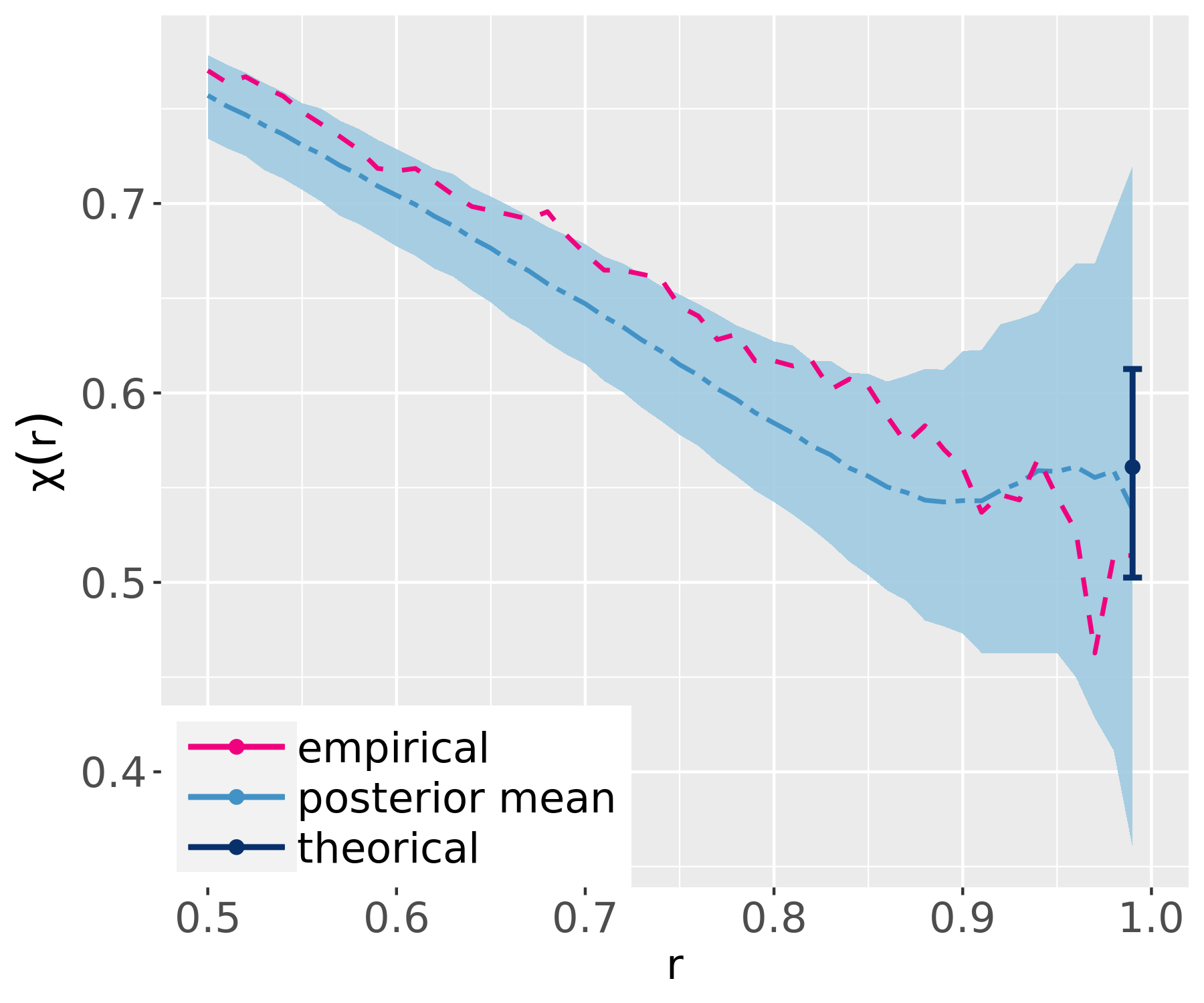}
\includegraphics[width=3in]{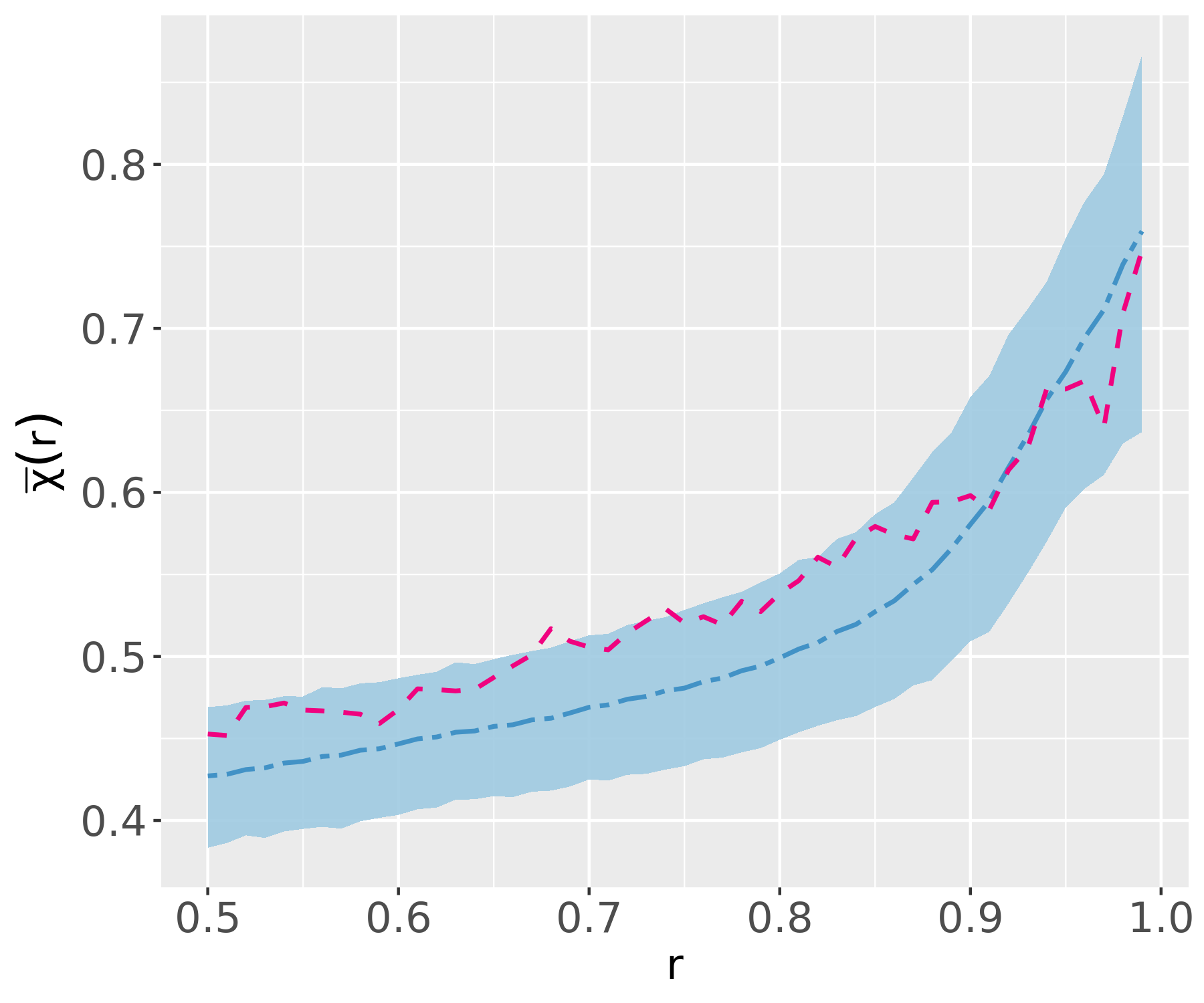}
\includegraphics[width=3in]{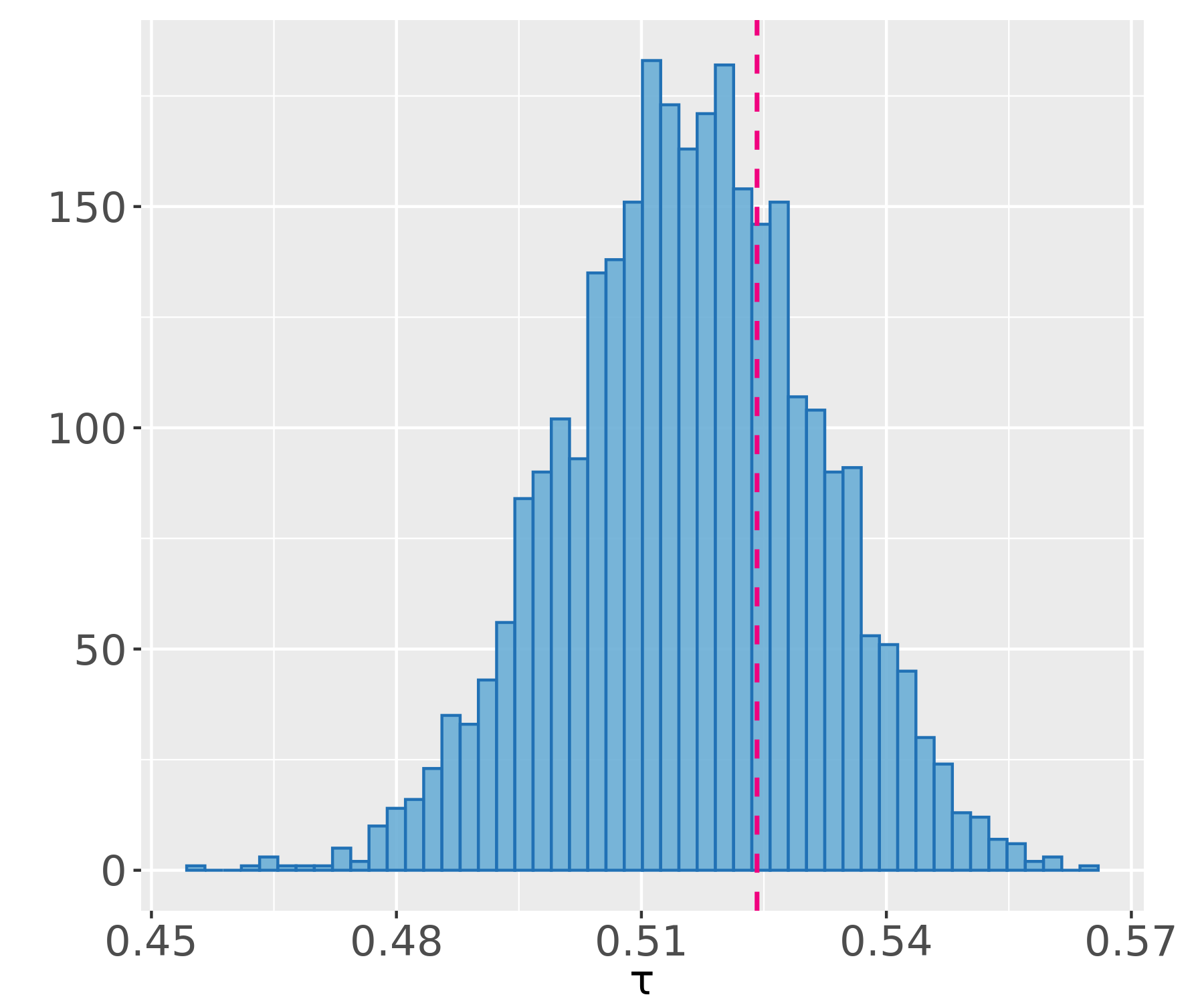}
\end{center}
\caption[\footnotesize{Based on 3,000 replicates from the posterior predictive distribution, the upper two plots present empirical values vs. means of the posterior prediction, each accompanied by a 95\% credible band, of $\chi(r)$ and $\bar{\chi}(r)$.}]{\footnotesize{Based on 3,000 replicates from the posterior predictive distribution, the upper two plots present empirical values vs. means of the posterior prediction, each accompanied by a 95\% credible band, of $\chi(r)$ and $\bar{\chi}(r)$. The error bar in the $\chi(r)$ plot represents the 95\% credible interval of the theoretical $\chi$, with the centre point indicating the posterior mean. The lower plot features a histogram of the posterior predictive Kendall's $\tau$, with the empirical value highlighted in red.}
\label{fig:temp_dependence}}
\end{figure}

\paragraph{Predictive accuracy.} We further compare the predictive accuracy of the BEMM and the Gaussian model using the energy score, a multivariate generalisation of the continuous ranked probability score (CRPS; \citealp{gneitingStrictlyProperScoring2007}). For an observation $\boldsymbol{y}$ and predictive distribution $F$, the energy score is
\begin{align}
    \text{ES}(F,\boldsymbol{y})=\mathbb{E}_F \left \| \boldsymbol{X}-\boldsymbol{y} \right \| - \frac{1}{2}\mathbb{E}_F \left \|  \boldsymbol{X}-\boldsymbol{X}^{\prime}\right \|,
    \label{eq:es}
\end{align}
where $\boldsymbol{X}$ and $\boldsymbol{X}'$ are independent draws from $F$, and $\|\cdot\|$ denotes the Euclidean norm. The score measures the expected distance between predictions and observations, adjusted for the spread of the predictive distribution, and is proper when $F$ has a finite expectation.

When specific regions of the distribution are of primary interest, such as the tail, threshold-weighted CRPS is often used.
This approach extends to the multivariate setting through the threshold-weighted energy score (twES; \citealp{allen2023weighted}):
\begin{equation*}
    \text{twES}(F,\boldsymbol{y}; v) = \mathbb{E}(\|v(\boldsymbol{X})-v(\boldsymbol{y})\|) -\frac{1}{2}\|v(\boldsymbol{X})-v(\boldsymbol{X}^\prime)\|
\end{equation*}
where $v: \mathbb{R}^d \to \mathbb{R}^d$ is a chaining function derived from a weight function $w: \mathbb{R}^d \to \mathbb{R}$ specifying the region of interest.
We consider two weighting schemes:
\begin{enumerate}
    \item \textbf{High-tail indicator weight}: $w_1(\boldsymbol{z})=\mathbbm{1}\{z_1>q_{1,0.9},z_2>q_{2,0.9}\}$, 
     where $q_{j,0.9}$ is the 0.9-quantile of margin $j$ of $\boldsymbol{E}$. The corresponding chaining function is
     \begin{equation*}
    v_1(\boldsymbol{z})=\begin{cases}
        \boldsymbol{z} &\quad \text{if} \quad w_1(\boldsymbol{z})=1\\
        \boldsymbol{x}_0 &\quad \text{if} \quad w_1(\boldsymbol{z})=0
    \end{cases}
    \end{equation*}
    for some fixed $\boldsymbol{x}_0 \in \mathbb{R}^2$.
  \item \textbf{Gaussian CDF weight}: $w_2(\boldsymbol{z})$ is defined as the CDF of a bivariate normal distribution fitted to $\boldsymbol{E}$. With estimated means $\hat{\mu}_1$, $\hat{\mu}_2$ and standard deviations $\hat{s}_1$, $\hat{s}_2$, the chaining function is
\footnotesize{
  \begin{equation*}
    v_2(\boldsymbol{z}) =\left( (z_1-\hat{\mu}_1)\Phi\left(\frac{z_1-\hat{\mu}_1}{\hat{s}_1}\right)+\hat{s}_1\phi\left(\frac{\hat{z}_1-\hat{\mu}_1}{\hat{s}_1}\right), (\hat{z}_2-\hat{\mu}_2)\Phi\left(\frac{z_2-\hat{\mu}_2}{\hat{s}_2}\right)+\hat{s}_2\phi\left(\frac{z_2-\hat{\mu}_2}{\hat{s}_2}\right) \right),
\end{equation*}
}
where $\Phi$ and $\phi$ are the standard normal CDF and PDF, respectively.
\end{enumerate}
Table \ref{tab:energyscore} reports the unweighted energy scores and their threshold-weighted variants for both the BEMM and Gaussian models. 
Across all metrics, the BEMM yields consistently lower scores, indicating superior predictive performance. However, because the models differ primarily in the tail region, the magnitude of the score differences is small.

\begin{table}[!htbp]
\centering
\small
\setlength{\tabcolsep}{20pt}
\renewcommand{\arraystretch}{1.1}

\caption[\footnotesize{Energy score (ES) and threshold-weighted energy score (twES) for the BEMM and Gaussian models.}]{\footnotesize{Energy score (ES) and threshold-weighted energy score (twES) for the BEMM and Gaussian models. 
Weighting scheme $W_1$ applies only to the region where all components exceed the 90th percentile of $\boldsymbol{E}$; 
$W_2$ uses the CDF of $\mathrm{MVN}(\hat{\boldsymbol{\mu}}, \hat{\boldsymbol{\Sigma}})$ as the weighting function, with 
$\hat{\boldsymbol{\mu}}$ and $\hat{\boldsymbol{\Sigma}}$ estimated from $\boldsymbol{E}$.}}
\vspace{0.5cm}
\begin{tabular}{l c c c}
\toprule
& \textbf{ES} & \textbf{twES $v_1$} & \textbf{twES $v_2$} \\
\midrule
\textbf{BEMM}     & 2.0155 & 0.2512 & 1.0361 \\
\textbf{Gaussian} & 2.0197 & 0.2517 & 1.0384 \\
\bottomrule
\end{tabular}

\label{tab:energyscore}
\end{table}

\section{Discussion}
\label{sec:discussion}
We now briefly discuss several practical challenges that may arise when implementing our model, along with potential solutions.

\subsection{Other distributions for the bulk and alternative representations for the tail}
The first issue concerns the choice of bulk and tail distributions.
For the bulk distribution, we illustrate our framework using a multivariate normal distribution for simplicity; however, this is by no means the only viable option.
Let $f_j(x_j)$ and $F_j(x_j)$, $j=1,\ldots,d$ denote the marginal density and CDF, respectively, and let $C(\cdot)$ be a copula with corresponding copula density $C(\cdot)$.
A multivariate distribution $f(\boldsymbol{x})$ can be constructed as  
\begin{equation*}
    f(\boldsymbol{x})= c(F_1(x_1),\ldots,F_d(x_d))\prod_{i=1}^d f_j(x_j),
\end{equation*}
with its CDF $F$ given by
\begin{equation*}
    F(x_1,\ldots,x_d) = C(F_1(x_1),\ldots,F_d(x_d)).
\end{equation*}
This allows for separate and flexible specification of the marginal distributions and the copula-based dependence structure to best fit the data.
The multivariate normal distribution fits naturally within this copula framework, with copula $C_R$ defined as
\begin{equation*}
    C_R(u_1,\ldots,u_d) = \Phi_R(\Phi^{-1}(u_1),\ldots,\Phi^{-1}(u_d)),\quad u_j \in [0,1],
\end{equation*}
where $\Phi_R$ is the CDF of a d-dimensional standard normal with correlation matrix $R$.

For mGPD in the tail, \cite{kiriliouk2019peaks} provided several alternative forms of generator $f_{\boldsymbol{U}}$ in \eqref{eq: mGPD_U_std}, including those with independent components of Gumbel, reverse Gumbel, log-gamma, and multivariate normal components.
These generators can be used in both R and T representations of the mGPD, leading to richer model forms.
Recent research explores the use of generative models from machine learning to reformulate the mGPD \citep{hu2025gpdflow, lhaut2025wasserstein}.
Although these generative-model-based mGPDs can be integrated into our framework, the inference procedure would need to shift to a frequentist approach because the large number of parameters in generative models renders posterior sampling via MCMC computationally infeasible.

\subsection{Model selection}
Given the variety of possible model components, a robust model selection procedure is essential.
Due to the large number of candidate models, information-criterion-based methods (e.g., WAIC \citep{watanabe2010asymptotic}) are generally preferred.
However, from a practical standpoint, fitting the full model in \eqref{eq:pdf} across all combinations of bulk and tail distributions is computationally prohibitive, due to the large number of possible distribution combinations and the long MCMC runs required to obtain sufficient effective sample sizes in the presence of strong within-chain autocorrelation.
We propose a heuristic two-stage approach for model selection:
\begin{enumerate}
    \item \textbf{Stage One}: Since the bulk and tail distributions are independent for a fixed threshold, set the threshold to a relatively high quantile and select the best-fitting bulk and tail distributions separately, using criteria such as WAIC, on the bulk and tail data, respectively. 
    This process can be repeated for multiple thresholds to account for variability, and the most frequently selected bulk and tail distributions are chosen.
    \item \textbf{Stage Two}: Fit the model in \eqref{eq:pdf} using the bulk and tail components identified in Stage One, following the inference procedure described in Section \ref{sec:posterior_inference}.
\end{enumerate}
This approach avoids exhaustive MCMC fitting for all model combinations while still yielding a reasonable and practical model choice.


\subsection{High-dimensional generalisation}
Our model is best suited for low-dimensional settings (e.g. $d\leq3$).
Although it is theoretically applicable in higher dimensions, achieving satisfactory performance becomes challenging for two main reasons.
First, in high-dimensional settings, complex cross-dimensional dependence structures are generally difficult to capture using distributions with simple copula constructions. 
While more flexible approaches, such as vine copulas \citep{kurowicka2010dependence} or normalising flows \cite{papamakarios2021normalizing}, can be used to model the bulk distribution, these methods are primarily designed for density estimation and do not yield closed-form cumulative distribution functions. 
As a result, the exceedance probabilities required by our framework are difficult to evaluate.
Second, the model’s tail behaviour is fully determined by the mGPD, which is only justified when the data distribution is in the max-domain of attraction of an mGEVD.
This is a strong condition that tends to hold in low dimensions but is rarely satisfied in high-dimensional contexts \citep{huser2025modeling}.
In practice, it is advisable to respect these constraints and limit the application of the model to low-dimensional cases.

\section{Conclusion}
\label{sec:conclusion}
This paper introduces a Bayesian multivariate extreme value mixture model designed to jointly capture bulk and extreme observations.
The proposed method models both the marginal distributions and the dependence structure simultaneously, and circumvents the threshold selection challenge in peaks-over-threshold modelling by treating the threshold as a learnable parameter.
Bivariate simulation studies demonstrate that our model performs well across diverse tail behaviours when properly specified, and remains competitive in estimation accuracy compared to alternative bivariate approaches, even under model misspecification.
An application to UK daily maximum air temperatures shows that the proposed model substantially improves the precision of residual tail estimates relative to a Gaussian benchmark.

We conclude by noting two limitations. First, the supports of the bulk and tail distributions may be misaligned.
The support of the mGPD depends on the marginal shape parameter $\gamma_j$, and it could be lower bounded $\gamma_j>0$ or upper bounded $\gamma_j<0$.
As a result, the bulk distribution may be unbounded below while the marginal mGPD has a lower bound, or the bulk distribution may be restricted to positive values while the mGPD extends to negative infinity.
Since the lower tail of the mGPD generally carries little mass (see Figure \ref{fig:plot_BEMM}), this mismatch is usually not problematic, but truncation of the mGPD can be applied if necessary.
Second, because the mGPD contributes to both bulk and tail components, censored-likelihood-based inference cannot be applied to values below the threshold to reduce bias in parameter estimation \citep{kiriliouk2019peaks}.
This means that mGPD estimation may remain biased even when the data exhibit asymptotic dependence.
A potential remedy is to employ an mGPD with a flexible and learnable dependence structure, such as one based on normalising flows \citep{hu2025gpdflow}, though this would require adapting the inference procedure to integrate deep learning techniques, as mentioned in Section \ref{sec:discussion}.

\bigskip
\begin{center}
{\large\bf Acknowledgements}
\end{center}
We thank L\'idia Andr\'e (Lancaster University) for facilitating the code to implement the bivariate mixture copula model.


\appendix
\section{ Supplementary Materials}
\renewcommand{\theequation}{A.\arabic{equation}}
\setcounter{equation}{0}
\renewcommand{\thefigure}{A.\arabic{figure}}
\setcounter{figure}{0}

\renewcommand{\thetable}{A.\arabic{table}}
\setcounter{table}{0}

\subsection{Code}
Code to fit the bivariate version of our extreme mixture model and reproduce the data application in Section~\ref{sec:app} is freely available at {\url{https://github.com/hcl516926907/Biv_Ext_Mix_Mod}}. 

\subsection{Proofs}\label{sec:appendix-proofs}
\subsubsection{The BEMM distribution is in the MDA of an mGEVD}

\begin{lemma}
The mGPD $H(\boldsymbol{x})$ in \eqref{eq:mgpd} is in the MDA of the mGEVD $G(\boldsymbol{x})$.
\label{lemma:1}
\end{lemma}

\begin{proof}
{Theorem 2.2 in \cite{rootzenMultivariateGeneralizedPareto2006} states that if $\boldsymbol{X}$ follows an mGPD $H(\boldsymbol{x})$ defined in \eqref{eq:mgpd}, then there exists an increasing curve $\boldsymbol{u}(t)$ with $P(\boldsymbol{X} \leq \boldsymbol{u}(t)) \rightarrow 1$ as $t \rightarrow \infty$ and a function $\boldsymbol{\sigma}(\boldsymbol{u}) > \boldsymbol{0}$ such that 
\begin{align*}
\mathbb{P}\left(\frac{\boldsymbol{X}-\boldsymbol{u}(t)}{\boldsymbol{\sigma}(\boldsymbol{u}(t))}\leq \boldsymbol{x}  \middle|\frac{\boldsymbol{X}-\boldsymbol{u}(t)}{\boldsymbol{\sigma}(\boldsymbol{u}(t))}\nleqslant \boldsymbol{0}\right) = H(\boldsymbol{x}).
\end{align*}
Following Theorem 2.1(ii) in the same paper, $H(\boldsymbol{x})$ is in the MDA of $G(\boldsymbol{x})$.}
\end{proof}

\begin{prop}
The distribution $F$ with density \eqref{eq:pdf} lies in the max-domain of attraction of an mGEVD, regardless of the choice of bulk distribution.
\label{Prop1}
\end{prop}
\begin{proof}[Proof of Proposition \ref{Prop1}]
When all components exceed the threshold, we have
$$F(\boldsymbol{x})=F_{\text{bulk}}(\boldsymbol{u}) + [1-F_{\text{bulk}}(\boldsymbol{u})]H_{\boldsymbol{U}}(\boldsymbol{x}-\boldsymbol{u}), \quad \boldsymbol{x}>\boldsymbol{u}.$$
Let $\boldsymbol{x}_F^*=(x_1^*,\cdots,x_d^*)$, where $x_j^*=\sup\{F_j(x)<1\}$ is the upper bound of margin $i$. Then,
\begin{align*}
\lim_{\boldsymbol{x} \to \boldsymbol{x}_F^{*-}} &= \frac{1-H_{\boldsymbol{U}}(\boldsymbol{x}-\boldsymbol{u})}{1-F(\boldsymbol{x})}\\
&=\frac{1-H_{\boldsymbol{U}}(\boldsymbol{x}-\boldsymbol{u})}{1- F_{\text{bulk}}(\boldsymbol{u}) - [1-F_{\text{bulk}}(\boldsymbol{u})]H_{\boldsymbol{U}}(\boldsymbol{x}-\boldsymbol{u})}\\
&=\frac{1}{1-F_{\text{bulk}}(\boldsymbol{u})}>0.
\end{align*}
By Lemma \ref{lemma:1}, there exist a sequence of vector $\boldsymbol{\alpha}_n>\boldsymbol{0}$ and $\boldsymbol{\beta}_n$ such that $H_{\boldsymbol{U}}^n(\boldsymbol{\alpha}_n(\boldsymbol{x}-\boldsymbol{u}) + \boldsymbol{\beta}_n) \rightarrow G(\boldsymbol{x})$. 
{Following the }proof to Theorem 2.1 in \cite{resnick_1971}{, the above} implies {that} $F^n(\boldsymbol{\alpha}_n \boldsymbol{x} + \boldsymbol{\beta}_n -\boldsymbol{\alpha}_n \boldsymbol{u}) \rightarrow G^{1-F_{\text{bulk}}(\boldsymbol{u})}(\boldsymbol{x})$, where $G^{1-F_{\text{bulk}}(\boldsymbol{u})}(\boldsymbol{x})$ is also an mGEVD by the max-stable property of $G$.
\end{proof}

\subsubsection{Condition for Finite Expectation of BEMM}
\begin{proof}
We restrict attention to the BEMM with a bivariate normal distribution for the bulk and an mGPD with an independent reverse exponential generator for the tail.  
Since the expectation of the bivariate normal is finite, it suffices to ensure that the mGPD component also has a finite expectation.  
Let $\boldsymbol{Z}=(Z_1,Z_2)$ follow the standardized mGPD in \eqref{eq:explict_mgpd}.
The mGPD random variable $\boldsymbol{X}=(X_1,X_2)$ is obtained from \eqref{eq:transformation} as
$$\boldsymbol{X}=\frac{\boldsymbol{\sigma}}{\boldsymbol{\gamma}}(\exp\{\boldsymbol{\gamma Z}\}-1).$$
Then,
$$\mathbb{E}(X_j)=\mathbb{E}(X_j\mid X_j \geq0)\mathbb{P}(X_j \geq0) + \mathbb{E}(X_j \mid X_j<0)\mathbb{P}(X_j<0)$$
The first expectation term, $\mathbb{E}(X_j \mid X_j>0)$, is the expectation of a univariate GPD, and hence is finite if $\gamma_j<1$.
The second term is
\begin{align*}
    \mathbb{E}(X_j\mid X_j \leq0)&=\frac{\sigma_j}{\gamma_j}\mathbb{E}(\exp\{\gamma_j Z_j\}\mid Z_j \leq0)-1,
\end{align*}
where
\begin{align*}
    \mathbb{E}(\exp\{\gamma_j Z_j\}\mid Z_j \leq0) &=\int_{-\infty}^0\exp\{\gamma_j z_j\}\int_{0}^{\infty}h_{\boldsymbol{U}}(z_j,z_{-j})\text{d}z_{-j} \text{d}z_j\\
    &\propto \int_{-\infty}^0 \exp\{\gamma_j z_j\}\cdot a_j^{-1} \exp\{a_j^{-1} z_j\} \text{d} z_j\\
    &\propto \int_{-\infty}^0 \exp\{(\gamma_j + a_j^{-1}) z_j\} \text{d} z_j
\end{align*}
The last integration is finite if and only if $\gamma_j + a_j^{-1} > 0$.
\end{proof}

\subsubsection{Theoretical $\chi$ for the BEMM}
\begin{proof}
This proof is based on the mGPD with an independent reverse exponential generator (density given in \eqref{eq:explict_mgpd}).  
Let $F_1$ and $F_2$ be the marginal distributions in \eqref{eq:pdf}, and let $\boldsymbol{u} = (u_1, u_2)$ be a threshold vector satisfying $F_j(u_j) < 1$, $j = 1, 2$.  
For {$\max \{F_1(u_1), F_2(u_2)\}<r<1$},
\begin{align*}
\mathbb{P}(F_j(X_j)>r) &= \mathbb{P}\left(F_{\text{bulk}}(u_1,u_2)+[1-F_{\text{bulk}}(u_1,u_2)]H_j(X_j-u_j)>r\right)\\
&= \mathbb{P}\left(H_j(X_j-u_j)>\frac{r-F_{\text{bulk}}(u_1,u_2)}{1-F_{\text{bulk}}(u_1,u_2)}\right),
\end{align*} 
where $H_j$ is the $j$th marginal of the mGPD in \eqref{eq:explict_mgpd}.  
Let $r^*=[r-F_{\text{bulk}}(u_1,u_2)]/[1-F_{\text{bulk}}(u_1,u_2)]$. 
Then
\begin{align*} 
\chi = \lim_{r^*\rightarrow 1^-}\chi(r^*)&=\frac{\mathbb{P}(H_1(X_1-u_1)>r^*,H_2(X_2-u_2)>r^*)}{1-r^*}\\
&=1-\left(\frac{1+a_{(1)}}{1+a_{(2)}} \right )^{1+a_{(2)}^{-1}}\frac{a_{(2)}}{a_{(1)}}\frac{a_1a_2}{a_1a_2+a_1+a_2},
\end{align*} where $a_{(1)}=\min \{a_1,a_2\}$, $a_{(2)}=\max \{a_1, a_2\}$. 
This follows from the $\chi$ of the mGPD with an independent reverse exponential generator, as derived in \cite{kiriliouk2019peaks}.
\end{proof}

\subsection{Stationarity check of the temperature data}
\label{sec:appendix_statchecks}
We apply \eqref{eq: temp_lr} to remove the seasonal cycle and reduce autocorrelation in the daily air temperature data.  
Stationarity is assessed using the autocorrelation function (ACF) of the residuals at each station (Figure \ref{fig:autocorrelation}).  
The ACF values are close to zero for lags up to 32 days, indicating that seasonality and autocorrelation have been effectively removed.
\begin{figure}[h]
\begin{center}
\includegraphics[width=3in]{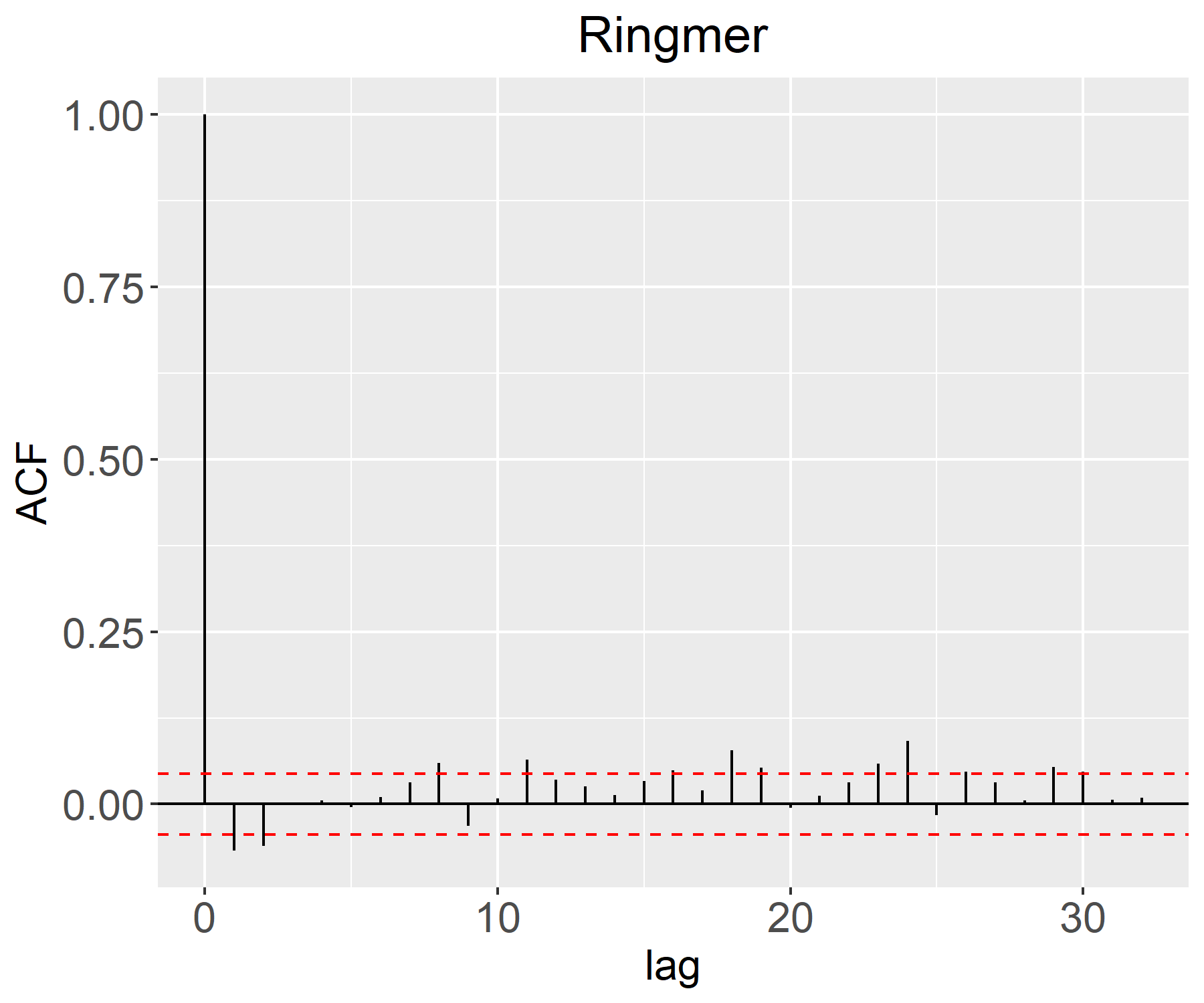}
\includegraphics[width=3in]{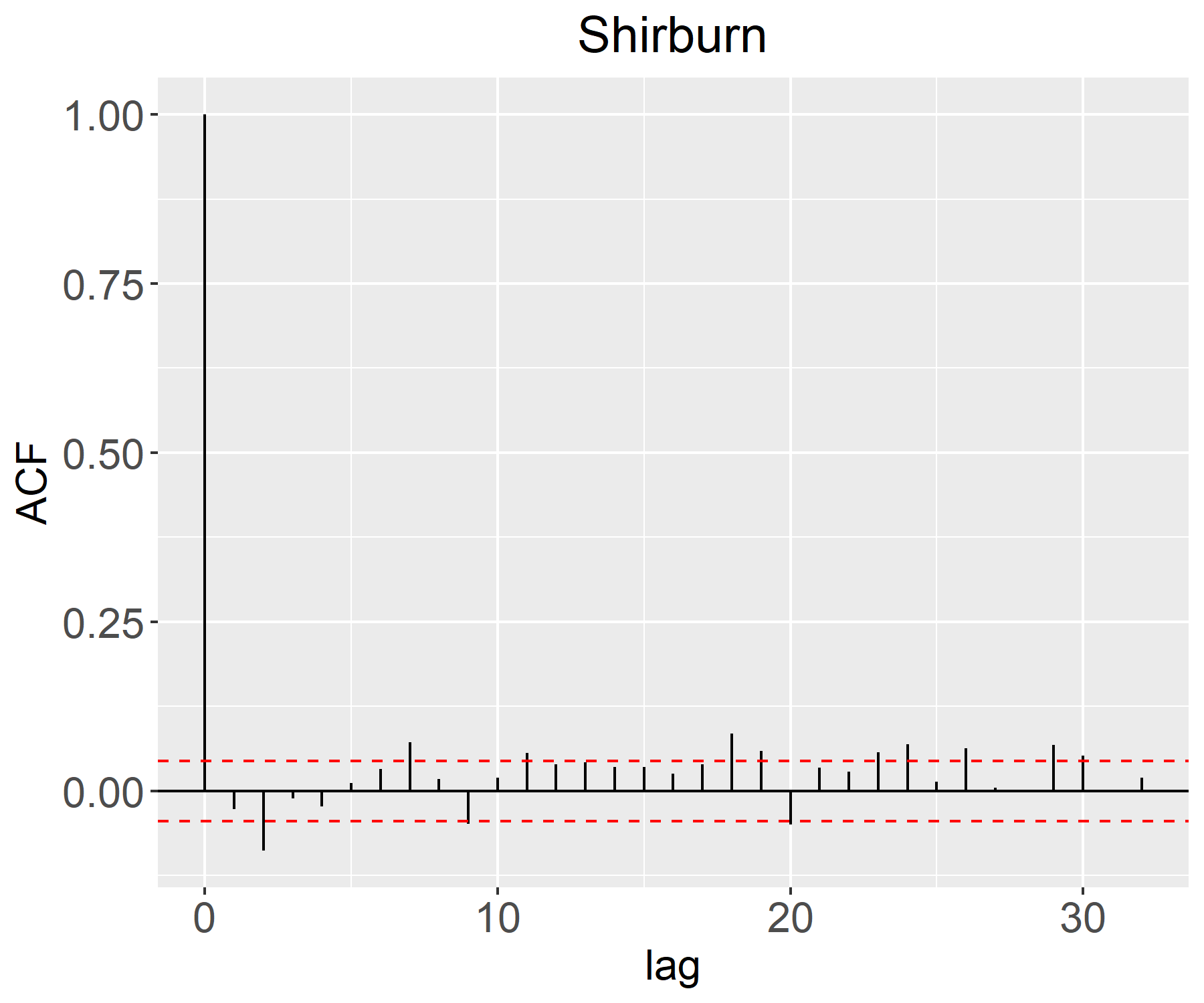}
\end{center}
\caption{\small{Autocorrelation plots of the two sites after removing the seasonality and autocorrelation by \eqref{eq: temp_lr}. The red dashed line indicates the confidence band at 95\% confidence level. 
}}
\label{fig:autocorrelation}
\end{figure}

\subsection{MCMC convergence diagnostics}
\label{sec:appendix_mcmcchecks}
\setcounter{figure}{0}   
\setcounter{table}{0} 

We assess convergence of the MCMC results, both in simulations and data applications, using the potential scale reduction factor $\hat{R}$ \citep{Gelman1992}, the effective sample size (ESS), and trace plots.  
The statistic $\hat{R}$ compares the variance of parameter estimates across multiple chains (total variance) with the average variance within each chain. Well-mixed chains yield $\hat{R} \approx 1$.  
ESS, derived from $\hat{R}$, measures the number of effectively independent posterior draws.  
We use the rank-normalised $\hat{R}$ and ESS proposed by \citet{Vehtari2021}, which are more robust than the traditional versions.  
For the data application, these diagnostics are shown in Table~\ref{tab:rhat} and Fig.~\ref{fig:traceplot}; simulation results are similar.  
Following \citet{Vehtari2021}, stable inference typically requires $\hat{R} < 1.01$ and $\text{ESS} > 400$.

\begin{table}[h]
\centering
\small
\setlength{\tabcolsep}{6pt}
\renewcommand{\arraystretch}{1.1}
\caption[\footnotesize{Rank-normalised $\hat{R}$ and ESS (in parentheses) for the MCMC results in the data application.}]{\footnotesize{Rank-normalised $\hat{R}$ and ESS (in parentheses) for the MCMC results in the data application. Three parallel chains of 20{,}000 iterations each were run. The first 10{,}000 iterations were discarded as burn-in, and the remaining samples were thinned by retaining every tenth draw.}}
\vspace{0.5cm}
\begin{tabular}{c c c c c c}
\toprule
\multicolumn{2}{c}{\textbf{Bulk Parameters}} &
\multicolumn{2}{c}{\textbf{Tail Parameters}} &
\multicolumn{2}{c}{\textbf{Threshold}} \\
\cmidrule(lr){1-2}\cmidrule(lr){3-4}\cmidrule(lr){5-6}
\textbf{Param} & $\boldsymbol{\hat R}$ (ESS) &
\textbf{Param} & $\boldsymbol{\hat R}$ (ESS) &
\textbf{Param} & $\boldsymbol{\hat R}$ (ESS) \\
\midrule
$U[1,2]$ & 1.0007 (2034) & $a_1$     & 1.0015 (2101) &        &            \\
$U[2,2]$ & 1.0007 (2030) & $a_2$     & 1.0009 (1992) &        &            \\
$\mu_1$  & 1.0044 (1478) & $\sigma_1$& 1.0000 (1918) & $u_1$  & 1.0007 (775) \\
$\mu_2$  & 1.0031 (1806) & $\sigma_2$& 0.9996 (2662) & $u_2$  & 1.0018 (931) \\
$s_1$    & 1.0014 (1794) & $\gamma_1$& 0.9997 (1670) &        &            \\
$s_2$    & 1.0005 (1885) & $\gamma_2$& 1.0006 (2716) &        &            \\
\bottomrule
\end{tabular}
\label{tab:rhat}
\end{table}

\begin{figure}[!htbp]
\begin{center}
\includegraphics[width=1.6in]{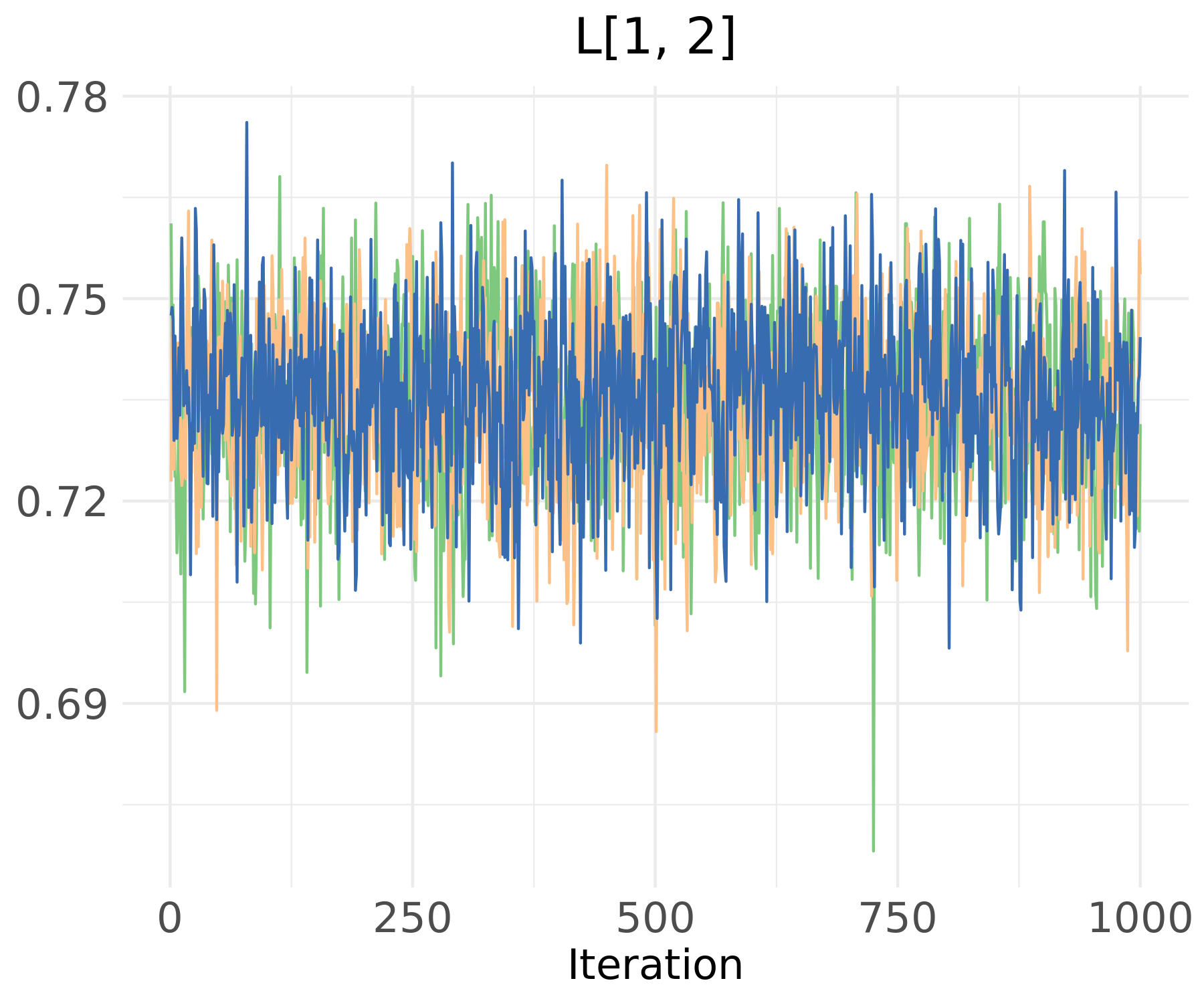}
\includegraphics[width=1.6in]{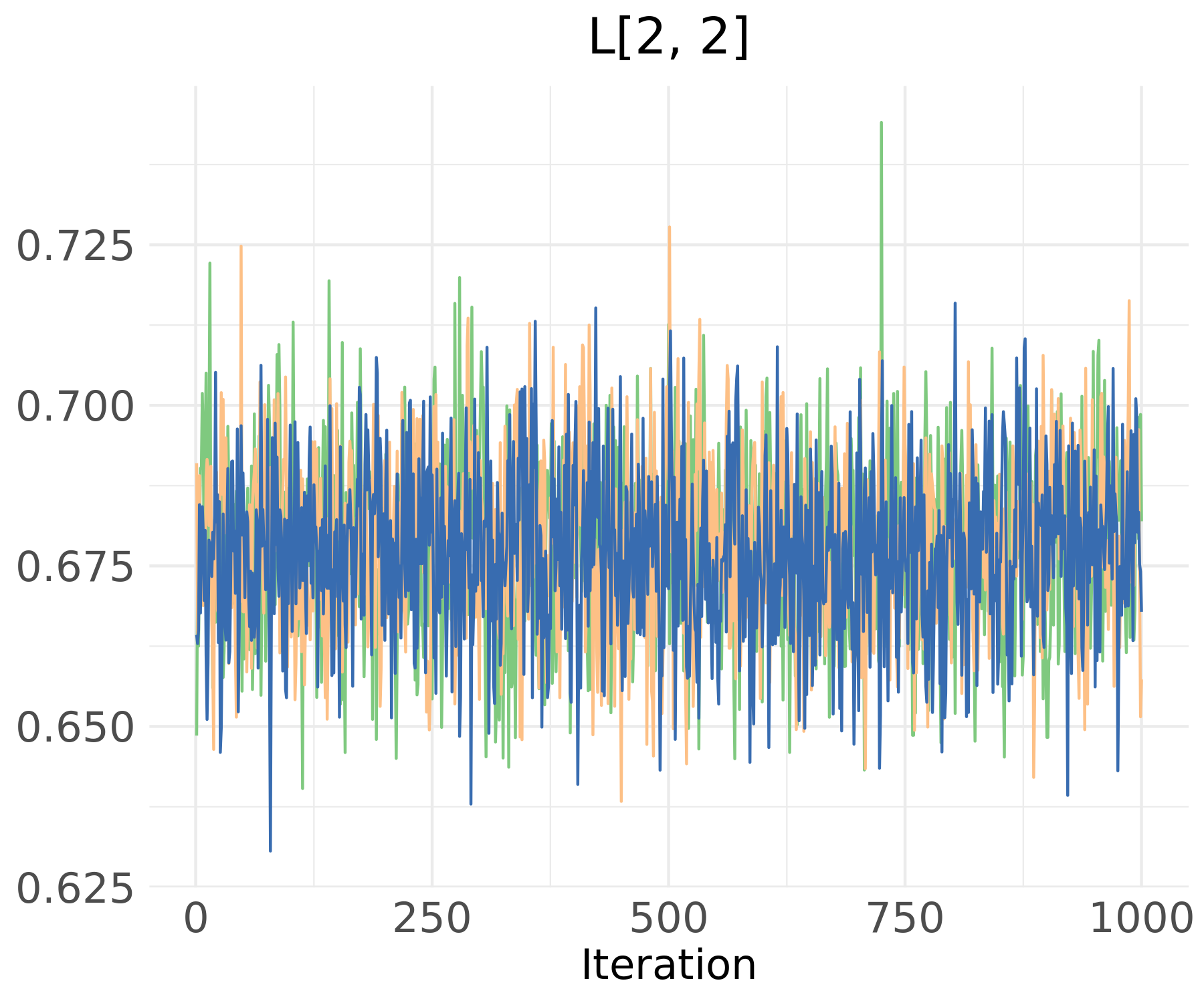}
\includegraphics[width=1.6in]{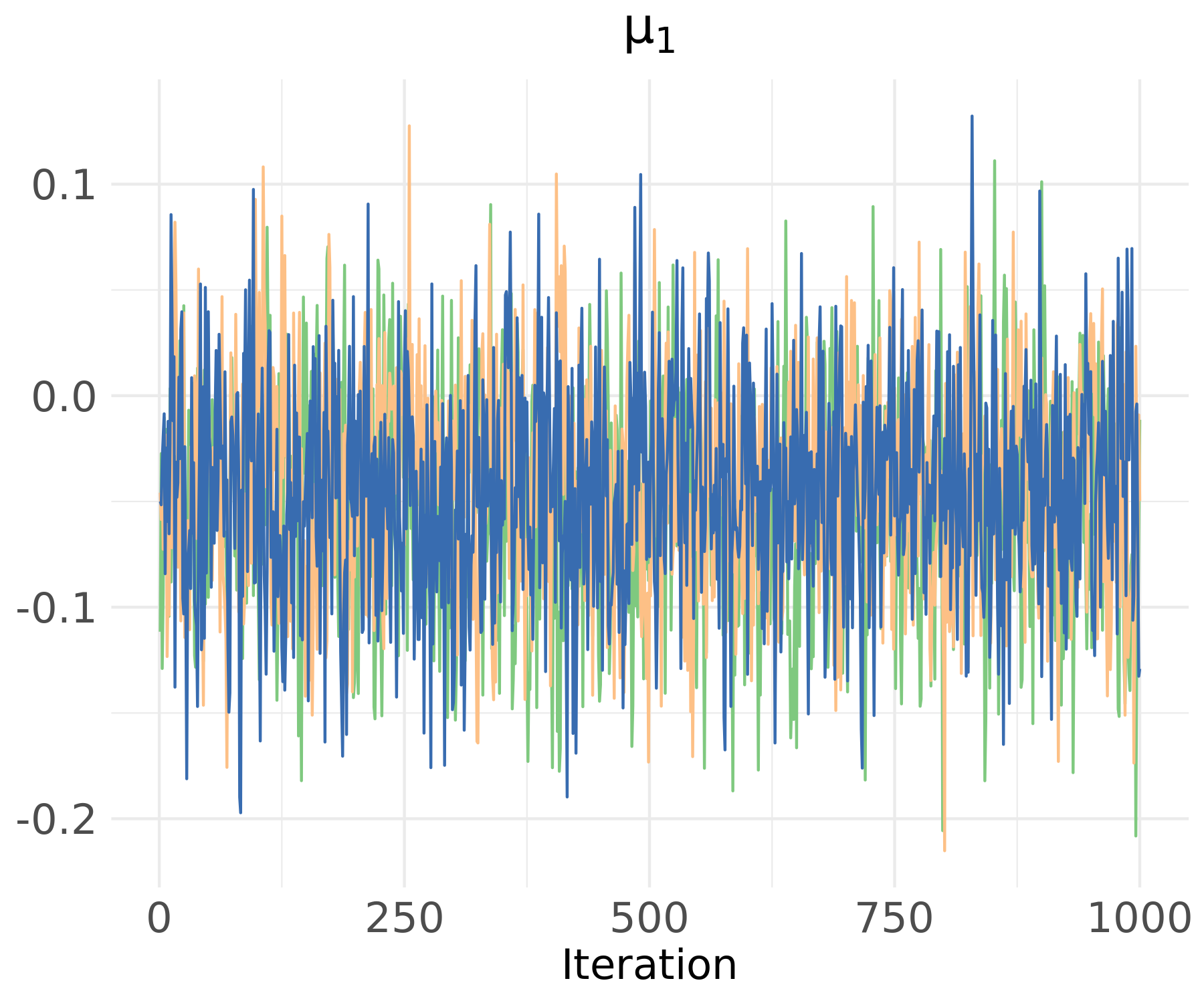}
\includegraphics[width=1.6in]{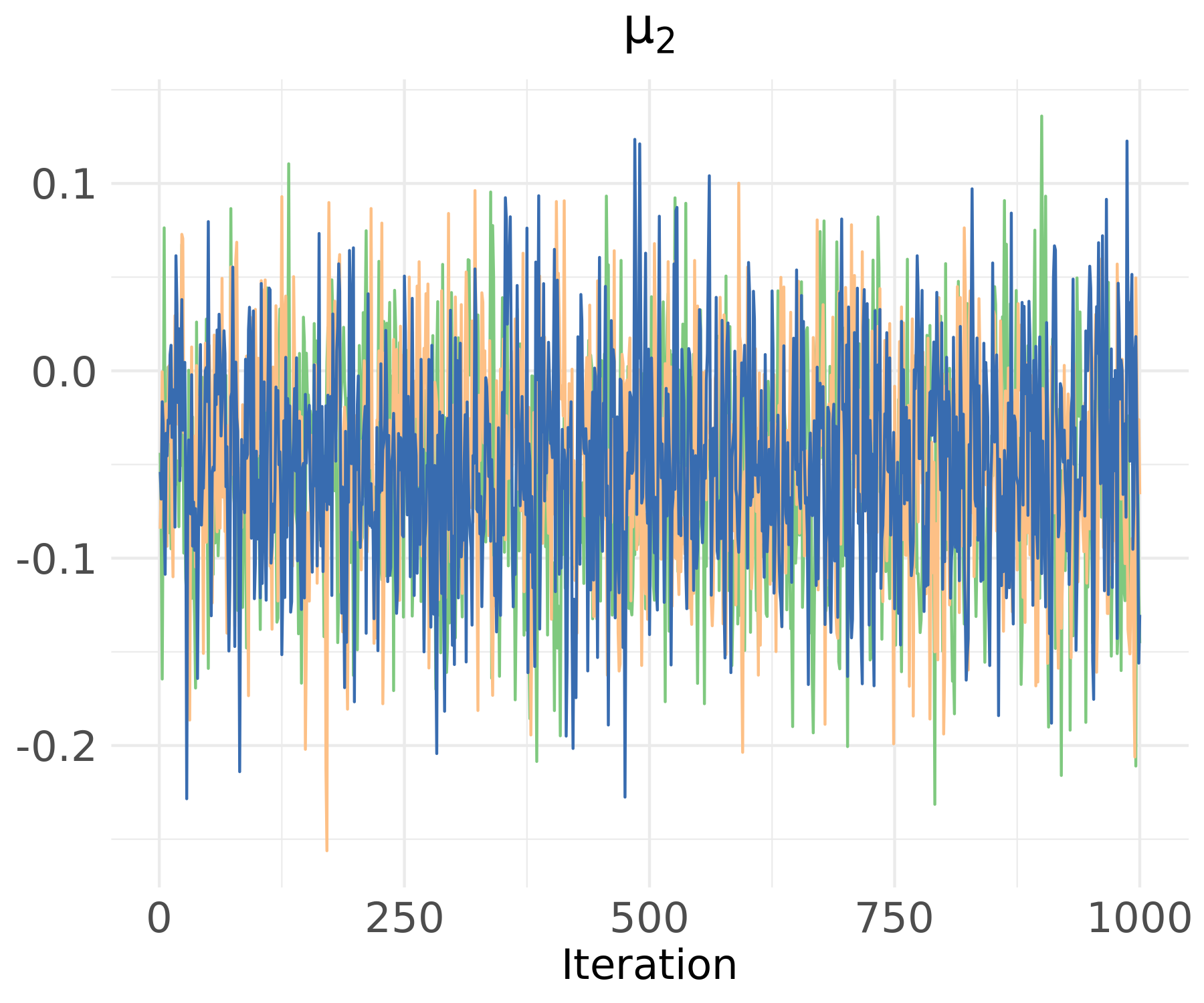}
\includegraphics[width=1.6in]{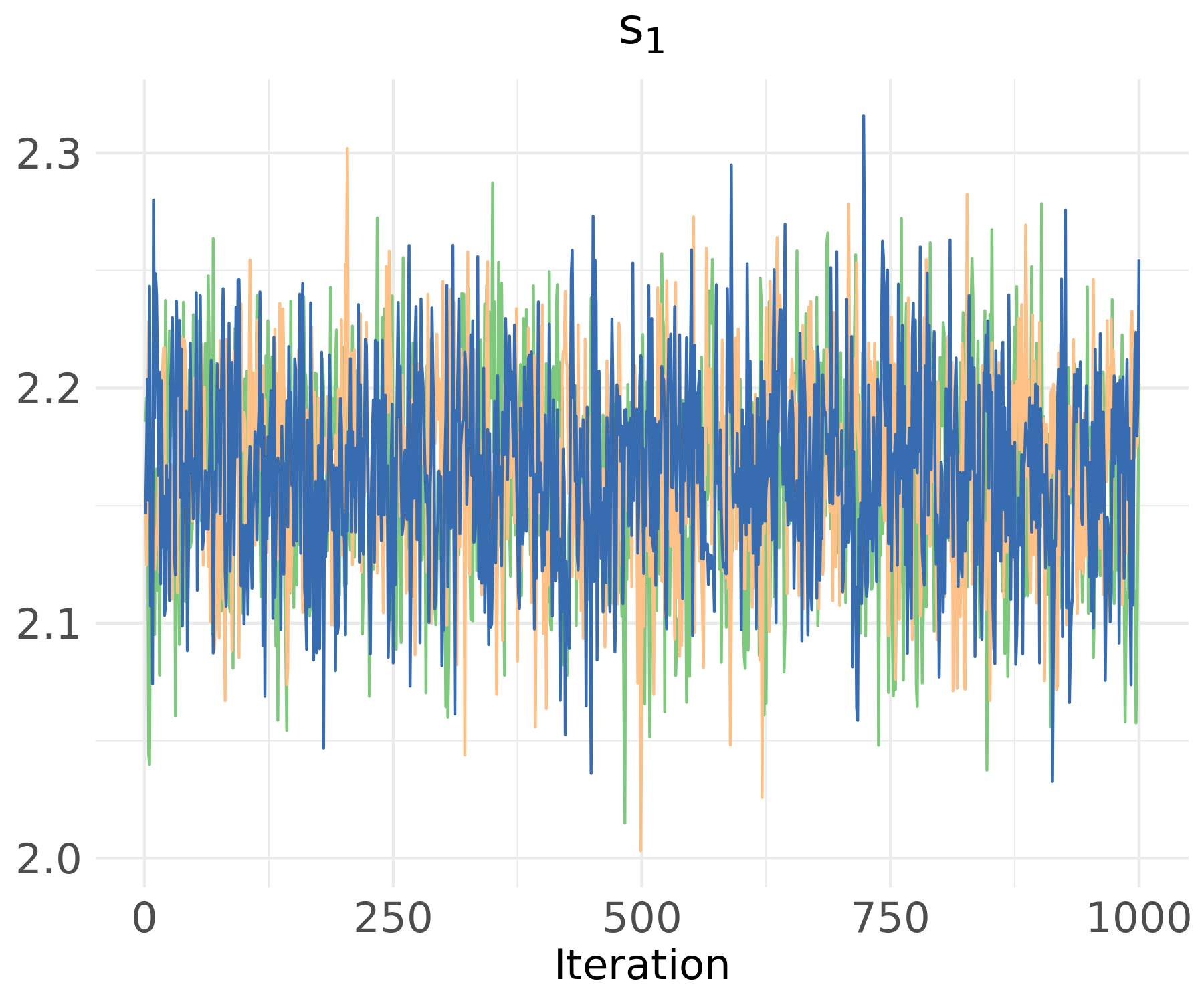}
\includegraphics[width=1.6in]{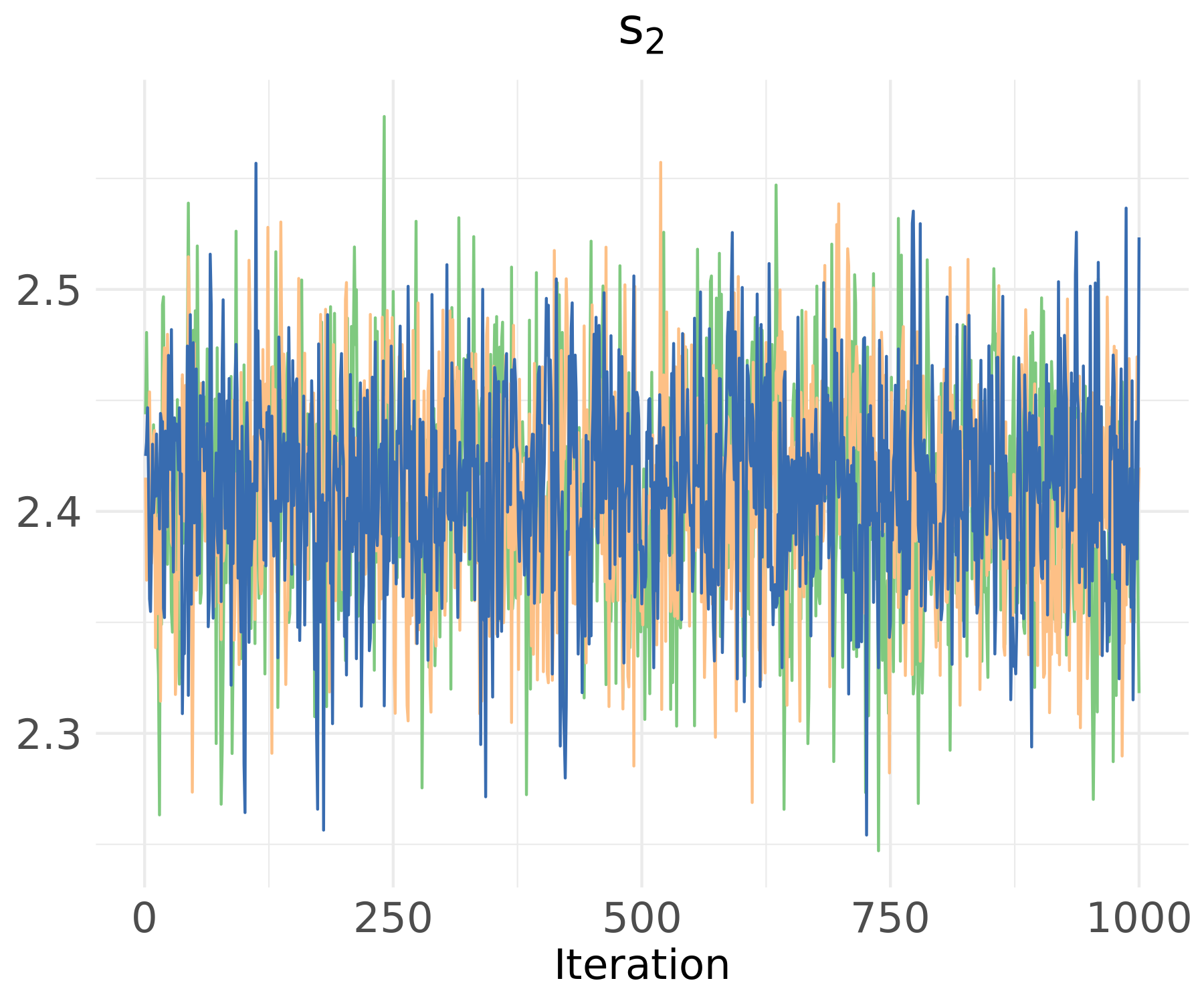}
\includegraphics[width=1.6in]{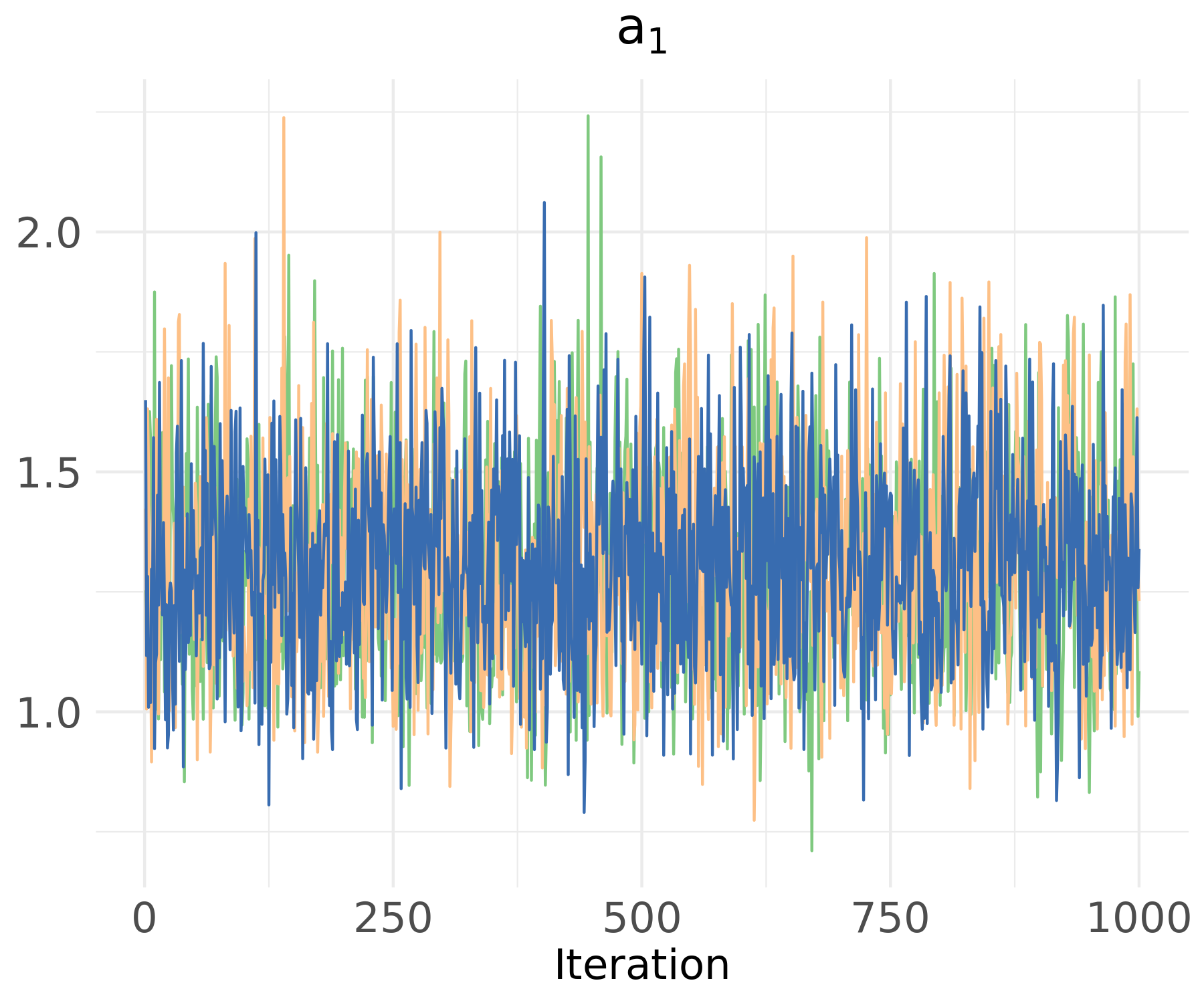}
\includegraphics[width=1.6in]{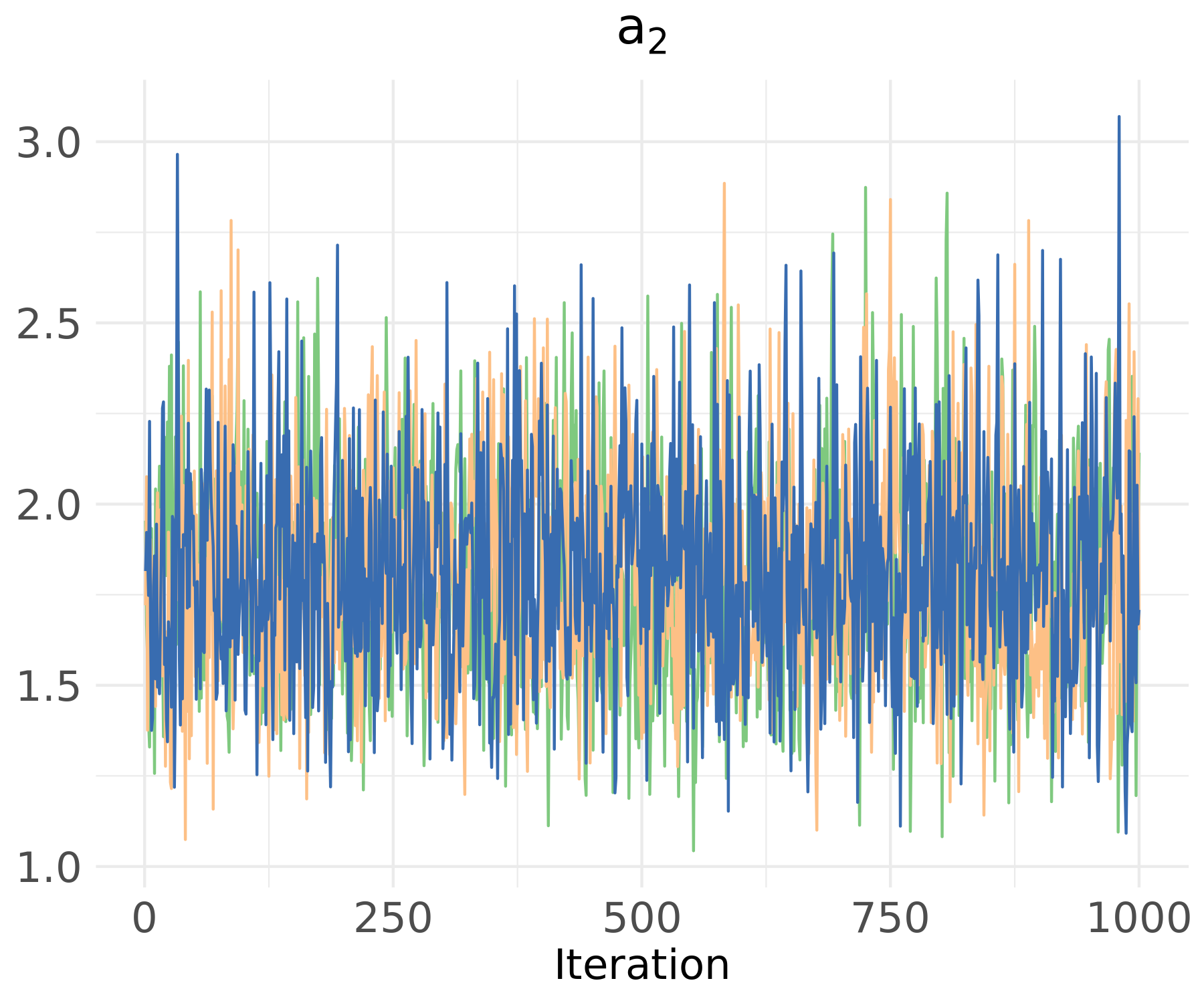}
\includegraphics[width=1.6in]{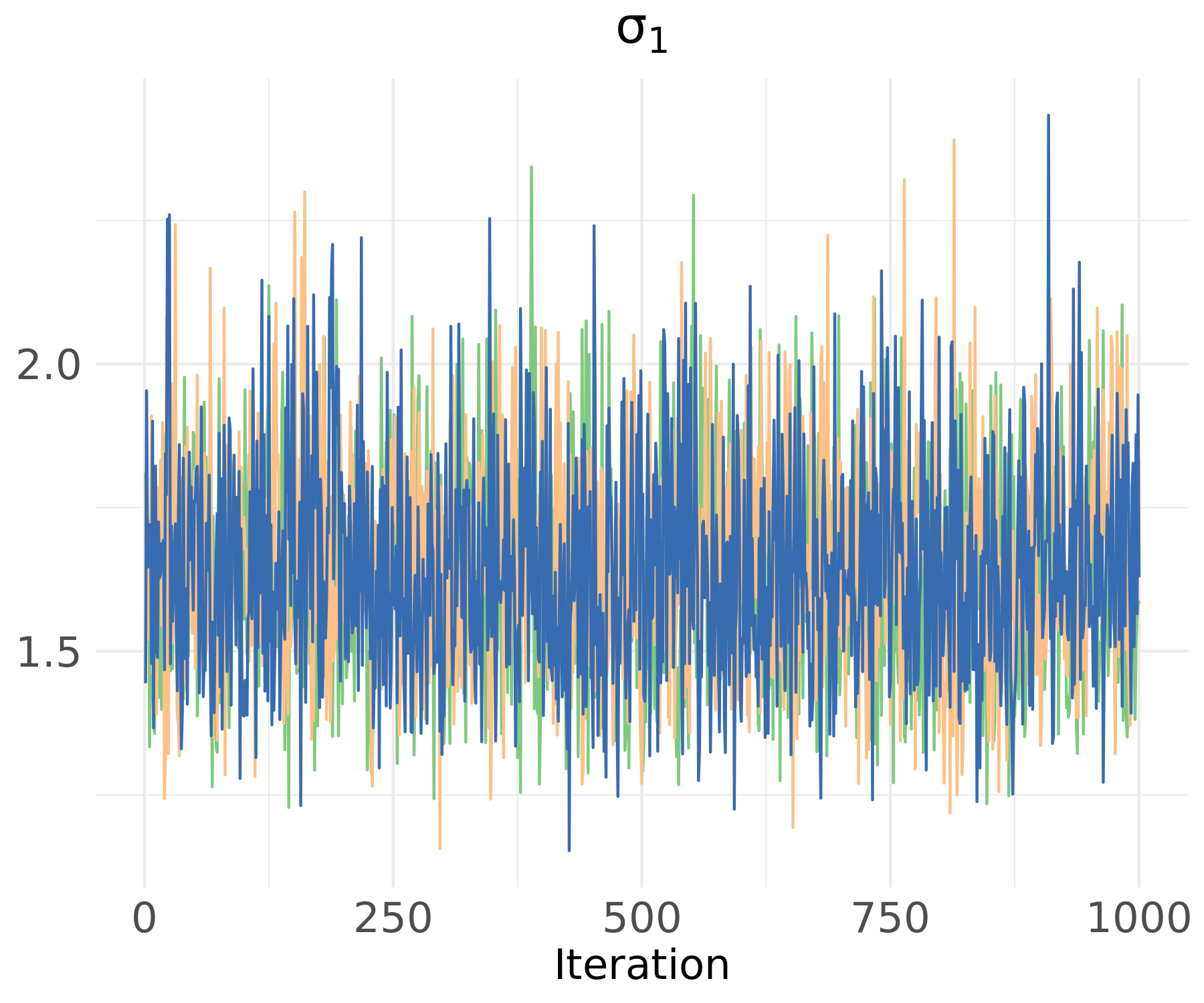}
\includegraphics[width=1.6in]{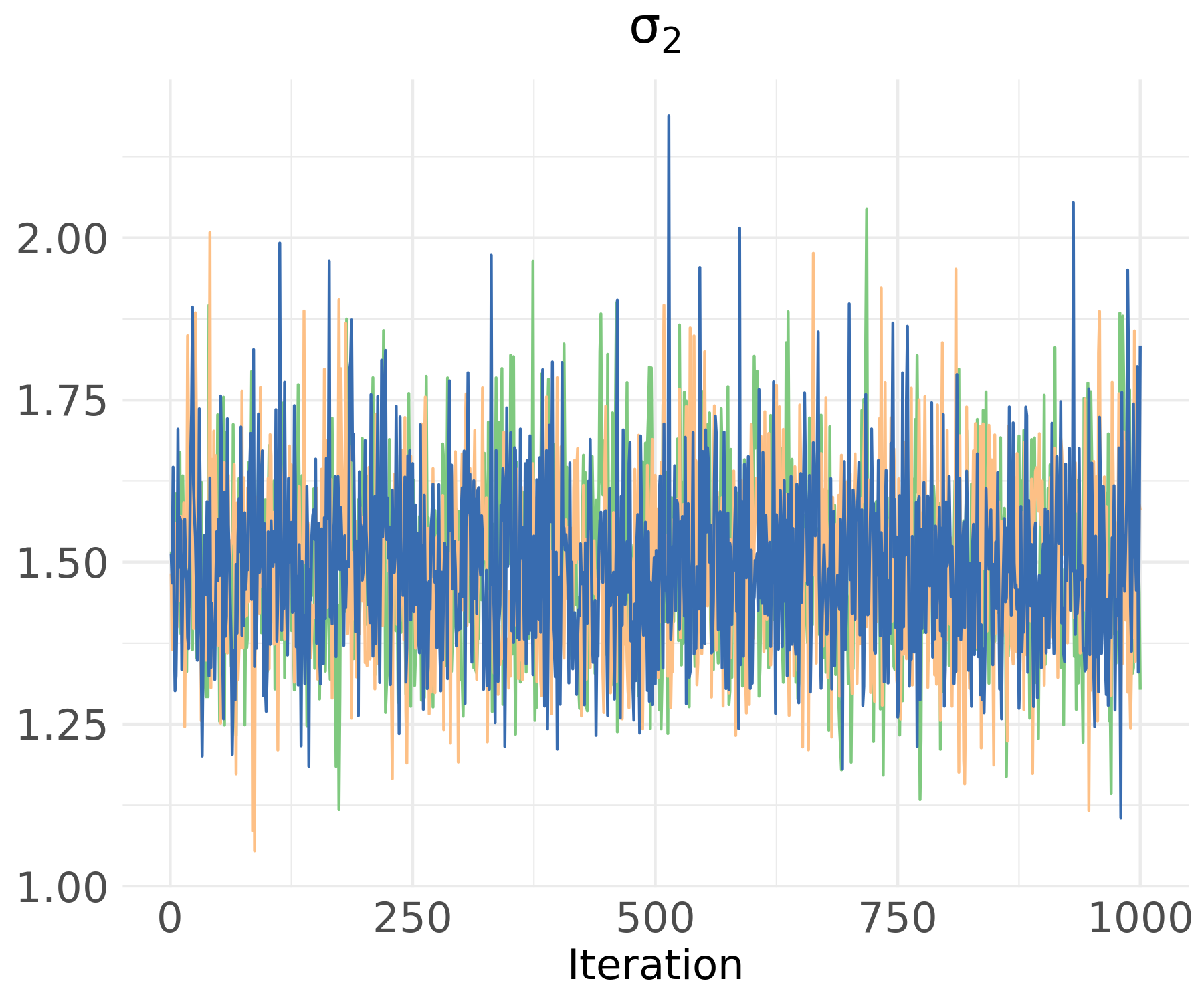}
\includegraphics[width=1.6in]{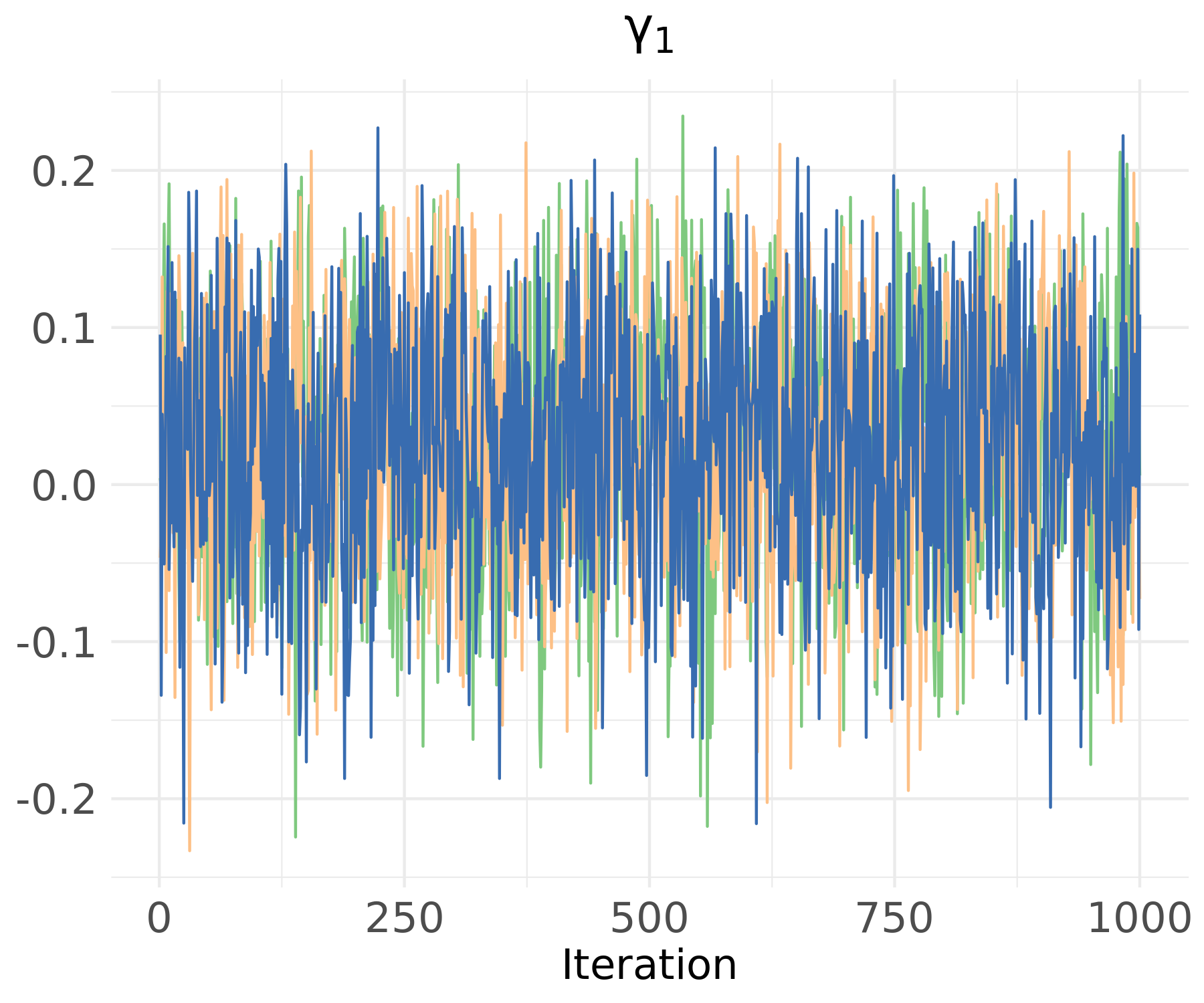}
\includegraphics[width=1.6in]{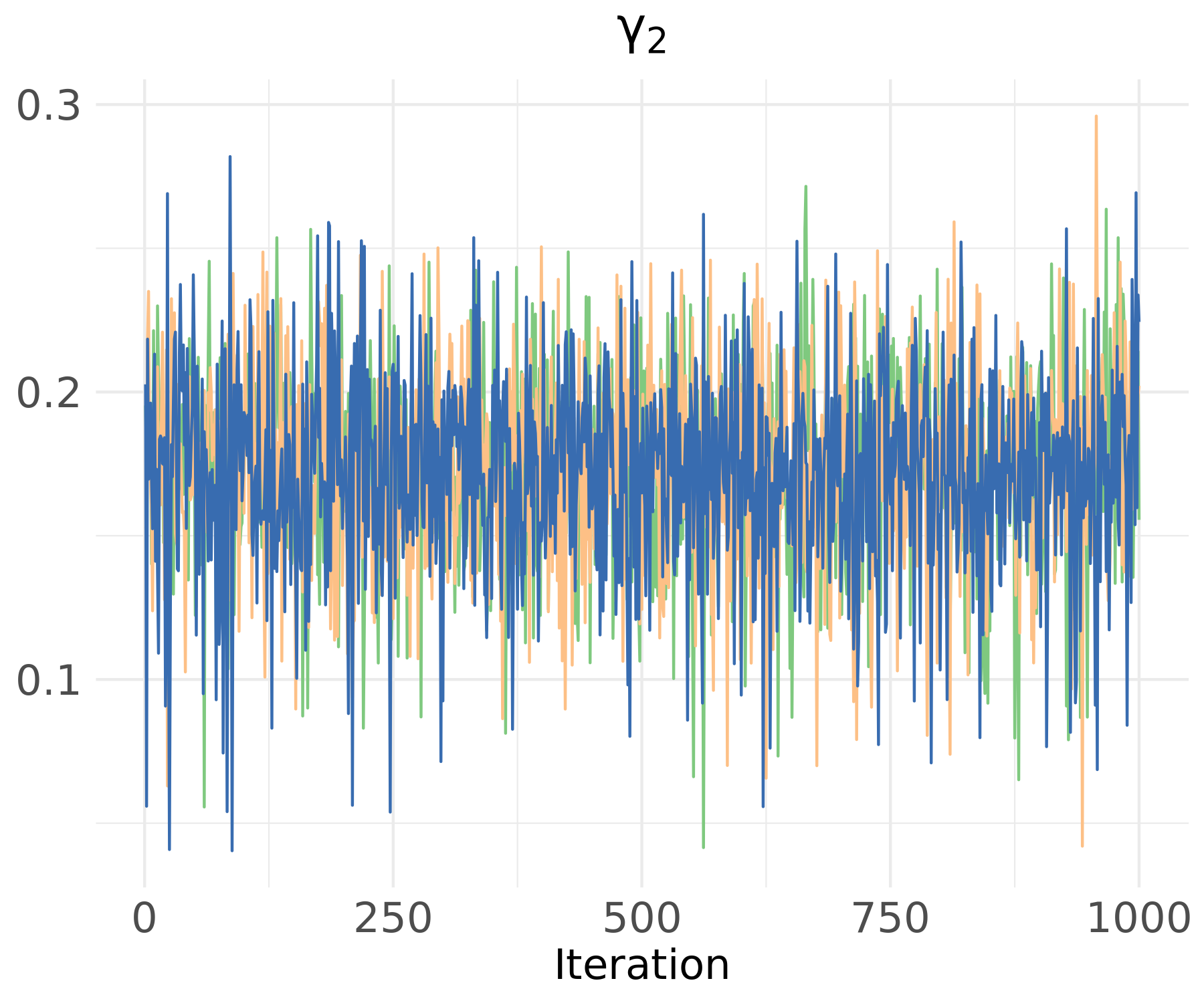}
\includegraphics[width=1.6in]{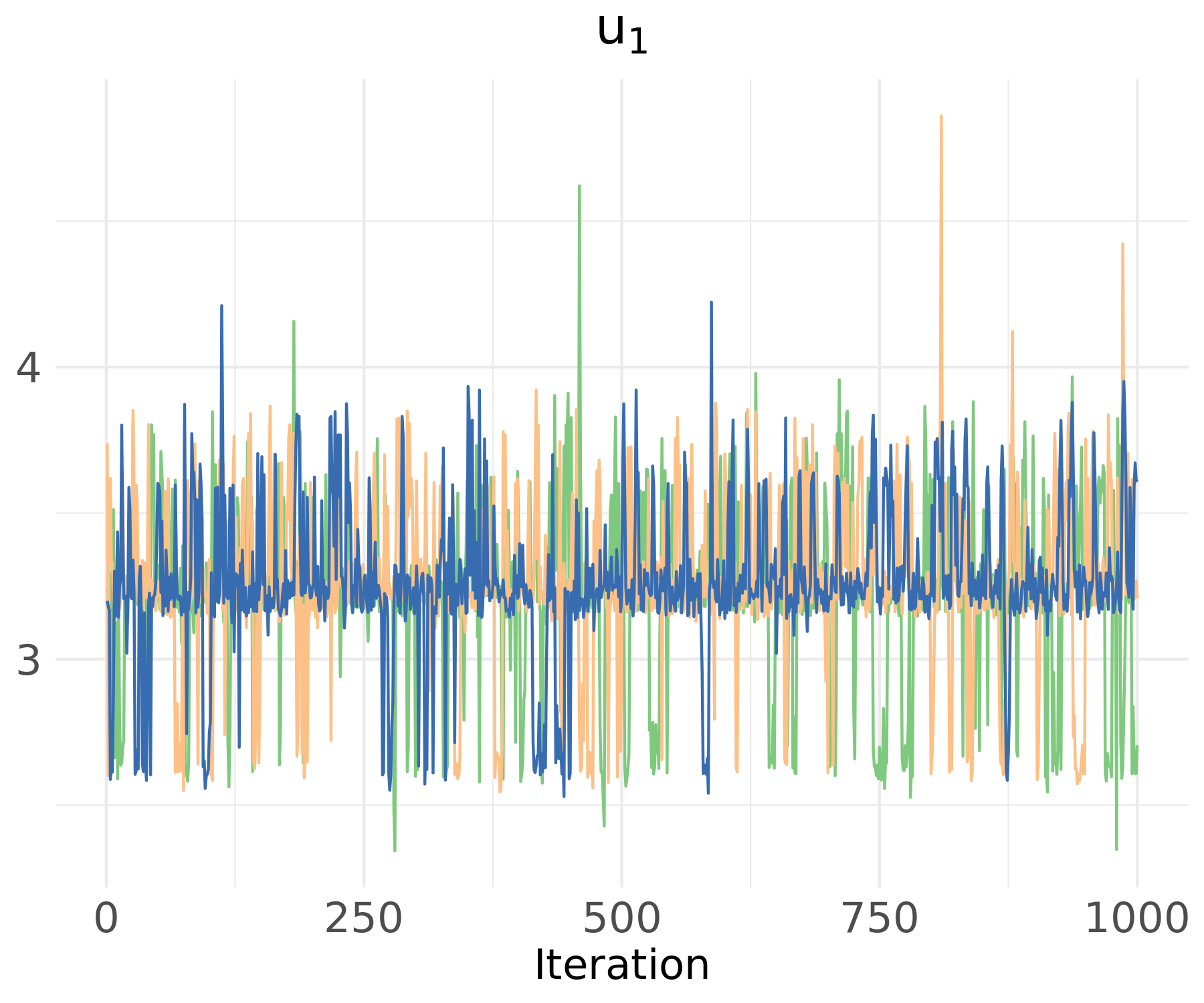}
\includegraphics[width=1.6in]{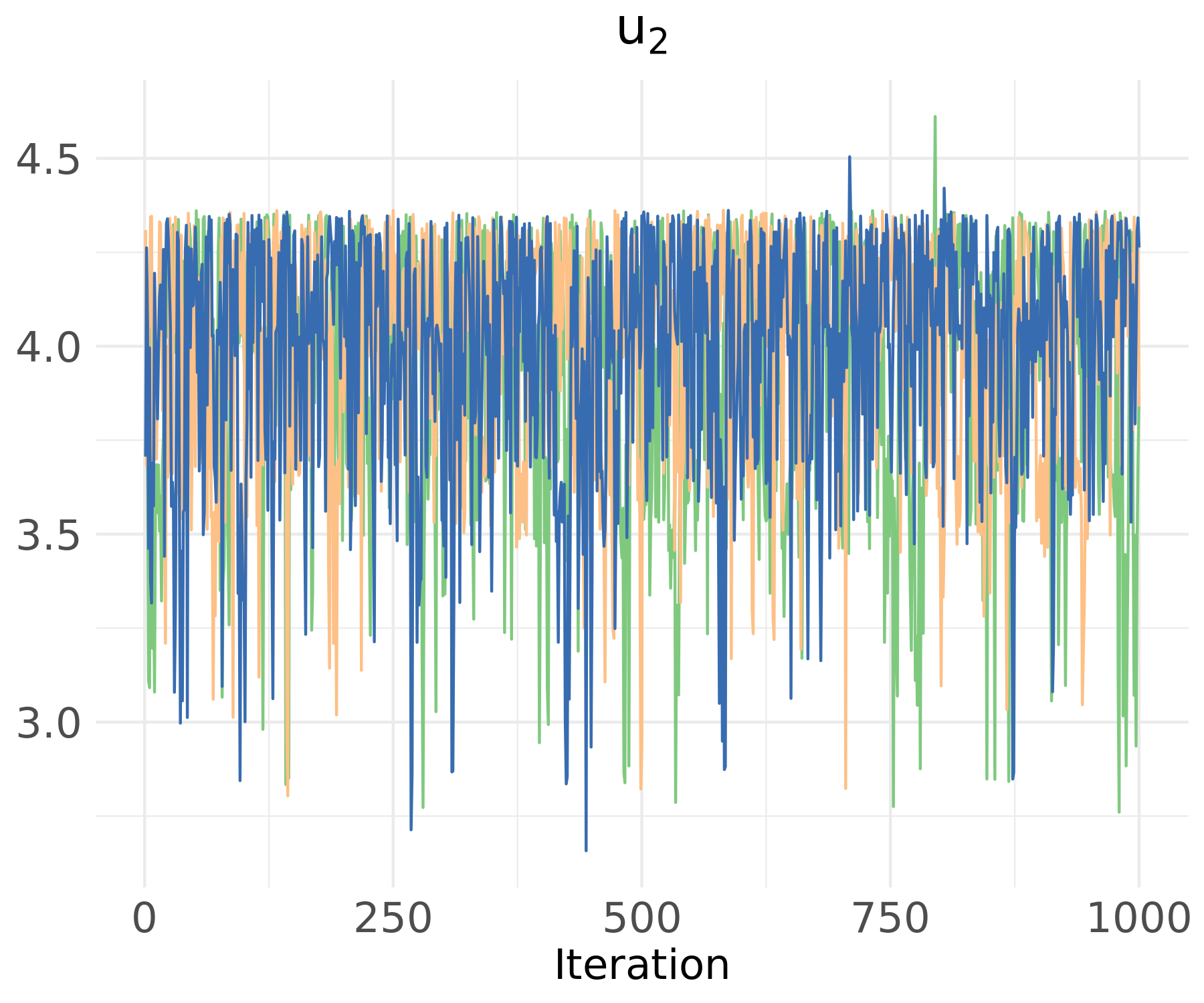}
\end{center}
\caption{\footnotesize{Trace plots for the MCMC results in the data application, with parameters ordered as in Figure ~\ref{fig: hist of posterior}.}
}
\label{fig:traceplot}
\end{figure}

\FloatBarrier
\bibliographystyle{apalike}
\bibliography{My_Library}
\end{document}